\documentclass[11pt, letterpaper]{article}

\usepackage{planted}
\usepackage{verbatim}

\def\authornotes{1pt}

\ifdim\authornotes=1pt
    \newcommand{\snote}[1]{\footnote{\color{red}Sidhanth: #1}}
    \newcommand{\jnote}[1]{\footnote{\color{orange}Jeff: #1}}
    \newcommand{\pnote}[1]{\footnote{\color{blue}Prasad: #1}}
    \newcommand{\Pnote}{\pnote}
    
\else
    \newcommand{\snote}[1]{}
    \newcommand{\jnote}[1]{}
    \newcommand{\pnote}[1]{}
    \newcommand{\Pnote}[1]{}
\fi

\newcommand{\Tr}[1]{\mathsf{Trace}\left(#1\right)}

\newcommand{\bbN}{\bbZ_{\ge 0}}

\newcommand{\Fourier}[1]{\chi_{#1}}

\newcommand{\Glyphs}{\mathrm{Glyphs}}
\newcommand{\glyphLe}[1][\calG]{\lessdot_{#1}}
\newcommand{\Glyph}{\calG}

\newcommand{\GraphPoly}[1]{\beta_{#1}}
\newcommand{\GraphMat}{\calQ}

\newcommand{\srib}{\calS}

\newcommand{\sparse}{\mathrm{sparse}}
\newcommand{\GOE}{\mathsf{GOE}}

\newcommand{\Parisi}{\mathsf{P}^*}

\newcommand{\maxcut}{\mathsf{MaxCut}}
\newcommand{\row}{{\color{ForestGreen} \texttt{row}}}
\newcommand{\column}{{\color{Fuchsia} \texttt{column}}}
\newcommand{\Bipartite}{\mathrm{Bipartite}}

\newcommand{\aspec}{\alpha_{\mathrm{spec}}}
\newcommand{\altwo}{\alpha_{\mathrm{row}}}
\newcommand{\amag}{\alpha_{\mathrm{mag}}}

\newcommand{\arow}{\altwo}
\newcommand{\HadDim}{\kappa}
\newcommand{\CholProd}[2]{\langle M[#1],M[#2]\rangle}
\newcommand{\SBV}{\textsf{SubspaceBooleanVector}}
\newcommand{\OPT}{\mathsf{OPT}}
\newcommand{\SOS}{\mathsf{SOS}}
\newcommand{\etal}{\text{et al.}}

\usepackage{tikz}
\begin{document}

\title{{Lifting Sum-of-Squares Lower Bounds: Degree-$2$ to Degree-$4$}}
\author{Sidhanth Mohanty\thanks{EECS Department, University of California Berkeley.  \texttt{sidhanthm@cs.berkeley.edu}.  Supported by NSF grant CCF-1718695.} \and Prasad Raghavendra\thanks{EECS Department, University of California Berkeley.  \texttt{prasad@cs.berkeley.edu}.  Supported by NSF grant CCF-1718695.} \and Jeff Xu\thanks{University of California Berkeley.  \texttt{jeffxusichao@berkeley.edu}}}
\date{\today}
\maketitle

\begin{abstract}
    The degree-$4$ Sum-of-Squares (SoS) SDP relaxation is a powerful algorithm that captures the best known polynomial time algorithms for a broad range of problems including MaxCut, Sparsest Cut, all MaxCSPs and tensor PCA.
    Despite being an explicit algorithm with relatively low computational complexity, the limits of degree-$4$ SoS SDP are not well understood.
    For example, existing integrality gaps do not rule out a $(2-\eps)$-algorithm for Vertex Cover or a $(0.878+\eps)$-algorithm for MaxCut via degree-$4$ SoS SDPs, each of which would refute the notorious Unique Games Conjecture.

    We exhibit an explicit mapping from solutions for degree-$2$ Sum-of-Squares SDP (Goemans-Williamson SDP) to solutions for the degree-$4$ Sum-of-Squares SDP relaxation on boolean variables.
    By virtue of this mapping, one can lift lower bounds for degree-$2$ SoS SDP relaxation to corresponding lower bounds for degree-$4$ SoS SDPs.  We use this approach to obtain degree-$4$ SoS SDP lower bounds for MaxCut on random $d$-regular graphs, Sherington-Kirkpatrick model from statistical physics and PSD Grothendieck problem. 
       
    Our constructions use the idea of pseudocalibration towards candidate SDP vectors, while it was previously only used to produce the candidate matrix which one would show is PSD using much technical work.  In addition, we develop a different technique to bound the spectral norms of {\it graphical} matrices that arise in the context of SoS SDPs.  The technique is much simpler and yields better bounds in many cases than the {\it trace method} -- which was the sole technique for this purpose.
\end{abstract}

\thispagestyle{empty}
\setcounter{page}{0}
\newpage

\tableofcontents
\thispagestyle{empty}
\setcounter{page}{0}
\newpage

\section{Introduction}   \label{sec:intro}

Sum-of-Squares (SoS) semidefinite programming hierarchy is one of the most powerful frameworks for algorithm design.
Its foundations lie in the so-called ``Positivestellensatz" whose history dates back to more than a century to the work of Hilbert and others.
The algorithmic insight of finding sum-of-squares proofs via the technique of semi-definite programming was only codified at the turn of the century by Parrillo \cite{parrilo2000structured} and Lasserre \cite{lasserre2001global} (also see \cite{Shor1987}).

Given a system of polynomial equations/inequalities $\cP$, the SoS SDP hierarchy yields a sequence of semi-definite programming relaxations to reason about the feasibility of $\cP$.  The $d$-th relaxation in the sequence referred to as the {\it degree-$d$ SoS SDP relaxation}.  Successive relaxations get increasingly accurate in reasoning about $\cP$ at the expense of computational complexity that grows exponentially with the degree.

SoS SDP hierarchy is an incredibly powerful algorithmic technique.  The best known approximation algorithms for a variety of combinatorial optimization problems including Maximum Cut, all Max-CSPs and Sparsest Cut are all subsumed by the first two levels (degree-$4$) of the hierarchy.  
More recently, there has been a flurry of work that uses SoS SDP hierarchy on problems in unsupervised learning such as dictionary learning, estimating parameters of mixtures of Gaussians, tensor PCA and linear regression.

The limits of SoS SDP hierarchy remain largely a mystery even at degree four.  The degree four SoS SDP relaxation could possibly yield a $(2-\eps)$-approximation for Minimum Vertex Cover or a $(0.878+\eps)$-approximation for Maximum Cut and thereby refute the notorious Unique Games Conjecture.  Despite the immense consequences, the integrality gap of degree-$4$ SoS SDP relaxations of Maximum Cut and Vertex Cover remain unresolved.

Understanding the precise limits of SoS SDP hierarchy has compelling implications even in the context of  average case problems.
Specifically, the SoS SDP hierarchy can be serve as a {\it lens} to understand the terrain of average case complexity. 
For example, consider the problem of refuting a random $3$-SAT formula. Here the input consists of a random $3$-SAT formula $\Phi$ with $m = pn$ clauses chosen uniformly at random on $n$ variables.  For all densities $p$ that are larger than some fixed constant, the formula $\Phi$ is unsatisfiable with high probability.
The goal of refutation algorithm is to certify that $\Phi$ is unsatisfiable.  Formally, a refutation algorithm outputs $1$ only on instances that are unsatisfiable and it does so on a non-negligible fraction of random $3$-SAT formulae.
Although the computational complexity of refuting random $3$-SAT formulae conceivably varies with the density $p$ of clauses, it seems difficult to glean this structure using {\it reductions} -- the central tool in worst-case computational complexity.  
In particular, it is quite difficult to devise reductions that produce random instances from simple probability distributions such as random $3$-SAT, though this has been sometimes achieved \cite{berthet2013complexity, DBLP:conf/colt/BrennanBH18}.
In such a setting, the {\it smallest degree of SoS SDP hierarchy} that can solve the refutation problem (henceforth referred to as just ``SoS degree")  can serve as a proxy for computational complexity.
While SoS SDP hierarchy doesn't capture all efficient  algorithms in every context, it unifies and subsumes many of the state-of-the-art algorithms for basic combinatorial optimization problems.
%

%
This paradigm has been fruitful for random $3$-SAT. Nearly matching upper and lower bounds on SoS degree of refutation \cite{GRIG01-a, Sch08, RaghavendraRS17} have been established, thereby painting a precise picture of how  the complexity of the problem changes with density of clauses.  Specifically, for all $ \omega(1) < p < n^{3/2}$, the sum-of-squares degree is $\tilde{\Theta}(n/p^2)$, yielding a complexity of $2^{\tilde{\Theta}(n/p^2)}$.

There is a rich landscape of average case problems with many having sharper computational thresholds than random $3$-SAT.  For example, the random regular NAESAT promises to exhibit an abrupt change in computational complexity as soon as the degree exceeds $13.5$ \cite{deshpande2019threshold}.
Chromatic number of random $d$-regular graphs and community detection on stochastic block models are  two other prominent examples with very sharp but conjectural computational thresholds.  Much is known about structural characterestics and phase transitions in the solution space as one varies the underlying parameters in these models.  Heuristically, certain phase transitions in the solution space are conjectured to be associated with  abrupt changes in the computational complexity.  
The sum-of-squares SDP can be harnessed towards quantitatively demonstrating these phenomenon.

\subsection{Our Results}
Our main result is an explicit mapping from solutions to degree-$2$ SoS SDP to solutions to degree-$4$ SoS SDP for boolean optimization.
To formally state the theorem, let us begin by setting up some notation.

First, the degree-$d$ SoS SDP relaxation can be succinctly described in terms of {\it pseudo-distributions}.  Intuitively, a {\it pseudo-distribution} corresponds to a function that looks like an actual distribution over solutions, to low-degree polynomial squares.  The definition is succinct and simple enough that we reproduce the formal definition here.
\begin{definition}
Fix a natural number $d \in \mathbb{N}$.  A {\it degree $d$ pseudo-distribution} $\mu$ is a function $\mu: \{-1,1\}^n \to \R$ satisfying
\begin{enumerate}
\item(Normalization) $$ \E_{x\in \{-1,1\}^n}[\mu(x)] = 1$$ 
\item(Positivity on degree $d$ squares) $$ \E_{x\in \{-1,1\}^n} [p^2(x)\cdot \mu(x)] \geq 0 \qquad \text{ for all } p \in \R[x_1,\ldots,x_n], \deg(p) \leq d/2$$
\end{enumerate}
\end{definition}
The degree-$d$ SoS SDP relaxation for maximizing a quadratic function $A(x) = x^{\dagger} A x$ can be written succinctly as:

{\center  $\textsf{SoS}_d$ Relaxation:  \text{ Maximize over degree $d$ pseudo-distributions $\mu$, } $\E_{x}[\mu(x) \cdot A(x)]$}

While the above description of degree-$d$ SoS SDP is accurate, we will now describe the associated semidefinite programs for degree two and four in detail. 
By the degree-$2$ SoS SDP for boolean optimization, we refer to the Goemans-Williamson SDP relaxation, first introduced in the context of the MaxCut problem.
Specifically, a feasible solution to the degree-$2$ SoS SDP solution is given by a p.s.d matrix $X \succeq 0$ whose diagonal entries are identically $1$.  Formally, the set of degree-$2$ SoS SDP solutions denoted by $\mathsf{SoS}_2$ is given by,
$$ \mathsf{SoS}_2 = \{ X \in \R^{n \times n} | X \succeq 0 \text{ and } X_{ii} = 1 \text{ for all } i \in [n]\} $$

\newcommand{\cM}{\mathcal{M}}
The solution to a degree-$4$ SoS SDP for boolean optimization consists of a matrix $\cM$ of dimension $\binom{n}{\leq 2} = 1 + \binom{n}{1} + \binom{n}{2}$.  The matrix $\cM$ is indexed by subsets of $[n] = \{1,\ldots,n\}$ of size at most $2$.  The set $\mathsf{SoS}_4$ is specified by the following SDP:
\begin{align}
    \cM[S,T] & = \cM[S',T'] & \text{for all }S,T,S',T' \in \binom{[n]}{\leq 2} \text{ such that }S\Delta T = S' \Delta T' \\
    \cM[\emptyset,\emptyset] & = 1\\
    \cM & \succeq 0
\end{align}
The above semidefinite programs are equivalent to the definition of SoS relaxations in terms of pseudo-distributions.  Specifically, 
the entries of the matrix $\cM$ are {\it pseudomoments} upto degree four of the pseudo-distribution $\mu$.  Formally, the entry $\cM[S,T]$ corresponds to the following moment:
$$ \cM[S,T] = \E_{x \in \{-1,1\}^n} \left[ \prod_{i \in S} x_i \prod_{j \in T} x_j\right]$$

We are now ready to state the main theorem of this work.
\begin{theorem}[Main theorem]   \label{thm:main-lifting}
    There is an explicit map $\Phi:\SOS_2\to\SOS_4$ such that $\Phi(X)[i,j]$ \footnote{We are using $\Phi(X)[i,j]$ to denote $\Phi(X)[\{i\},\{j\}]$.} is given by
    \begin{align}   \label{eq:deg-4-ij}
        \Phi(X)[i,j] = \frac{X_{ij}+X_{ij}^3}{1+C\amag\cdot(1+\arow^4)\cdot(1+\aspec^2)} 
    \end{align}
        where $\amag,\arow$ and $\aspec$ are the maximum off-diagonal entry, maximum row norm and spectral norm respectively of the degree two SDP solution $X$, and $C$ is an absolute constant.
    Moreover for every pair of subsets $S,T \in \binom{[n]}{\leq 2}$, $\Phi(X)[S,T]$ is an explicit function of $\{X_{ij} | i,j \in S\cup T\}$.
    \end{theorem}
\noindent All the entries of $\Phi(X)$ are explicit constant degree polynomials in $X$.  We refer the reader to \pref{sec:construction} for the definition of   $\Phi$ and the proof of \pref{thm:main-lifting}.  
Let us suppose we have an objective value given by $\langle A, X \rangle = \sum_{i,j} A_{ij} X_{ij}$ for a Hermitian matrix $A$.  The  corresponding objective value of degree-$4$ SoS SDP is given by $\langle A, \cM \rangle = \sum_{i,j} A_{ij} \cM[i,j]$.  We show the following bound on change in objective value (see \pref{lem:sos-lift-obj} in \pref{sec:construction}):
\begin{theorem} \label{thm:main-lift-obj-val}
    Let $\alpha\coloneqq C\amag\cdot(1+\arow^4)\cdot(1+\aspec^2)$ where $\amag,\arow$ and $\aspec$ are as defined in \pref{thm:main-lifting}, then for any Hermitian matrix $A\in\R^{n\times n}$,
    \[
        \langle{A, \Phi(X)} \rangle \ge \frac{1}{1+\alpha} \langle A,X \rangle - \frac{\alpha}{1+\alpha} \cdot \left(\sqrt{n} \|A\|_{F} - \Tr{A}\right) 
    \]
\end{theorem}
The existence of a non-trivial and useful mapping from degree-$2$ SoS SDP solutions to degree-$4$ SoS SDP solutions comes as a surprise to the authors.
Consider the following immediate consequence of such a mapping. 
Given the degree-$2$ SoS SDP on an instance of MaxCut,   the above theorem yields an easily computable lower bound on the degree-$4$ SoS SDP value on the same instance.  
For example, this yields an efficiently verifiable sufficient condition (checkable in time $O(n^2)$) under which the degree-$4$ SoS SDP yields no better bound than the degree-$2$ SoS.

We use the lifting theorem to recover lower bounds for degree-$4$ SoS SDP relaxations for a few average case problems -- which was the original motivation behind this work.  The problems and the corresponding lower bounds are described below.

\paragraph{Sherrington--Kirkpatrick Model.}  Let $\bW$ be a random $n\times n$ matrix with independent Gaussian entries, let $\bG\coloneqq \frac{1}{\sqrt{2}}\left(\bW+\bW^{\dagger}\right)$;  we say that $\bG$ is sampled from $\GOE(n)$, a distribution known as the Gaussian Orthogonal Ensemble.   
A fundamental model in the study of spin glasses from statistical physics is the \emph{Sherrington--Kirkpatrick (SK) model} where the energy of a system of $n$ particles in a state $x \in \{-1,+1\}^n$ states is given by $-x^{\dagger} \bG x$. 
The Sherrington-Kirkpatrick (SK) model has been extensively studied in various areas including the study of spin glasses, random satisfiability problems, and learning theory \cite{Engel:2001:SML:558792, mezard1987spin, nishimori01, Mezard812, Mezard:2009:IPC:1592967, Mon18}. 

For the SK model, a quantity of particular interest is the minimum possible energy, i.e.,
$$ \OPT(\bG) = \max_{x \in \{-1,1\}^n} x^{\dagger} \bG x \ .$$
In a highly influential work, Parisi predicted in \cite{Par79, Parisi1980ASO} that $\OPT(\bG)$ concentrates around $2\cdot\Parisi n^{3/2}$, where $\Parisi$ is an explicit constant now referred to as the Parisi constant.  The value of $\Parisi$ is roughly $0.763166$.  This prediction was eventually rigorously proven twenty five years later in a celebrated work of Talagrand \cite{Tal06}, thereby confirming that $\OPT(\bG) \approx (1.52633\dots) \cdot n^{3/2}$.  

This brings us to our natural average case refutation problem, that of certifying an upper bound on $x^{\dagger} \bG x$ for $x \in \{-1,1\}^n$.
A natural refutation algorithm is the \emph{spectral refutation}.  Indeed
\[
    \OPT(\bG) = \max_{x\in\{\pm 1\}^n}x^{\dagger}\bG x \le n\cdot\max_{\|x\|=1}x^{\dagger} \bG x = n\cdot\lambda_{\max}(\bG),
\]
the algorithm which outputs $\lambda_{\max}(\bG)$ given $\bG$ as input is an efficient refutation algorithm.  Since $\lambda_{\max}(\bG)$ concentrates around $2\sqrt{n}$, it certifies an upper bound $OPT(\cG) \leq 2 n^{3/2}$ which is larger than the true value of the optimum $OPT(\cG) = 2\Parisi \cdot n^{3/2} = 1.52 \cdot n^{3/2}$. 

This raises the question whether efficient algorithms can certify an upper bound stronger than the simple spectral bound?  In this work, we show that the degree-$4$ SoS SDP fails to certify a bound better than the spectral bound.  To this end, we start with a feasible solution to the degree-$2$ SoS SDP relaxation for the SK model and apply our lifting theorem \pref{thm:main-lifting} to construct a degree-$4$ SoS SDP solution.

%

%
\begin{theorem}[Degree-$4$ SoS lower bound for Sherrington--Kirkpatrick] \label{thm:sk-main}
    Let $\bG\sim\GOE(n)$. With probability $1-o_n(1)$, there exists a degree-$4$ SoS SDP solution with value at least $(2 - o_n(1)) \cdot n^{3/2}$
\end{theorem}
In an independent and concurrent work, Kunisky and Bandeira \cite{KB19} also obtained a degree-$4$ SoS integrality gap for the Sherrington--Kirkpatrick refutation problem. 

\paragraph{$\maxcut$ in random $d$-regular graphs.}  
Akin to the Sherrington--Kirkpatrick model, it is known from the work of Dembo et al. \cite{DMS17} that the fraction of edges cut by the max-cut in a random $d$-regular graph $\bG$ on $n$ vertices is concentrated around
\[
    \frac{1}{2} + \frac{\Parisi}{\sqrt{d}} + o_d\left(\frac{1}{\sqrt{d}}\right) + o_n(1).
\]
On the other hand, it was proved in \cite{Fri08,Bor19} that the spectral refutation algorithm, which outputs the maximum eigenvalue of $\frac{L_{\bG}}{4m}$, certifies an upper bound of
\[
    \frac{1}{2} + \frac{\sqrt{d-1}}{d}+o_n(1).
\]
Once again the question remains whether more sophisticated refutation algorithms can beat the spectral bound. Through our lifting theorem, we show that degree $4$ SoS SDP is no better than spectral algorithm asymptotically as $d \to \infty$\footnote{
    We believe that \pref{thm:max-cut-main} is not tight and conjecture that there should exist pseudoexpectations with objective value $\frac{1}{2} + (1-o_n(1))\frac{\sqrt{d-1}}{d}$ for all values of $d$.
}.

\begin{theorem}[Degree-$4$ SoS lower bound for $\maxcut$ in random $d$-regular graphs]  \label{thm:max-cut-main}
    Let $\bG$ be a random $d$-regular graph.  For every constant $\eps > 0$ with probability $1-o_n(1)$, there is a degree-$4$ SoS SDP solution with MaxCut value at least \[
     \frac{1}{2} + \frac{\sqrt{d-1}}{d} \left(1-\eps - \frac{\gamma(\eps)}{d^{1/2}}\right)
    \]
    for some constant $\gamma$ that depends only on $\eps$.
\end{theorem}

The degree-$2$ SoS SDP solution for the SK model on which we apply our lifting theorem is presented in \pref{thm:SK-deg-2}.  Analogously,  \pref{thm:deg-2-seed-maxcut} describes the degree $2$ SoS SDP solution we use for the MaxCut problem.

 \paragraph{``Boolean Vector in Random Subspace'' Problem.}
 The refutation problem for the SK model is closely tied to the following problem: given a random subspace $\bV$ of dimension $d$ in $\R^n$, can we certify that there is no hypercube vector $\{\pm 1\}^n$ `close' to $\bV$ in polynomial-time?  
 Formally, if $\Pi_{\bV}$ denotes the projection operator onto a random subspace, then let $\OPT(\bV)$ denote the maximum correlation of a boolean vector with $\bV$, i.e.,
 $$ \OPT(\bV) = \frac{1}{n} \max_{x \in \{-1,1\}^n} x^{\dagger} \Pi_{\bV} x \ .$$
Using a simple $\eps$-net argument, one can show that with high probability $\OPT(\bV) \sim \frac{2}{\pi} + \gamma(d/n)$ for some function $\gamma: [0,1] \to \R^+$ such that $\lim_{\eps \to 0} \gamma(\eps) = 0$\footnote{$\OPT(\bV)=\|A_{\bV}\|_{2\to 1}^2$, where columns of $A_{\bV}$ are an orthogonal basis for $\bV$.  So for fixed unit $x\in\R^d$, $\|A_{\bV}x\|_1$ concentrates around $\sqrt{2/\pi}$ with a subgaussian tail.  A union bound over an $\eps$-net of $\R^d$ completes the calculation.}.
In other words, for a low dimensional subspace with $d \ll n$, $\OPT(\bV)$ is close to $2/\pi$ with high probability over choice of $\bV$.

The spectral algorithm can only certify $\OPT(\bV) \leq \| \Pi_{\bV} \| = 1$ which is a trivial bound.  A natural question is whether one can efficiently certify a stronger upper bound.  We show that the degree-$4$ SoS SDP fails to improve on the spectral bound by a non-negligible amount.

\begin{theorem}[Boolean Vector in Random Subspace] \label{thm:bool-vector-subspace}
    If $\bV$ is a random $d$-dimensional subspace where $d\ge n^{.99}$, then with probability $1- o_n(1)$ there exists a degree-$4$ SoS SDP solution with value at least $1-o_n(1)$. 
\end{theorem}


\subsection{Related Work}
Early work on lower bounds for sum-of-squares SDPs arose out of the literature on proof complexity.  In particular, these included lower bounds on sum-of-squares refutations of Knapsack \cite{grigoriev2001complexity}, Parity principle (non-existence of a perfect matching in a complete graph on odd number of vertices) \cite{GRIG01-a} and $\textsf{3XOR}/\textsf{3SAT}$ \cite{GRIG01-a}.  
For $\textsf{3SAT}/\textsf{3XOR}$, it was proven by Grigoriev \cite{GRIG01-a} and later independently by Schoenbeck \cite{Sch08} that the polynomial time regime of Sum-of-Squares fails to refute random instances whenever the density of clauses is $o(\sqrt{n})$.  
This lower bound for $\textsf{3SAT}$ is the starting point of lower bounds for a host of other problems.  Specifically, the use of polynomial time reductions to convert integrality gaps for one problem into another, first pioneered in \cite{khot2015unique}, was shown to be applicable to the SoS SDP hierarchy \cite{Tul09}.  By harnessing the known reductions, Tulsiani \cite{Tul09} recovers exponential lower bounds for a variety of constraint satisfaction problems (CSP) starting from $\textsf{3SAT}$.

More recently, Kothari \etal~\cite{KMOW17} obtained lower bounds for all CSPs corresponding to predicates whose satisfying assignments support a pairwise independent distribution.
This class of CSPs is well beyond the reach of current web of NP-hardness reductions.  $2$-CSPs such as MaxCut are not pairwise independent, and are thus not within the realm of known lower bounds for SoS SDPs.


The problem of certifying the size of maximum clique on Erdos-Renyi random graphs (closely related to the planted clique problem) has received much attention lately.
Following a series of works \cite{DM15, HopkinsKPRS18} that obtained the tight lower bounds for degree four, the breakthrough tour-de-force of Barak \etal~\cite{BHKKMP19} obtained lower bounds for upto degree $O(\log n)$.
In this work, Barak \etal~\cite{BHKKMP19} introduced a heuristic technique for constructing candidate solutions to Sum-of-Squares SDPs called \emph{pseudocalibration}.
Subsequently, the pseudocalibration technique was used in \cite{HKPRSS17} to show SoS lower bounds for Tensor PCA and Sparse PCA.  
Building on ideas from pseudocalibration, Hopkins and Steurer \cite{HS17} recovered conjectured computational thresholds in community detection, while \cite{CohenR20} use it towards showing LP extended formulation lower bounds for $\textsf{Random 3SAT}$.

In an independent work, Kunisky and Bandeira \cite{KB19} also obtained a degree-$4$ SoS integrality gap for the Sherrington--Kirkpatrick refutation problem.

\subsection{Technical overview} \label{sec:tech-overview}
The mapping $\Phi$ alluded to in \pref{thm:main-lifting} is quite intricate and we are unable to motivate the construction of the mapping in a canonical fashion.
Instead, we focus on how the map $\Phi$ was first constructed in the context of the Boolean Vector in Random Subspace problem.

Fix a randomly chosen subspace $\bV$ of dimension $d$ in $\mathbb{R}^n$.
With high probability, no boolean vector $x \in \{-1,1\}^n$ is close to $\bV$ (every boolean vector $x$ has correlation less than $\frac{2}{\pi}+o_n(1)$ with $\bV$).
To prove that the degree $4$ SoS SDP cannot refute the existence of a boolean vector in $\bV$, we need to construct a degree $4$ pseudodistribution $\mu$ such that,
$$ \E_{x \in \{-1,1\}^n } [ \mu(x) x^{\dagger} \Pi_{\bV} x] \approx n \ .$$
In words, the pseudo-distribution $\mu$ is seemingly supported on vectors $x$ in the subspace $\bV$. 

\paragraph{Pseudocalibration.}
We will now use the pseudocalibration recipe of Barak \cite{DBLP:journals/corr/BarakHKKMP16} to arrive at the pseudo-distribution $\mu$.

The idea is to construct a planted distribution $\Theta$ over pairs $(x,\bV)$ where $x \in \{-1,1\}^n$, $x \in \bV$ and the subspace $\bV$ is a {\it seemingly} random subspace.
For example, a natural planted distribution $\Theta$ would be given by the following sampling procedure:
\begin{itemize}
    \item Sample $x \in \{-1,1\}^n$ uniformly at random.
    \item Sample a uniformly random subspace $W$ of dimension $\dim(W) = d-1$ and set $\bV = \textsf{Span}(W \cup \{x\})$.
\end{itemize}
It is clear that the pair $(x,\bV)$ satisfies all the desired properties of the planted distribution.

Let $\textsf{Gr}(n,d)$ denote the space of all $d$-dimensional subspaces of $\R^n$.  Let $\Theta$ denote the density associated with the planted distribution, i.e., $\Theta$ is a function over $\textsf{Gr}(n,d) \times \{-1,1\}^n $. \footnote{Technically, the density $\Theta$ needs to be represented by a distribution}

For any specific $\bV \in \textsf{Gr}(n,d)$, notice that the restriction $\Theta_{\bV}(x) = \Theta(x,\bV)$ is up to a factor normalization, a valid probability distribution over $\{-1,1\}^n$.  Therefore, $\Theta_{\bV}$ is a solution to the degree $d$ SoS SDP relaxation for all $d$, upto the normalization factor.
Ignoring the issue of the normalization factor for now, the candidate degree $4$ moment matrix would be given by,
\begin{equation}\label{eq:pseudocal1}
\cM^*_{\bV} [S,T] = \E_{x \in \{-1,1\}^n} \left[\left(\prod_{i \in S} x_i\right) \left(\prod_{j \in T} x_j\right) \cdot \Theta(x,\bV) \right]
\end{equation}
The matrix $\cM^*$ is clearly positive semidefinite for each $\bV$.  To formally construct the Cholesky factorization of $\cM^*$, one defines the vectors $\{ V_{S}: \{-1,1\}^n \to \R\}$ to be the functions $V_S^*(x) = \prod_{i \in S} x_i \cdot (\Theta(x,\bV))^{1/2}$.  The inner product between the vectors $f,g$ is given by
$$ \langle f(x), g(x) \rangle = \E_{x \in \{-1,1\}^n} \left[ f(x) g(x)  \right].$$
With these definitions, we will have
\begin{equation}\label{eq:idealsdp}
\cM^*[S,T] = \langle V^*_S,V^*_T \rangle 
\end{equation}
as desired.  While the above ideal SDP solution and vectors satisfies most of the constraints, it fails the normalization.
In fact, the normalization factor $\Gamma_{\bV} = \E_{x \in \{-1,1\}^n}\left[ \Theta_{\bV}(x)\right]$ is very spiky, it is zero on almost all instances $\bV$ except being very large on subspaces $\bV$ containing a boolean vector.

The key insight of pseudocalibration is to project the planted density $\Theta$ to low-degree functions in $\Theta$, or equivalently truncate away the part of $\Theta$ that is high degree in the instance $\bV$.
Let $\Theta^{\leq D}$ denote the low-degree truncation of the planted density $\Theta$.  For any $\bV \in \textsf{Gr}(n,d)$, the pseudo-calibrated pseudodensity  $\Theta^{\leq D}[\bV]: \{-1,1\}^n \to \R$
is given by $\Theta^{\leq D}[\bV](x) = \Theta^{\leq D}(\bV,x)$.
More concretely, the candidate SDP solution specified by pseudo-calibration is
\begin{equation}\label{eq:pseudocal1}
\cM_{\bV} [S,T] = \E_{x \in \{-1,1\}^n} \left[\left(\prod_{i \in S} x_i\right) \left(\prod_{j \in T} x_j\right) \cdot \Theta^{\leq D}[\bV](x) \right]
\end{equation}
for all $S,T$.  The feasibility of $\cM_{\bV}$ needs to be established, which often requires considerable technical work, especially the proof of positive semidefiniteness of $\cM_{\bV}$.

A natural approach to prove psdness of $\cM_{\bV}$ is to construct the corresponding SDP vectors (Cholesky factorization) by using a low-degree truncation of the ideal SDP vectors $V_S^*$ defined above.  Since $\cM_{\bV}$ is obtained by truncating an ideal solution $\cM^*$ to low-degree polynomials, it would be conceivable that the low-degree truncation of the ideal SDP vectors yield Cholesky factorization of $\cM_{\bV}$.
Unfortunately, this hope does not come to fruition and to our knowledge does not hold for any problem.

\paragraph{Representations.}
Executing the above strategy over $\textsf{Gr}(n,d)$ is technically challenging since low-degree polynomials over $\textsf{Gr}(n,d)$ are complicated.
To cope with the technical difficulty, it is better to work with an explicit representation of the subspace $\bV$.   Specifically, $\bV$ can be represented by a $n \times \kappa$ matrix $M_{\kappa}$ in that $\bV = \textsf{Col-Span}(M_{\kappa})$.  Any choice of $\kappa \geq d$ would suffice to represent a $d$-dimensional subspace $\bV$, and in our construction we will set $\kappa \to \infty$.

With this representation, a candidate planted distribution $(x,M_{\kappa})$ is sampled as follows:
\begin{itemize}
    \item Sample $x \in \{-1,1\}^n$ uniformly at random.
    \item Sample $d-1$ vectors $w_1,\ldots,w_{d-1} \in \R^n$ from the standard normal distribution $N(0,1)^n$.  Let $M$ be the $n \times d$ matrix whose columns are $x$ and $w_1,\ldots,w_{d-1}$.
    \item Let $U_\kappa \in \R^{\kappa \times \kappa}$ be a random unitary matrix, and let $U_\kappa^{\leq n} \in \R^{n \times \kappa}$ matrix denote the first $n$ rows of $U_\kappa$.  Set $M_\kappa = M \cdot  U_\kappa^{\leq n}$
\end{itemize}
First, notice that $x \in \textsf{Col-Span}(M)$ as needed.  However, the representations are not unique in that each subspace $\bV$ has infinitely many different representations.  Further, the original SoS optimization problem depends solely on the subspace $\bV$, and is independent of the matrix $M_{\kappa}$ representing $\bV$.

At first, these redundant representations or inherent symmetries of the planted density, seem to be an issue to be dealt with.  It turns out that these redundancy in representations is actually useful in constructing the SDP vectors!  

\paragraph{Planted Distribution.}
Before proceeding, we will first simplify our planted distribution even further.  Since computations over random unitary matrices are technically difficult, we will select a much simpler finite subgroup of the unitary group to work with.
In particular, the planted distribution $\Theta$ over pairs $(x,M)$ is sampled as follows:
\begin{itemize}
    \item Sample $x \in \{-1,1\}^n$ uniformly at random.
    \item Sample $d-1$ vectors $w_1,\ldots,w_{d-1} \in \R^n$ from the standard normal distribution $N(0,1)^n$.  Let $M$ be the $n \times d$ matrix whose columns are $x$ and $w_1,\ldots,w_{d-1}$.
    \item Let $H_{\kappa}^{\leq n}$ denote the $n \times \kappa$ matrix obtained by taking the first $n$ rows of the Hadamard matrix $H_\kappa$.  Let $\bZ \in \R^{\kappa \times \kappa}$ denote a diagonal matrix with random $\{\pm 1\}$ entries. Set $M_\kappa = M H^{\leq n}_{\kappa} \bZ$
\end{itemize}

The above construction uses $H_{\kappa} \bZ$
instead of a unitary random matrix $\bU_\kappa$.  In particular, the continous unitary group is replaced with a finite set of $2^{\kappa}$ transformations indexed by the familiar $\{-1,1\}^{\kappa}$, making the calculations tractable.

\paragraph{Exploiting multiple representations.}

Applying the pseudo-calibration heuristic to the planted density $(x,M_\kappa)$ defined above, we get a candidate {\it ideal} SDP solution $\cM_{M_{\kappa}}$ 
\begin{equation} \label{eq:pseudocal2}
\cM^*_{M_{\kappa}} [S,T] = \E_{x \in \{-1,1\}^n} \left[\left(\prod_{i \in S} x_i\right) \left(\prod_{j \in T} x_j\right) \cdot \Theta(M_{\kappa},x) \right]
\end{equation}
This ideal SDP solution needs to be truncated to low-degree with $\Theta$ to be replaced by $\Theta^{\leq D}$.
The specifics of the low-degree projection used to define $\Theta^{\leq D}$ are intentionally left vague at this time.

The construction thus far is essentially the pseudocalibration heuristic albeit on a somewhat complicated planted distribution.
It is at this time that we will exploit the symmetries of the planted density.
Recall that the underlying subspace $\bV$ depends only on $\textsf{Col-Span}(M_{\kappa}) = \textsf{Col-Span}(M)$, and so does the underlying SoS SDP relaxation.  Therefore, it is natural to average out the above pseudocalibrated solution over the various representations of $\bV$, i.e., 
define the solution $\cM_{\bV}$ as,
\begin{equation} \label{eq:pseudocal2}
\cM^*_{\bV} [S,T] = \E_{\bZ} \left[ \E_{x \in \{-1,1\}^n} \left[\left(\prod_{i \in S} x_i\right) \left(\prod_{j \in T} x_j\right) \cdot \Theta[MH_{\kappa}^{\leq n} \bZ](x) \right]\right]
\end{equation}
Analogous to the ideal SDP vectors \pref{eq:idealsdp}, one can define SDP vectors $V^*_S$ here, but this time as functions over both $x$ and $\bZ$.  That is if we $V^*_S(x,\bZ) = (\prod_{i \in S} x_i) \cdot \Theta[M H_{\kappa}^{\leq n} \bZ](x) $ then,
$$ \cM_{\bV}^*[S,T] = \langle V^*_S(x,\bZ), V^*_T(x,\bZ) \rangle$$
where $\langle f(x,\bZ), g(x,\bZ) \rangle = \E_{\bZ} \E_{x \in \{-1,1\}^n}\left[ f(x,\bZ)g(x,\bZ)  \right]$.

The above construction looks similar to \pref{eq:pseudocal1} and \pref{eq:idealsdp} with one important difference.  The quantities are a function of the matrix $M$ defining the subspace and a set of redundancies in representation given by $\bZ$.
In particular, {\it low-degree truncation} $\Theta^{\leq D}$ can include truncation in the degree over $M$ and over $\bZ$ separately.

Somewhat mysteriously, it turns out that by choosing a low-degree truncation (in {\bf both $M$ and $\bZ$}) of both the ideal SDP solution  $\cM^*$ and the ideal vectors $V^*_S$, we can recover SDP solution along with an approximate Cholesky factorization (analogous to \pref{eq:idealsdp}).  
While the above discussion describes how we arrive at the definition of the mapping.  The proof that the mapping works amounts to showing that the truncated vectors yield an approximate Cholesky factorization of the pseudo-calibrated matrix, which forms the technical heart of the paper.
We defer the details of the construction to \pref{sec:construction}.

\paragraph{Bounding Spectral Norm} 
We exhibit a candidate SoS SDP solution $\cM^{(1)}$ and show that there exists a psd matrix $\cM^{(2)}$ that is close in spectral norm to $\cM^{(2)}$.  The difference $\cM^{(1)} - \cM^{(2)}$ is matrix with entries that are low-degree polynomials in the input $M$, and our goal is to upper bound the spectral norm $\| \cM^{(2)} - \cM^{(1)} \|$.

As is typical, this involves obtaining spectral norm bounds on matrices whose entries are low-degree polynomials.
Earlier works on Planted Clique \cite{DM15, BHKKMP19} and others have developed technical machinery based on the trace method towards bounding spectral norms.
We present a simpler factorization based technique to obtain bounds on spectral norms here.
Owing to its simplicity, it is broadly applicable to more complicated ensembles of random matrices such as those arising in sparse $d$-regular random graphs. 
Furthermore, in some cases, the technique yields tighter bounds than trace method.
For example, consider the following random matrix.  Let $A \in \R^{n \times n}$ be a random symmetric matrix with $A_{ii} = 0$ for all $i $ and $A_{ij}$ being independent $\{\pm 1\}$ entry otherwise.
Consider the random matrix $B \in \R^{[n]^2 \times [n]^2}$ defined as,
\begin{equation*}
    B[(i_1,i_2),(j_1,j_2)] = A_{i_1 j_1} \cdot A_{i_2 j_1} \cdot A_{i_2 j_2} \ .
\end{equation*}
The best known bounds for $\| B\|$ using the trace method imply that $\|B \| \leq n \cdot (\log{n})^c$ for some constant $c$ \cite{DM15}.  On the other hand, the factorization technique outlined in \pref{sec:PSDness}
can be easily used to obtain a $\Theta(n)$ upper bound (specifically, an upper bound of $4n$).

All our spectral norm bounds are obtained via the factorization method, starting from bounds on the norm of the original matrix $A$.

\section{Lifts of a degree-$2$ pseudoexpectation}   \label{sec:construction}

In this section, we describe how to obtain a degree-$4$ pseudoexpectation $\pE_4$ from a degree-$2$ pseudoexpectation $\pE_2$.  We specify $\pE_4$ via its pseudomoment matrix $\calM$ whose rows and columns are indexed by sets of size at most $2$, with $\calM[{S,T}] = \pE_4\left[x^{S\Delta T}\right]$.  Let $\calM'$ be the following $n\times n$ submatrix of the degree-$2$ pseudomoment matrix:
\[
    \calM'[\{i\},\{j\}] \coloneqq \pE_2[x_ix_j]\qquad i,j\in[n].
\]
Since $\calM'$ is positive semidefinite, we can write $\calM'$ in its Cholesky decomposition $MM^{\dagger}$ where $M$ is some $n\times n$ matrix.

For each $\HadDim\ge n$ that is a power of $2$, let $H_{\HadDim}^{\le n}$ denote the $n \times \HadDim$ matrix obtained by taking the first $n$ rows of the Hadamard matrix $H_{\HadDim}$.  We first define a $n\times \HadDim$ matrix $M_\kappa \coloneqq MH_{\HadDim}^{\le n}$.  A key property of $M$ we use is:
\begin{fact}    \label{fact:inner-product-preserved}
    $\langle M[i], M[j] \rangle = \langle M_{\HadDim}[i], M_{\HadDim}[j] \rangle$ where $M[t]$ denotes the $t$-th row of $M$ since the rows of $H_{\HadDim}^{\le n}$ are orthogonal unit vectors.
\end{fact}

Fix a set of indeterminates $z_1,\dots,z_{\HadDim}$ obeying $z_i^2=1$.  For each $i\in[n]$, we define ``seed polynomials''
\[
    q_{i,\HadDim}(z) \coloneqq \sum_{j\in[\HadDim]}M_{\HadDim}[i,j]z_j - 2\sum_{\{j_1,j_2,j_3\}\subseteq[\HadDim]}M_{\HadDim}[i,j_1]M_{\HadDim}[i,j_2]M_{\HadDim}[i,j_3]z_{j_1} z_{j_2}z_{j_3}
\]
and for each subset $S\subseteq[n]$ define ``set polynomials''
\[
    q_{S,\HadDim}(z) \coloneqq \prod_{i\in S}q_{i,\HadDim}(z).
\]
We now define matrix $\calM^{(1)}$ as follows:
\[
    \calM^{(1)}[S,T] \coloneqq \lim_{\HadDim\to\infty}\E_{\bz\sim\{\pm1\}^{\HadDim}} [q_{S\Delta T,\HadDim}(\bz)] \numberthis \label{eq:lim-def-M1}
\]
We prove that the limit on the right-hand side of the above expression exists in \pref{cor:moment-matrices-well-defined}.

We pick our pseudomoment matrix $\calM$ as a mild adjustment to $\calM^{(1)}$.  Specifically, we define
\[
    \calM \coloneqq (1-\eta)\calM^{(1)} + \eta\cdot\Id.
\]
where we choose $\eta$ later.

It is clear that $\calM$ satisfies the ``Booleanness'' and ``symmetry'' constraints.  It remains to prove that $\calM$ is positive semidefinite for appropriate choice of $\eta$.

Towards doing so, we define a new matrix $\calM^{(2)}$.  Define ``truncated polynomials''
\[
    p_{S,\HadDim}(z) \coloneqq q_{S,\HadDim}(z)^{\le |S|}
\]
where $q_{S,\HadDim}(z)^{\le \tau}$ denotes the projection of $q_{S,\HadDim}$ onto the space of polynomials spanned by $\Fourier{T}$ where $|T|\le\tau$.  And define $\calM^{(2)}$ as:
\[
    \calM^{(2)}[S,T] \coloneqq \lim_{\HadDim\to\infty}\E_{\bz\sim\{\pm 1\}^{\HadDim}}[p_{S,\HadDim}(\bz)p_{T,\HadDim}(\bz)] \numberthis \label{eq:lim-def-M2}
\]
Once again, we defer the proof that the limit on the right-hand side exists to \pref{cor:moment-matrices-well-defined}.  $\calM^{(2)}$ is PSD as it is the limit of second moment matrices, each of which is PSD.

To show $\calM$ is PSD, we first bound the spectral norm of $\calM^{(1)}-\calM^{(2)}$.
\begin{lemma}   \label{lem:spec-norm-bound}
    Let $\aspec \coloneqq \|\calM'\|_2$, $\altwo\coloneqq\max_{i\in[n]}\sqrt{\sum_{j\ne i}\calM'[i,j]^2}$, $\amag\coloneqq\max_{i,j:i\ne j}\calM'[i,j]$.  There is an absolute constant $C>0$ such that $\alpha \coloneqq C\amag\cdot(1+\arow^4)\cdot(1+\aspec^2)$ and
    $\|\calM^{(1)}-\calM^{(2)}\|_2 \le \alpha$.
\end{lemma}
\noindent \pref{lem:spec-norm-bound} is an immediate consequence of \pref{lem:spec-norm-summary}, which \pref{sec:PSDness} is dedicated to proving.

\begin{corollary}   \label{cor:almost-PSD}
    Let $\alpha$ be as in the statement of \pref{lem:spec-norm-bound}.  Then $\lambda_{\min}(\calM^{(1)}) \ge -\alpha$.
\end{corollary}
\begin{proof}
    For any unit vector $x$,
    \begin{align*}
        x^\dagger \calM^{(1)}x &= x^\dagger \left(\calM^{(1)} - \calM^{2} + \calM^{(2)}\right)x\\
        &= x^\dagger \left(\calM^{(1)}-\calM^{(2)}\right)x + x^\dagger\calM^{(2)}x\\
        &\ge -\alpha &\text{(by \pref{lem:spec-norm-bound} and PSDness of $\calM^{(2)}$)}
    \end{align*}
\end{proof}
Set $\eta\coloneqq\frac{\alpha}{1+\alpha}$.  The PSDness of $\calM$ follows from \pref{cor:almost-PSD} and the fact that adding $\eta\cdot\Id$ to any matrix increases all its eigenvalues by $\eta$.
\begin{theorem}
    $\calM \psdge 0$.
\end{theorem}

\begin{lemma} \label{lem:sos-lift-obj}
    Let $\alpha$ be as in the statement of \pref{lem:spec-norm-bound}.  For any Hermitian matrix $A\in\R^{n\times n}$,
    \[
        \pE_4[x^{\dagger}Ax] \ge \left(1-\frac{\alpha}{1+\alpha}\right)(\pE_2[x^{\dagger}Ax]-\alpha\sqrt{n}\|A\|_{F})+\frac{\alpha}{1+\alpha}\Tr{A}.
    \]
\end{lemma}
\begin{proof}
    For a matrix $L$ with rows and columns indexed by subsets of $[n]$, we use the notation $L_{1,1}$ to denote the submatrix of $L$ with rows and columns indexed by sets of size exactly equal to $1$.
    \begin{align*}
        \pE_4[x^{\dagger}Ax]-\frac{\alpha}{1+\alpha}\Tr{A} &= \langle\calM_{1,1},A\rangle - \frac{\alpha}{1+\alpha}\Tr{A} \\
        &= \left(1-\frac{\alpha}{1+\alpha}\right)\langle\calM^{(1)}_{1,1},A\rangle\\
        &= \left(1-\frac{\alpha}{1+\alpha}\right)(\langle\calM^{(2)}_{1,1},A\rangle + \langle \calM^{(1)}_{1,1}-\calM^{(2)}_{1,1}, A\rangle)\\
        &\ge \left(1-\frac{\alpha}{1+\alpha}\right)\left(\langle\calM^{(2)}_{1,1},A\rangle - \left\|\calM_{1,1}^{(1)}-\calM_{1,1}^{(2)}\right\|_F\cdot\|A\|_F\right)\\
        &\ge \left(1-\frac{\alpha}{1+\alpha}\right)\left(\langle\calM^{(2)}_{1,1},A\rangle - \alpha\cdot\sqrt{n}\cdot\|A\|_F\right) &\text{(by \pref{lem:spec-norm-bound})}
    \end{align*}
    Observe that $\calM^{(2)}_{1,1}$ is exactly equal to $\pE_2[xx^{\dagger}]$ and hence the statement of the lemma follows.
\end{proof}

\section{Spectral Norm Bounds}
\label{sec:PSDness}
This section is dedicated to proving \pref{lem:spec-norm-bound}.  We first make some structural observations about $\calE \coloneqq \calM^{(1)}-\calM^{(2)}$.

\begin{observation}
    Suppose $|S\Delta T|$ is odd.  Then $\calE[{S,T}] = 0$.
\end{observation}
\begin{proof}
    Since $q_{i,\HadDim}(z)$ is a sum odd degree terms in $z$, so is $q_{S\Delta T,\HadDim}(z)$ when $|S\Delta T|$ is odd, and so the expected value of each term over the choice of random $\bz$ is $0$.  Thus, $\calM_{\HadDim}^{(1)}[S,T] = 0$, and by extension $\calM^{(1)}[S,T]=0$.  Note that for any set $S$ all terms in $p_{S,\HadDim}$ have the same parity as $|S|$, and thus all terms in $p_{S,\HadDim}p_{T,\HadDim}$ have the same parity as $|S|+|T|$, whose parity is the same as $|S\Delta T|$.  Thus, $\calM_{\HadDim}^{(2)}[S,T]=0$ and consequently $M^{(2)}[S,T]=0$.
\end{proof}

\begin{observation} \label{obs: empty-empty}
    Suppose $S = \emptyset$ or $T = \emptyset$.  Then $\calE[{S,T}] = 0$.
\end{observation}
\begin{proof}
    Without loss of generality, say $S = \emptyset$.  Then $\calM^{(1)}[{S,T}] = \lim_{\HadDim\to\infty}\E_{\bz\sim\{\pm 1\}^d}[q_{T,\HadDim}(\bz)] = \lim_{\HadDim\to\infty} \wh{q_{T,\HadDim}}(\emptyset)$.  Similarly, $\calM^{(2)}[{S,T}] = \lim_{\HadDim\to\infty} \E_{\bz\sim\{\pm 1\}^d}[p_{T,\HadDim}(\bz)] = \lim_{\HadDim\to\infty} \wh{p_{T,\HadDim}}(\emptyset) = \lim_{\HadDim\to\infty} \wh{q_{T,\HadDim}}(\emptyset)$.
\end{proof}

Thus, we can split $\calE$  into four parts.
\begin{align*}
    \calE^{(1)}[S,T] &:=
    \begin{cases}
        \calE[{S,T}] &\text{$S=T$}\\
        0 &\text{otherwise}
    \end{cases}\\
    \calE^{(2)}[S,T] &:=
    \begin{cases}
        \calE[{S,T}] &\text{if $|S|=|T|=1$, $|S\cap T| = 0$}\\
        0 &\text{otherwise}
    \end{cases}\\
    \calE^{(3)}[S,T] &:=
    \begin{cases}
        \calE[{S,T}] &\text{if $|S|=|T|=2$, $|S\cap T| = 1$}\\
        0 &\text{otherwise}
    \end{cases}\\
    \calE^{(4)}[S,T] &:=
    \begin{cases}
        \calE[{S,T}] &\text{if $|S|=|T|=2$, $|S\cap T| = 0$}\\
        0 &\text{otherwise}
    \end{cases}
\end{align*}
Since $\calE = \calE^{(1)}+\calE^{(2)}+\calE^{(3)}+\calE^{(4)}$, proving a spectral norm bound on each individual piece also gives a bound of the spectral norm of $\calE$ via the triangle inequality.  In later parts of the section, the following are proved.
\begin{lemma}   \label{lem:spec-norm-summary}
    The following spectral norm bounds hold:
    \begin{align*}
        \|\calE^{(1)}\| &\le O(\amag)\\
        \|\calE^{(2)}\| &\le O(\arow^2\cdot\amag)\\
        \|\calE^{(3)}\| &\le O(\amag\cdot(1+\aspec+\arow^2))\\
        \|\calE^{(4)}\| &\le O(\amag\cdot(1+\arow^4)\cdot(1+\aspec^2)).
    \end{align*}
    In particular, this implies $\displaystyle\left\|\calE\right\|\le O(\amag\cdot(1+\arow^4)\cdot(1+\aspec^2))$.
\end{lemma}

$\left\|\calE^{(1)}\right\|$ is bounded in \pref{sec:E-1}, $\left\|\calE^{(2)}\right\|$ is bounded in \pref{sec:E-2}, $\left\|\calE^{(3)}\right\|$ is bounded in \pref{sec:E-3}, and $\left\|\calE^{(4)}\right\|$ is bounded in \pref{sec:E-4}.

Before diving into the proofs, we introduce the language of graphical matrices.

\subsection{Graphical Polynomials and Graphical Matrices}   \label{sec:graph-mat}
Akin to \cite{BHKKMP19}, we give a way to associate matrices with constant sized graphs.  To motivate studying graphical matrices, we start with some simple examples.  Let $H$ be some graph with vertex set $[n]$.  Now, consider the graph $\Glyph$ in the figure below.

\begin{figure}[h]
\centering
\begin{tikzpicture}[
every edge/.style = {draw=black,very thick},
 vrtx/.style args = {#1/#2}{%
      circle, draw, thick, fill=white,
      minimum size=5mm, , label=#1:#2}]
       
\node(A) [vrtx=left/] at (0, 1) {$a_1$};
\node(B) [vrtx=left/] at (0, 0) {$a_2$};
\node(D) [vrtx=right/] at (3,0) {$a_4$};
\node(C) [vrtx=right/] at (3, 1) {$a_3$};
\node(m1) [vrtx=right/] at (1.5, 2) {$b_1$};
\node(m2) [vrtx=right/] at (1.5, 1) {$b_2$};
\node(m3) [vrtx=right/] at (1.5, 0) {$b_3$};
\node(m4) [vrtx=right/] at (1.5, -1) {$b_4$};

\path   (A) edge (m1)
        (m1) edge  (C)
        (A) edge (m2)
        (m2) edge (C)
        (B) edge (m4)
        (m4) edge (D)
        (A) edge (m3)
        (m3) edge (C)
        ;
\end{tikzpicture}
    \caption{ Graph $\Glyph$  }
\end{figure}
Next, define an $n^2\times n^2$ matrix $\GraphMat_{\Glyph}$, which is the ``graphical matrix'' of $\Glyph$ with rows and columns indexed by size-$2$ subsets of $[n]$ where
\[
    \GraphMat_{\Glyph}[\{i,j\},\{k,\ell\}] \coloneqq \#\{\text{subgraphs of $H$ isomorphic to $\Glyph$ so that $a_1,a_2,a_3,a_4$ map to $i,j,k,\ell$}\}.
\]
Our reason for considering matrices that encode `constant-sized graph statistics' such as the above, which we call \emph{graphical matrices}, is that we are able to naturally view $\calM^{(1)}$ and $\calM^{(2)}$ as a sum of simple graphical matrices\footnote{Where $H$ is replaced with a complete $(n,L)$-bipartite graph, and the edges are equipped with weights from the matrix $M_{\kappa}$.}.  Thus, a natural way to obtain a handle on the spectral norm of $\calM^{(1)}-\calM^{(2)}$ is understanding the spectral behavior of the graphical matrices that constitute it.

\subsubsection{Sketch of Graphical Matrices}
We dig into the specific graphical matrices that arise in this section.  We view the matrix $M_{\HadDim}$ as a weighted bipartite graph with left vertex set $[n]$ and right vertex set $[\HadDim]$, where the weight of the edge between $i\in[n]$ and $j\in[\HadDim]$ is $M_{\HadDim}[i,j]$ --- we call this $\Bipartite(M_{\HadDim})$.
Now, let $\Glyph$ be a bipartite graph on constant number of vertices where each left vertex of $\Glyph$ is one or two of two colors, \row~or \column,~and each right vertex is uncolored.  The \emph{graphical matrix} associated with $\Glyph$ is the $n^{|\row(\Glyph)|}\times n^{|\column(\Glyph)|}$ matrix $\GraphMat_{\Glyph}$ with rows and columns indexed by subsets of $[n]$ of size $|\row(\Glyph)|$ and $|\column(\Glyph)|$ respectively where we obtain the $S,T$ entry in the following way.
\begin{displayquote}
    Enumerate over all subgraphs of $\Bipartite(M_{\HadDim})$ that are isomorphic to $\Glyph$, and vertices colored \row~map into $S$ and the vertices colored \column~map into $T$, take the product of edge weights of each subgraph, and then take the sum over all subgraphs enumerated over.
\end{displayquote}
Symbolically,
\[
    \GraphMat_{\Glyph}[S,T] \coloneqq \sum_{\substack{\calH\text{ subgraph of }\Bipartite(M_{\HadDim})\\
    \calH\text{ isomorphic to }\Glyph\\
    \row(\Glyph)\text{ maps into $S$}\\
    \column(\Glyph)\text{ maps into $T$}}} \prod_{\{i,j\}\in\calH}M_{\HadDim}[i,j].
\]

\subsubsection{Definitions}

\begin{definition}[Half-Glyph]
    A half-glyph $\cHG$ is a bipartite (multi-)graph with a left vertex set $L(\cHG)\coloneqq\{\ell_1,\dots,\ell_{|L(\cHG)|}\}$, a middle vertex set $M(\cHG)\coloneqq\{m_1,\dots,m_{|M(\cHG)|}\}$ and edges $E(\cHG)$.  We use $\cHG_{a,b}$ to represent the number of edges between $\ell_a$ and $m_b$.
\end{definition}

\begin{figure}[h]
\centering
    \begin{subfigure}[b]{0.3\textwidth}\centering
\begin{tikzpicture}[
every edge/.style = {draw=black,very thick},
 vrtx/.style args = {#1/#2}{%
      draw, thick, fill=white,
      minimum size=5mm, label=#1:#2}]
       
\node[rectangle] (A) [vrtx=right/] at (1, 1) {};
\node[circle] (B) [vrtx=left/$\ell_1$] at (0, 0) {};
\node(C) [vrtx=right/] at (1,0) {};
\node(D) [vrtx=right/] at (1, -1) {};
\node[circle] (E) [vrtx=left/$\ell_2$] at (0, -1) {};
\path   (A) edge (B)
        (B) edge (C)
        (B) edge (D)
        (E) edge (D)
        ;
\end{tikzpicture}
    \caption{Half-Glyph $\cHG_1$} \label{fig: a}
    \end{subfigure}
\hfil    
  \begin{subfigure}[b]{0.3\textwidth}\centering
\begin{tikzpicture}[
every edge/.style = {draw=black,very thick},
 vrtx/.style args = {#1/#2}{%
      draw, thick, fill=white,
      minimum size=5mm, label=#1:#2}]
       
\node[rectangle] (A) [vrtx=right/] at (1, 1) {};
\node[circle] (B) [vrtx=right/$\ell_1$] at (2, 0) {};
\node(C) [vrtx=right/] at (1,0) {};
\node[circle] (E) [vrtx=right/$\ell_2$] at (2, -1) {};
\path   (A) edge (B)
        
        (E) edge (C)
        ;
\end{tikzpicture}
    \caption{Half-Glyph $\cHG_2$ } \label{fig: b}
    \end{subfigure}

\caption{Half-Glyphs \protect\footnotemark} 
    \label{fig:figure-1}
\end{figure}
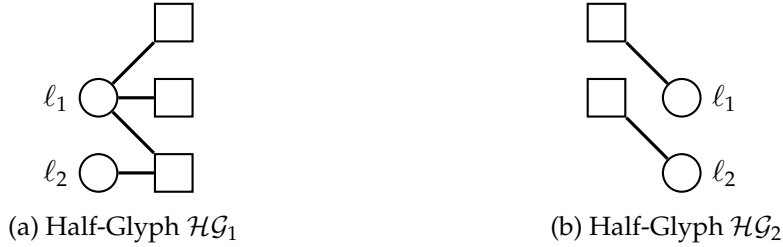
\footnotetext{We will use circles to represent vertices in $L(\cHG)$ (and later $L(\cG)$ and $R(\cG)$) that should be thought as vertices in $[n]$ and square to represent vertices in $M(\cG)$ that should be thought as indeterminates $z$.}
\begin{definition}[Half-Glyph Labeling]
   For a half-glyph $\cHG$, we call $S:L(\cHG) \rightarrow[n]$ a valid labeling if
   \begin{enumerate}
       \item It is a injective map from $L(\cHG)$ to $[n]$.
       \item $S(\ell_i)<S(\ell_j)$ if and only if $i<j$.\footnote{This ``order-preserving'' requirement is an artifact of our proof.}
   \end{enumerate}
   
\end{definition}

\begin{remark}
    For simplicity, we represent each valid labeling as a size-$|L(\cHG)|$ subset of $[n]$.
\end{remark}

\begin{definition}[Cluster of $M(\cHG)$]
For a half-glyph $\cHG$, we call a set of vertices $\{v_1,\dots, v_{|B|}\}$ in cluster $B$ if they have the same neighborhood on $L(\cHG)$, i.e., $\forall i,j\in B$, $\cHG(\ell,i)= \cHG(\ell,j)$ for any $\ell\in L(\cHG)$. We let $\calB(\cHG)=(B_1,\dots, B_k) $ be the set of clusters in $\cHG$ where $k\leq \kappa$ is the number of clusters.
\end{definition}


\begin{definition}[$z$-labeling of half-glyph]
    We say $\pi:M(\cHG)\to[\kappa]$ is a $z$-labeling if it is injective, and if for each cluster $B_i\in\calB(\cHG)$ and $m_a,m_b\in B_i$, $\pi(m_a)<\pi(m_b)$ if and only if $a<b$. We denote the set of $z$-labelings by $\Pi(\cHG)$.
\end{definition}

\begin{definition}[$\kappa$-Graphical Polynomial of a Half-Glyph] 
    For any $\kappa$, every half-glyph $\cHG$ with a valid labeling $S$ is associated with a polynomial over indeterminates $z=(z_1,\ldots, z_\kappa)$ given by
    \begin{equation*}
        \beta_{\cHG, \kappa, S}(z) \coloneqq \sum_{\pi \in \Pi(\cHG)} \prod_{i  \in L(\cHG)} \prod_{j \in M(\cHG)} (M_{\kappa}[S(i),\pi(j)] \cdot z_{\pi(j)})^{\cHG_{i,j}}
    \end{equation*}
\end{definition}   
    


\begin{definition}[Glyph]
A glyph $\cG$ is a multi-graph on the vertex set $V(\cG)=L(\cG) \cup M(\cG)\cup R(\cG)$ and edge set $E(\cG)$, where $L(\cG)\cup R(\cG) = \{v_1,v_2,\dots, v_{|L(\cG)\cup R(\cG)|} \}$ and $M(\cG) = \{m_1,m_2,\dots, m_{|M(\cG)|}\}$.  We use $\cG_{a,b}$ to represent the number of edges between $v_a$ and $m_b$.
\end{definition}

\begin{remark}
Our definition of cluster and $z$-labeling for half-glyph extends naturally to glyph.
\end{remark}{}

We will refer to $L(\cG)$ as {\it left} vertices, $M(\cG)$ as {\it middle} vertices, and $R(\cG)$ as {\it right} vertices of the glyph.  We emphasize that $L(\cG)$ and $R(\cG)$ need not be disjoint; in particular some vertices can be both {\it left} and {\it right} vertices.  In the following figure, $\cG_1$ and $\cG_2$ are different glyphs because $L$ and $R$ intersect in $\cG_1$ but not in $\cG_2$. 

\begin{figure}[h]
\centering
    \begin{subfigure}[b]{0.3\textwidth}\centering
\begin{tikzpicture}[
every edge/.style = {draw=black,very thick},
 vrtx/.style args = {#1/#2}{%
      draw, thick, fill=white,
      minimum size=5mm, label=#1:#2}]
       
\node[rectangle] (A) [vrtx=right/] at (1, 1) {};
\node[circle] (B) [vrtx=left/$v_1$] at (0, 0) {};
\node[circle] (B') [vrtx=right/$v_1$] at (2, 0) {};
\node(C) [vrtx=right/] at (1,0) {};
\node(D) [vrtx=right/] at (1, -1) {};
\node[circle] (E) [vrtx=left/$v_2$] at (0, -1) {};
\node[circle] (F) [vrtx=right/$v_3$] at (2, -1) {};
\path   (A) edge (B)
        (B) edge (C)
        (B) edge (D)
        (E) edge (D)
        (B') edge (A)
        (F) edge (C);        ;
\end{tikzpicture}
    \caption{Glyph $\cG_1$} \label{fig: a}
    \end{subfigure}
\hfil    
  \begin{subfigure}[b]{0.3\textwidth}\centering
\begin{tikzpicture}[
every edge/.style = {draw=black,very thick},
 vrtx/.style args = {#1/#2}{%
      draw, thick, fill=white,
      minimum size=5mm, label=#1:#2}]
       
\node[rectangle] (A) [vrtx=right/] at (1, 1) {};
\node[circle] (B) [vrtx=left/$v_1$] at (0, 0) {};
\node[circle] (B') [vrtx=right/$v_4$] at (2, 0) {};
\node(C) [vrtx=right/] at (1,0) {};
\node(D) [vrtx=right/] at (1, -1) {};
\node[circle] (E) [vrtx=left/$v_2$] at (0, -1) {};
\node[circle] (F) [vrtx=right/$v_3$] at (2, -1) {};
\path   (A) edge (B)
        (B) edge (C)
        (B) edge (D)
        (E) edge (D)
        (B') edge (A)
        (F) edge (C);        ;
\end{tikzpicture}
    \caption{Glyph $\cG_2$} \label{fig: b}
    \end{subfigure}

\caption{Glyphs}
    \label{fig:figure-1}
\end{figure}

%
Observe that any glyph can be seen as being "composed" of two half-glyphs: the \emph{left half-glyph} $\calL(\cG)$ which is the induced subgraph on $L(\cG)\cup M(\cG)$, and the \emph{right half-glyph} $\calR(\cG)$ which is the induced subgraph on $R(\cG)\cup M(\cG)$.  We now extend the definition of labeling and graphical polynomial to glyphs.

\begin{definition}[Glyph Labeling]
    For any glyph $\cG$, let $S$ be a valid labeling for $\calL(\cG)$, and $T$ be a valid labeling for $\calR(\cG)$, $S$ and $T$ are {\it $\cG$-compatible}  if they agree on $L(\cG) \cap R(\cG)$, i.e.
    $S\vert_{L(\cG) \cap R(\cG)} = T\vert_{L(\cG) \cap R(\cG)}$ and are disjoint on their symmetric difference, i.e. $S (L(\cG) \setminus R(\cG)) \cap T(R(\cG) \setminus L(\cG)) = \emptyset$.  For two $\cG$-compatible labelings $S$ and $T$, let $S \circ T : L(\cG) \cup R(\cG) \to [n]$ denote the joint labelling induced by both.
\end{definition}

\begin{definition}[$\kappa$-Graphical Polynomial of a Glyph] For any $\kappa$, for a glyph $\cG$ with half-glyphs $\calL(\cG)$ and $\calR(G)$ and a pair of compatible labelings $S,T$, we associate it with a polynomial over indeterminates $z=(z_1,\dots,z_\kappa)$ given by
\[ 
    \beta_{\cG, \kappa, S\circ T}(z) \coloneqq \sum_{\pi\in \Pi(\cG) }\prod_{i\in L(\cG)\cup R(\cG)}\prod_{j\in M(\cG)} (M_{\kappa}[S\circ T(i),\pi(j)]\cdot z_{\pi(j)})^{\cG_{i,j}}
\]
\end{definition}

\begin{definition}
    A glyph $\Glyph$ is called \emph{well-glued} if every middle vertex has even degree.
\end{definition}

\begin{remark}
    The $\kappa$-graphical polynomial of a well-glued glyph does not depend on $z$.  Specifically,
    \[
        \beta_{\cG, \kappa, S\circ T} = \sum_{\pi\in \Pi(\cG) }\prod_{i\in L(\cG)\cup R(\cG)}\prod_{j\in M(\cG)} M_{\kappa}[S\circ T(i),\pi(j)]^{\cG_{i,j}}
    \]
\end{remark}

\begin{definition}[$\kappa$-Graphical Matrix of a Well-Glued Glyph]
For each well-glued glyph $\cG$, we associate a matrix indexed by ${[n]\choose{L(\cG)}}\times {[n]\choose{R(\cG)}}$ defined as \[
\GraphMat_{\cG, \kappa}[S, T] \coloneqq 1[S,T \text{ are $\cG$-compatible}]\cdot \beta_{\cG,\kappa, S\circ T} 
\]
which we call the $\kappa$-graphical matrix of $\Glyph$.
\end{definition}

\begin{claim} \label{claim:graph-mat-limit}
    Let $\Glyph$ be a well-glued $(A,B)$-glyph.  The limit $\lim_{\HadDim\to\infty}\GraphMat_{\Glyph,\HadDim}$ exists.
\end{claim}
We defer the proof of the claim to \pref{app:lim-graph-mat}.

\begin{definition}[Graphical matrix of a well-glued glyph]
    For a well-glued glyph $\Glyph$, we call the matrix
    \[
        \GraphMat_{\Glyph} \coloneqq \lim_{\HadDim\to\infty}\GraphMat_{\Glyph,\HadDim}
    \]
    the \emph{graphical matrix} of $\Glyph$.
\end{definition}

\begin{definition}
    Given a well-glued glyph $\Glyph$ and a length-$2$ walk that starts at $u\in L(\cG)\cup R(\cG)$, takes an edge to middle vertex $m\in M(\cG)$, and takes a different edge from $m$ to $u'\in L(\cG)\cup R(\cG)$. We call the length-$2$ walk a \emph{cyclic walk} if $u=u'$; otherwise, we call it an \emph{acyclic walk}.
\end{definition}



We also give an explicit expression for the entries of $\GraphMat_{\Glyph}$.
\begin{lemma}   \label{lem:graph-mat-explicit}
    Let $\Glyph$ be a well-glued glyph.  Suppose any middle vertex of $\Glyph$ has degree $\ge 4$, $\GraphMat_{\Glyph} = 0$.  Suppose all middle vertices of $\Glyph$ have degree $2$ and $S\circ T$ is a valid labeling of $\Glyph$ and for $i,j\in L(\Glyph)\cup R(\Glyph)$ let $\calP_{i,j}$ be the collection of length-$2$ walks from $i$ to $j$.  Then:
    \[
        \GraphMat_{\Glyph}[S,T] = \prod_{i\le j\in L(\Glyph)\cup R(\Glyph)} \frac{\langle M[S\circ T(i)], M[S\circ T(j)]\rangle^{|\calP_{i,j}|}}{|\calP_{i,j}|!}.
    \]
\end{lemma}

We also defer the proof of \pref{lem:graph-mat-explicit} to \pref{app:lim-graph-mat}.

\subsection{Glyph Factorization and Spectral Norm Bound}
A useful ingredient towards our goal is a generic way to bound the spectral norm of a graphical matrix.  In \pref{lem:graph-mat-explicit}, we show that the entries of the graphical matrix of a well-glued graph can be written as a product of inner products.  We use this insight to factor the graphical matrices we need to deal with into simpler matrices.  We start with a few basic definitions of types of simple matrices we encounter.
\begin{definition}[Growth and shrinkage matrices]
    We call a matrix a \emph{growth} matrix if it is block-diagonal and each block is a subrow of $MM^{\dagger}$.  We define a \emph{shrink} matrix as one that can be written as the transpose of a growth matrix.
\end{definition}

\begin{definition}[Swap matrices]
    We call a matrix a \emph{swap} matrix if it is block diagonal and each block can be written as either (a) $W-\Id$ where $W$ is a principal submatrix of $MM^{\dagger}$, or (b) $W$ where $W$ is a (not necessarily principal) submatrix of $MM^{\dagger}$.
\end{definition}

\begin{definition}[Residue matrices]
    We call a matrix a \emph{residue} matrix if it is a diagonal matrix and each entry is an off-diagonal entry of $MM^{\dagger}$.
\end{definition}

\begin{lemma}   \label{lem:basic-matrices-spec-bound}
    If $\calL$ is a growth/shrinkage matrix, its spectral norm is bounded by $\altwo$; if it is a swap matrix, its spectral norm is bounded by $\aspec$; and if it is a residue matrix, its spectral norm is bounded by $\amag$.
\end{lemma}
\begin{proof}
    \textbf{Spectral norm bounds on growth and shrinkage matrices.}  Since the rows of a growth matrix are disjointly supported, its spectral norm is equal to the maximum $\ell_2$ norm of its rows.  As each row is a submatrix of $MM^{\dagger}$, $\altwo$ bounds the spectral norm of growth matrices (and shrink matrices too).

    \noindent \textbf{Spectral norm bounds on swap matrices.}  The spectral norm of a swap matrix is equal to the maximum of the spectral norms of its blocks.  If a block is simply a submatrix of $MM^{\dagger}$, its spectral norm is bounded by $\aspec$.  We now treat the case where a block is of the form $W-\Id$ for principal submatrix $W$.  Any principal submatrix of $MM^\dagger$ is PSD and its maximum eigenvalue is bounded by $\aspec$.  Thus all eigenvalues of such a block are between $-1$ and $\aspec-1$.  Combined with the fact that $\aspec\ge 1$ (which follows from the trace of $MM^{\dagger}$ being $n$) we can conclude that the spectral norm of any block is bounded by $\aspec$.

    \noindent \textbf{Spectral norm bounds on residue matrices.}  Since a residue matrix is diagonal, its spectral norm is bounded by the maximum magnitude entry, and since all nonzero entries are off-diagonal entries of $MM^{\dagger}$ a bound of $\amag$ holds on the spectral norm.
\end{proof}

Before jumping into the full proof, we illustrate the efficacy of our method on the following toy example that will appear in our analysis of $\calE^{(4)}$.  Consider the following glyph $\Glyph$ with entries:
\[
    \GraphMat_{\Glyph}[\{i,j\},\{k,\ell\}]=\frac{1}{3!}\langle M[i],M[k]\rangle^3\langle M[j],M[\ell]\rangle
\]
for $i,j,k,\ell\in [n]$ distinct and $i<j$, $k<\ell$.
\begin{figure}[h]
    \centering
\begin{tikzpicture}[
every edge/.style = {draw=black,very thick},
 vrtx/.style args = {#1/#2}{%
      draw, thick, fill=white,
      minimum size=5mm, , label=#1:#2}]
       
\node[circle] (A) [vrtx=left/] at (0, 1) {$i$};
\node[circle] (B) [vrtx=left/] at (0, 0) {$j$};
\node[circle] (D) [vrtx=right/] at (3,0) {$\ell$};
\node[circle](C) [vrtx=right/] at (3, 1) {$k$};
\node[rectangle] (m1) [vrtx=right/] at (1.5, 2) {};
\node[rectangle] (m2) [vrtx=right/] at (1.5, 1) {};
\node[rectangle](m3) [vrtx=right/] at (1.5, 0) {};
\node[rectangle] (m4) [vrtx=right/] at (1.5, -1) {};

\path   (A) edge (m1)
        (m1) edge  (C)
        (A) edge (m2)
        (m2) edge (C)
        (A) edge (m4)
        (m4) edge (C)
        (B) edge (m3)
        (m3) edge (D)
        ;
\end{tikzpicture}
\caption{Glyph $\Glyph$}
\end{figure}

$\GraphMat_\Glyph$ can be written as a product of simpler matrices --- define matrices $\calL_1,\calL_2,\calL_3,\calL_4$ as follows.  For all $i,j,k,\ell$ distinct in $[n]$ with $i < j$ and $k < \ell$, 
\begin{align*}
    \calL_1[\{i,j\},\{i,j,k\}] \coloneqq \langle  M[i], M[k]\rangle\\
    \calL_2[\{i,j,k\},\{i,j,k\}] \coloneqq \langle  M[i],M[k]\rangle\\
    \calL_3[\{i,j,k\},\{j,k\}] \coloneqq \langle M[i], M[k]\rangle\\
    \calL_4[\{j,k\},\{k,\ell\}] \coloneqq \langle M[j], M[\ell]\rangle
\end{align*}
The above matrices are set to $0$ wherever they are undefined.  It can be verified that
\[
    \GraphMat_{\Glyph} = \calL_1\cdot\calL_2\cdot\calL_3\cdot\calL_4
\]

A major advantage of glyph factorization is that it offers a unified framework to bound the spectral norm of graphical matrices of the complex glyphs in terms of spectral norms of simpler matrices. In our example, we have
\[
    \left\|\GraphMat_{\Glyph}\right\| \le \left\|\calL_1\right\|\cdot\left\|\calL_2\right\|\cdot\left\|\calL_3\right\|\cdot\left\|\calL_4\right\|.
\]

We wrap up by giving spectral norm bounds on $\calL_i$, and we will generalize from them all the basic glyphs that we will use throughout this section.

\paragraph{Bounding $\left\|\calL_1\right\|$ and $\left\|\calL_3\right\|$.}  $\calL_1$ and $\calL_3$ are growth and shrinkage matrices respectively and hence their spectral norms are bounded by $\altwo$.

\paragraph{Bounding $\left\|\calL_2\right\|$.}  $\calL_2$ is a residue matrix and hence its spectral norm is at most $\amag$.

\paragraph{Bounding $\left\|\calL_4\right\|$.}  $\calL_4$ is a swap matrix and hence its spectral norm is at most $\aspec$.

Combining the above gives $\|\GraphMat_{\Glyph}\|\leq \altwo^2\cdot\amag\cdot\aspec$.  More generally:
\begin{lemma}   \label{lem:bound-from-factorization}
    Let $\Glyph$ be a well-glued glyph whose graphical matrix factorizes as $\GraphMat_{\Glyph} = \calL_1\cdot\ldots\cdot\calL_k$ where each $\calL_i$ is either a growth/shrinkage/swap residue matrix.  Let the number of growth/shrinkage matrices be $t_1$, the number of residue matrices be $t_2$, and the number of swap matrices be $t_3$, then
    \[ 
        \|\GraphMat_{\Glyph}\| \leq \altwo^{t_1} \cdot \amag^{t_2}\cdot \aspec^{t_3}.
    \]
\end{lemma}


\subsection{Spectral Norm Bounds on $\calE^{(1)}$}
\label{sec:E-1}
\begin{lemma}[$\calE^{(1)}$ has a small spectral norm] \[ \|\calE^{(1)}\|\leq O(\amag).\]
\end{lemma}\label{lem:E-1}
\begin{proof}
When $|S|=|T|=0$ by \pref{obs: empty-empty}, $\calE^{(1)}[S,T ] = 0$.  When $S=T=\{i\}\subseteq[n]$, $\E_{\bz\sim\{\pm 1\}^\HadDim} [p_{S,\HadDim}(\bz)p_{T,\HadDim}(\bz)]=1$ for all $\HadDim$ and hence,
\[ 
    \calE^{(1)}[S,T] = \lim_{\HadDim\to\infty} \E_{\bz\sim\{\pm1\}^\HadDim}[q_{\emptyset, \HadDim}(\bz)]-\E_{\bz\sim\{\pm1\}^\HadDim} [p_{S,\HadDim}(\bz)p_{T,\HadDim}(\bz)]=0
\]
Next we treat the case when $|S|=|T|=2$.  In particular, we are interested in bounding the spectral norm of $\calE^{(1)}$ restricted to entries indexed by $S,T$ of size $2$.  This matrix can be written as $\Id-L$ where $L$ is the diagonal matrix obtained by setting the $(\{i,j\},\{i,j\})$-entry to $\lim_{\HadDim\to\infty}\E_{\bz\sim \pm\{1 \}^\HadDim}[p_{\{i,j\},\HadDim}(\bz)^2]$.  $L$ can be written as a sum of graphical matrices of constantly many glyphs $\Glyph_1,\dots,\Glyph_V$ where $\Glyph_1$ is illustrated below, and the remaining glyphs have at least one length-$2$ acyclic walk.  This means $\GraphMat_{\Glyph_2},\dots,\GraphMat_{\Glyph_V}$ are diagonal matrices with entries bounded in magnitude by $\amag$.
\begin{figure}[h]
\centering
      \begin{subfigure}[h]{0.3\textwidth}\centering
\begin{tikzpicture}[
every edge/.style = {draw=black,very thick},
 vrtx/.style args = {#1/#2}{%
     draw, thick, fill=white,
      minimum size=5mm, label=#1:#2}]
       
\node[circle](B) [vrtx=left/$i$] at (0, 1) {};
\node[rectangle](C) [vrtx=right/] at (1.5,1) {};
\node[circle] (E) [vrtx=right/$j$] at (3, 0) {};
\node(D) [vrtx=left/] at (1.5,0) {};

\path  
        (B) edge[bend left] (C)
        (B) edge[bend right] (C)
        (E) edge[bend left] (D)
        (E) edge[bend right] (D)
        ;
\end{tikzpicture}
\end{subfigure}
\caption{Glyph $\Glyph_1$}
\end{figure}

Note that \pref{lem:graph-mat-explicit}, the graphical matrix of $\Glyph_1$ is the following diagonal matrix where
\[
    \GraphMat_{\Glyph_1}[\{i,j\},\{i,j\}]=\langle M[i], M[i]\rangle\langle M[j], M[j]\rangle=1
\]

Hence, for $S=T=\{i,j\}$, we have
\begin{align*}
    \calE^{(1)}[S,T] &= 1 - \sum_{t=1}^V \GraphMat_{\Glyph_t}[S,T] \\
    &= 1 - 1 - \sum_{t=2}^V \GraphMat_{\Glyph_t}[S,T]\\
    &\in [-(V-1)\amag, (V-1)\amag]
\end{align*}
Thus $\calE^{(1)}$ is a diagonal matrix with entries bounded by $O(\amag)$, the desired bound follows.
\end{proof}

\subsection{Spectral Norm Bounds on $\calE^{(2)}$}\label{sec:E-2}
\begin{lemma}[$\calE^{(2)}$ has a small spectral norm]\label{lem:E-2}
\[
    \|\calE^{(2)}\|\leq O(\altwo^2\cdot \amag).
\]
\end{lemma}
\begin{proof}
By \pref{lem:graph-mat-explicit}, for any $i\neq j\in [n]$,
\[ 
    \calE^{(2)}[i, j] = \lim_{\HadDim\to\infty}\E_{z\sim \{\pm1\}^n}\left[q_{\{i,j\},\HadDim}\right] - \E_{z\sim \{\pm1\}^n}\left[p_{\{i\},\HadDim}p_{\{j\},\HadDim}\right] = \frac{4}{3!} \langle M[i], M[j]\rangle^3
\]

We can thus express $\calE^{(2)}$ as a product of simpler matrices $\frac{4}{3!}\calL_1 \cdot \calL_2 \cdot \calL_3$ where
\begin{align*}
    \calL_1[\{i\},\{i,j\}] &\coloneqq \langle M[i], M[j]\rangle &\text{(a growth matrix)}\\
    \calL_2[\{i,j\},\{i,j\}] &\coloneqq \langle M[i], M[j]\rangle &\text{(a residue matrix)} \\
    \calL_3[\{i,j\},\{j\}] &\coloneqq \langle M[i], M[j]\rangle &\text{(a shrinkage matrix)}
\end{align*}
The desired bound then follows from \pref{lem:bound-from-factorization}.
\end{proof}
\subsection{Spectral Norm Bounds on $\calE^{(3)}$}\label{sec:E-3}

\begin{lemma}[$\calE^{(3)}$ has small spectral norm] \label{lem:E-3}
\[
    \|\calE^{(3)}\| \leq O(\amag\cdot(1+\aspec+\arow^2)).
\]
\end{lemma}
\begin{proof}
Note that $p_{\{i,j\},\HadDim}(\bz)$ can be written as a linear combination of graphical polynomials of the following half-glyphs from \pref{fig:half-rib-E3};  in particular
\[
    p_{\{i,j\},\HadDim} = \GraphPoly{A_{(i,j)},\HadDim} + \GraphPoly{B_{(i,j)},\HadDim} - 2\GraphPoly{C_{(i,j)},\HadDim} - 2\GraphPoly{C_{(j,i)},\HadDim} + 4\GraphPoly{D_{(i,j)},\HadDim} + 4\GraphPoly{E_{(i,j)},\HadDim}
\]


\begin{figure}[h]
\centering
      \begin{subfigure}[b]{0.3\textwidth}\centering
\begin{tikzpicture}[
every edge/.style = {draw=black,very thick},
 vrtx/.style args = {#1/#2}{%
   draw, thick, fill=white,
      minimum size=5mm, label=#1:#2}]
       
\node[circle](B) [vrtx=left/$i$] at (0, 0) {};
\node[rectangle](C) [vrtx=right/] at (1,0) {};
\node[circle](E) [vrtx=left/$j$] at (0, -1) {};
\path  
        (B) edge (C)
        (E) edge (C)
        ;
\end{tikzpicture}
\caption{$A_{(i,j)}$}
    \end{subfigure}
\hfil    
 \begin{subfigure}[b]{0.3\textwidth}\centering
\begin{tikzpicture}[
every edge/.style = {draw=black,very thick},
 vrtx/.style args = {#1/#2}{%
       draw, thick, fill=white,
      minimum size=5mm, label=#1:#2}]
       
\node[circle] (B) [vrtx=left/$i$] at (0, 0) {};
\node[rectangle] (C) [vrtx=right/] at (1,0) {};
\node(D) [vrtx=right/] at (1, -1) {};
\node[circle](E) [vrtx=left/$j$] at (0, -1) {};
\path  
        (B) edge (C)
        (E) edge (D)
        ;
\end{tikzpicture}
\caption{$B_{(i,j)}$}
    \end{subfigure}
\hfil    

    \begin{subfigure}[b]{0.3\textwidth}\centering
\begin{tikzpicture}[
every edge/.style = {draw=black,very thick},
 vrtx/.style args = {#1/#2}{%
       draw, thick, fill=white,
      minimum size=5mm, label=#1:#2}]
       
\node[rectangle](A) [vrtx=right/] at (1, 1) {};
\node[circle](B) [vrtx=left/$i$] at (0, 0) {};
\node(C) [vrtx=right/] at (1,0) {};
\node(D) [vrtx=right/] at (1, -1) {};
\node[circle](E) [vrtx=left/$j$] at (0, -1) {};
\path   (A) edge (B)
        (B) edge (C)
        (B) edge (D)
        (E) edge (D)
        ;
\end{tikzpicture}
\caption{$C_{(i,j)}$}
   \end{subfigure}
\hfil    

    \begin{subfigure}{0.3\textwidth}\centering
\begin{tikzpicture}[
every edge/.style = {draw=black,very thick},
 vrtx/.style args = {#1/#2}{%
 draw, thick, fill=white,
      minimum size=5mm, label=#1:#2}]
       
\node[rectangle](A) [vrtx=right/] at (1, 1) {};
\node[circle](B) [vrtx=left/$i$] at (0, 0) {};
\node(C) [vrtx=right/] at (1,0) {};
\node(D) [vrtx=right/] at (1, -1) {};
\node[circle](E) [vrtx=left/$j$] at (0, -1) {};
\node(F) [vrtx=right/] at (1, -2) {};

\path   (A) edge (B)
        (B) edge (C)
        (B) edge (D)
        (E) edge (D)
        (E) edge (C)
        (E) edge (F)
        ;
\end{tikzpicture}
\caption{$D_{(i,j)}$}
\end{subfigure}\hfil
\begin{subfigure}{0.3\textwidth}\centering
\begin{tikzpicture}[
every edge/.style = {draw=black,very thick},
 vrtx/.style args = {#1/#2}{%
       draw, thick, fill=white,
      minimum size=5mm, label=#1:#2}]
       
\node[circle] (B) [vrtx=left/$i$] at (0, 0) {};
\node[rectangle](C) [vrtx=right/] at (2,0) {};
\node(D) [vrtx=right/] at (2, -1) {};
\node[circle](E) [vrtx=left/$j$] at (0, -1) {};
\node(F) [vrtx=right/] at (2, -2) {};

\path   
        (B) edge (C)
        (B) edge (D)
        (E) edge (D)
        (E) edge (C)
        (E) edge (F)
        (B) edge (F)

        ;
\end{tikzpicture}
\caption{$E_{(i,j)}$}
\end{subfigure}
\caption{Half-Glyphs for $p_{\{i,j\},\HadDim}$}
\label{fig:half-rib-E3}
\end{figure}
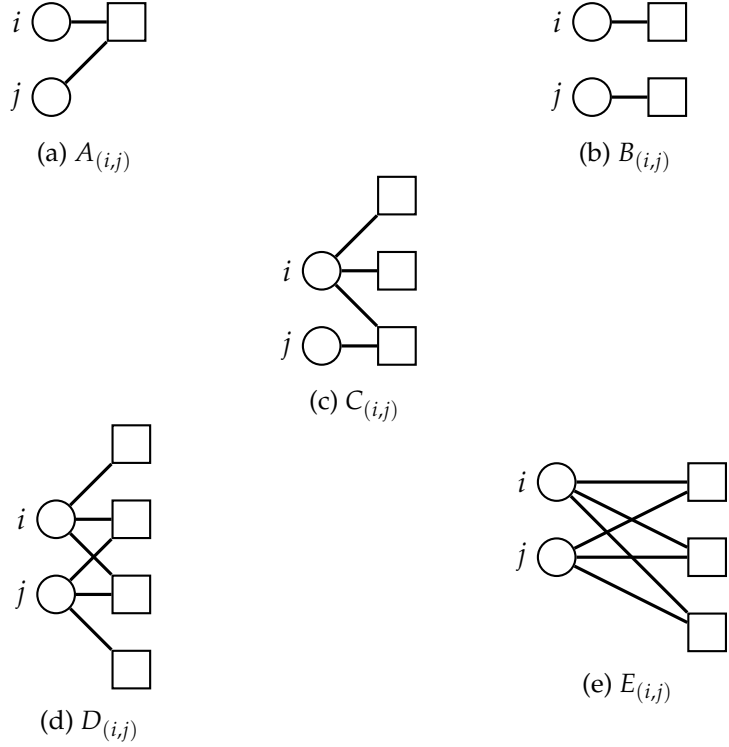
\noindent Let $\calH$ be the restriction of $\calM^{(2)}$ to the entries $(\{i,j\},\{j,k\})$ for $i,j,k$ distinct; then
\[
    \calH[\{i,j\},\{j,k\}] = \lim_{\HadDim\to\infty} \E_{\bz\sim\{\pm 1\}^{\HadDim}} [p_{\{i,j\},\HadDim}(\bz)p_{\{j,k\},\HadDim}(\bz)]
\]
Expanding out the above gives us an expression that is a sum of terms of the form
\[
    \lim_{\HadDim\to\infty}\E_{\bz\sim\{\pm1\}^{\HadDim}}[\GraphPoly{\srib_1,\HadDim}(\bz)\GraphPoly{\srib_2,\HadDim}(\bz)],
\]
which we denote $\srib_1\cdot\srib_2$ to make notation simpler.  Note that the $\srib_1\cdot\srib_2$ is $0$ if the number of odd-degree right vertices in $\srib_1$ and $\srib_2$ are not equal, and thus after discarding such pairs
\begin{align*}
    \calH[\{i,j\},\{j,k\} ] &= A_{(i,j)}\cdot A_{(j,k)} + B_{(i,j)}\cdot B_{(j,k)} +4C_{(i,j)}\cdot C_{(j,k)}+4C_{(j,i)}\cdot C_{(k,j)} + 16D_{(i,j)} \cdot D_{(j,k)} \\&+ 16 E_{(i,j)}\cdot E_{(j,k)})+ (4A_{(i,j)}\cdot E_{(j,k)}+ 4E_{(i,j)}\cdot A_{(j,k)}- 2B_{(i,j)}\cdot C_{(j,k)}- 2C_{(i,j)}\cdot B_{(j,k)}\\& - 2B_{(i,j)}\cdot C_{(k,j)} -2 C_{(i,j)}\cdot B_{(k,j)}  + 4B_{(i,j)}\cdot D_{(j,k)} +  4D_{(i,j)}\cdot B_{(j,k)} + 4C_{(i,j)}\cdot C_{(k,j)} \\&+ 4C_{(j,i)}\cdot C_{(j,k)} -8C_{(i,j)}\cdot D_{(j,k)} -8 D_{(i,j)}\cdot C_{(j,k)} -8C_{(j,i)}\cdot D_{(j,k)} -8 D_{(i,j)}\cdot C_{(k,j)} \numberthis \label{eq:E3-expansion}
\end{align*}
We can write $\calH$ as a sum of matrices $\calH_1+\calH_2+\dots+\calH_{20}$ where $\calH_t[\{i,j\},\{k,\ell\}]$ contains the $t$-th term of the above expression.  We alternatively use the notation $\calH_{M_1\cdot M_2}$ for $\calH_t$ where $cM_1\cdot M_2$ is the $t$-th term of the right hand side of \pref{eq:E3-expansion}, where $c\in\mathbb{R}$.

We introduce the \emph{symmetrized graphical matrix} $\wt{\GraphMat}_{\Glyph}$ of an $(A,B)$-glyph $\Glyph$ where $A = \{i,j\}$ and $B=\{j,k\}$.  Let $\Glyph_1,\Glyph_2,\Glyph_3,\Glyph_4$ be copies of $\Glyph$ with only ordering of left vertices changed such that they satisfy
\begin{align*}
    i\glyphLe[\Glyph_1]j\glyphLe[\Glyph_1]k\\
    k\glyphLe[\Glyph_2]j\glyphLe[\Glyph_2]i\\ 
    i\glyphLe[\Glyph_3]j,k\glyphLe[\Glyph_3]j\\
    j\glyphLe[\Glyph_4]i,j\glyphLe[\Glyph_4]k
\end{align*}
$\wt{\GraphMat}_{\Glyph}$ is then defined as
\[
    \GraphMat_{\Glyph_1}+\GraphMat_{\Glyph_2}+\GraphMat_{\Glyph_3}+\GraphMat_{\Glyph_4}.
\]
Note that if two glyphs are isomorphic, then they have the same symmetrized graphical matrix.  $\calH_t$ breaks further into a linear combination of symmetrized graphical matrices.  We use $\Glyphs(\calH)$ to refer to the collection of all glyphs that appear when each $\calH_t$ is written as a linear combination of symmetrized graphical matrices.  Symbolically,
\begin{align*}
    \calH &= \sum_{1\le t\le 20} \calH_i\\
    &= \sum_{1\le t\le 20} \sum_{\Glyph\in\Glyphs(\calH)} c_{t,\Glyph}\wt{\GraphMat}_{\Glyph}\\
    &= \sum_{\Glyph\in\Glyphs(\calH)}\left(\sum_{1\le t\le 20} c_{t,\Glyph}\right)\cdot\wt{\GraphMat}_{\Glyph}\\
    &= \sum_{\Glyph\in\Glyphs(\calH)} \alpha_{\Glyph}\cdot\wt{\GraphMat}_{\Glyph} \numberthis.\label{eq:graph-mat-decomp}
\end{align*}
where $\alpha_{\Glyph}\coloneqq\sum_{1\le t\le 20}c_{t,\Glyph}$.  We now enumerate over all glyphs in $\Glyphs(\calH)$, determine all $\alpha_{\Glyph}$, classify glyphs as ``ill-behaved'' or ``well-behaved'', and give bounds on the spectral norms of symmetrized graphical matrices of well-behaved glyphs.

We define a matrix $\calK$ whose rows are indexed by sets of size $2$, i.e., by $\{i,j\}$ for distinct $i,j$ and whose columns are indexed by \emph{ordered} tuples $(i,j)$ for distinct $i,j$, and its entries are defined as follows.
\[
    \forall i,j:~\calK[\{i,j\},(i,j)]=1,\calK[\{i,j\},(j,i)]=1,~\text{remaining entries are $0$.}
\]
A fact about $\calK$ we will need is that its spectral norm is $\sqrt{2}$.  This is a consequence of the fact that its rows are orthogonal and have $\ell_2$ norm equal to $\sqrt{2}$ each.




\begin{enumerate}
    \item Glyph $\calA^2$
    \[
        \wt{\GraphMat}_{\calA^2} \left[\{i,j \}, \{j,k\}\right] = \langle M[i], M[j] \rangle \langle M[j], M[k]\rangle
    \]
    $\calA^2$ appears in $\calH_{A_{(i,j)}\cdot A_{(j,k)}}$ and $\calH_{B_{(i,j)}\cdot B_{(j,k)}}$ with coefficient $1$ each.  Thus, $\alpha_{\calA^2} = 2$.  This glyph is ``ill-behaved''.

    \begin{figure}[h]
    \centering
\begin{tikzpicture}[
every edge/.style = {draw=black,very thick},
 vrtx/.style args = {#1/#2}{%
      draw, thick, fill=white,
      minimum size=5mm, label=#1:#2}]
       
\node[circle] (B) [vrtx=left/$i$] at (0, 1) {};
\node[rectangle] (C) [vrtx=/] at (1.5,1) {};
\node[circle] (E) [vrtx=left/$j$] at (0, 0) {};
\node[circle] (E') [vrtx=right/$j$] at (3, 1) {};

\node(D) [vrtx=below/] at (1.5,0) {};
\node[circle](K) [vrtx=right/$k$] at (3, 0) {};

\path  
        (B) edge(C)
        (E') edge (C)
        (E) edge (D)
        (K) edge (D)
        ;
\end{tikzpicture}
\caption{Glyph $\calA^2$}
    \end{figure}

    \item Glyph $\calB^2$
    \[
        \wt{\GraphMat}_{\calB^2}  \left[\{i,j\}, \{j,k\}\right] = \langle M[i], M[k] \rangle \langle M[j], M[j]\rangle 
    \]
    $\calB^2$ appears in $\calH_{B_{(i,j)}\cdot B_{(j,k)}}$ with coefficient $1$.  Thus, $\alpha_{\calB^2}=1$.  This glyph is ``ill-behaved''.

    \begin{figure}[h]
        \centering
        
\begin{tikzpicture}[
every edge/.style = {draw=black,very thick},
 vrtx/.style args = {#1/#2}{%
   draw, thick, fill=white,
      minimum size=5mm, label=#1:#2}]
       
\node[circle](B) [vrtx=left/$i$] at (0, .5) {};
\node[rectangle] (C) [vrtx=above/] at (1.5,.5) {};
\node[circle](E) [vrtx=left/$j$] at (0, -.5) {};
\node(D) [vrtx=below/] at (1.5,-.5) {};
\node[circle](K) [vrtx=right/$k$] at (3, -0.5){} ;
\node[circle](E')  [vrtx=right/$j$] at (3, 0.5) {};

\path  
        (B) edge(C)
        (C) edge (K)
        (E) edge (D)
        (E') edge (D)
        ;
\end{tikzpicture}
\caption{Glyph $\calB^2$}
    \end{figure}{}

    
    \item Glyph $\calC_1^2$
    \[
        \wt{\GraphMat}_{C_1^2} \left[\{i,j\}, \{j,k\}\right] = \frac{1}{2!} \langle M[i], M[k]\rangle^2\langle M[i], M[j]\rangle\langle M[j],M[k]\rangle
    \]
    $\calC_1^2$ appears in $\calH_{C_{(i,j)}\cdot C_{(k,j)}}$ with coefficient $4$.  Thus, $\alpha_{\calC_1^2}=4$.
    This glyph is ``well-behaved'' and we can prove
    \[
        \|\wt{\GraphMat}_{\calC^2}\|\leq\altwo^2\cdot\amag^2
    \]
    via the following factorization combined with \pref{lem:bound-from-factorization}:
    \[
        \wt{\GraphMat}_{\calC_1^2}=\frac{1}{2!}\cdot\calK \cdot \calL_1 \cdot \calL_2 \cdot \calL_3 \cdot \calK^{\dagger}
    \]
    where for all $i,j,k$ distinct
    \begin{align*}
        \calL_1[(i,j),(i,j)] &\coloneqq \langle M[i], M[j]\rangle &\text{(a residue matrix)} \\
        \calL_2[(i,j),(i,j,k)]&\coloneqq \langle M[i],M[k]\rangle &\text{(a growth matrix)} \\
        \calL_3[(i,j,k),(j,k)]&\coloneqq \langle M[i],M[k]\rangle &\text{(a shrinkage matrix)}\\
        \calL_4[(j,k),(j,k)]&\coloneqq \langle M[j],M[k]\rangle &\text{(a residue matrix)}
    \end{align*}
\begin{figure}[h]
    \centering
    
\begin{tikzpicture}[
every edge/.style = {draw=black,very thick},
 vrtx/.style args = {#1/#2}{%
       draw, thick, fill=white,
      minimum size=2mm, label=#1:#2}]

\node[circle] (I) [vrtx=left/$i$] at (0, 0) {};
\node[circle](J) [vrtx=left/$j$] at (0, -1) {};
\node[circle](J') [vrtx=right/$j$] at (3, 0) {};

\node[circle] (K) [vrtx=right/$k$] at (3, -1) {};

\node[rectangle] (a) [vrtx=above/] at (1.5,1.5){};
\node (b) [vrtx=above/] at (1.5, 0.5) {};
\node (c) [vrtx=above/] at (1.5,-.5) {};ß
\node (d) [vrtx=below/] at (1.5,-1.5) {};
\path  
    (I) edge (a)
    (a) edge (K)
    (K) edge (b)
    (b) edge (I)
    (J') edge (c)
    (c) edge (I)
    (J) edge (d)
    (d) edge (K)
    ;
\end{tikzpicture}
\caption{Glyph $\calC_1^2$}
\end{figure}

    \item Glyph $\calC_2^2$
    \[
        \wt{\GraphMat}_{\calC_2^2} \left[\{i,j\}, \{j,k\}\right] = \frac{1}{2!}\langle M[i],M[j]\rangle\langle M[j],M[j]\rangle^2\langle M[j],M[k]\rangle
    \]
    $\calC_2^2$ appears in $\calH_{C_{(j,i)}\cdot C_{(j,k)}}$ with coefficient $4$.  Thus, $\alpha_{\calC_2^2}=4$.  This glyph is ``ill-behaved''.

    \begin{figure}[h]
        \centering
        
\begin{tikzpicture}[
every edge/.style = {draw=black,very thick},
 vrtx/.style args = {#1/#2}{%
      draw, thick, fill=white,
      minimum size=2mm, label=#1:#2}]

\node[circle] (I) [vrtx=left/$i$] at (0, 0) {};
\node[circle](J) [vrtx=left/$j$] at (0, -1) {};
\node[circle](J') [vrtx=right/$j$] at (3, 0) {};
\node[circle] (K) [vrtx=right/$k$] at (3, -1) {};
\node[rectangle] (a) [vrtx=above/] at (1.5,1.5){};
\node (b) [vrtx=above/] at (1.5, 0.5) {};
\node (c) [vrtx=above/] at (1.5,-.5) {};
\node (d) [vrtx=below/] at (1.5,-1.5) {};

\path  
        (I) edge (a)
    (a) edge (J')
    (J') edge (b)
    (b) edge (J)
    (J') edge (c)
    (c) edge (J)
    (J) edge (d)
    (d) edge (K)
        ;
\end{tikzpicture}
\caption{Glyph $\calC_2^2$}
   \end{figure}

    
    \item Glyph $\calD^2_1 $\\
    \[
        \wt{\GraphMat}_{\calD^2_1} \left[\{i,j\}, \{j,k\}\right] = \frac{1}{2!\cdot 2!}\langle M[i], M[j]\rangle^2\langle M[j], M[k]\rangle^2\langle M[i], M[k]\rangle \langle M[j],M[j]\rangle
    \]
    $\calD^2_1$ appears in $\calH_{D_{(i,j)}\cdot D_{(j,k)}}$ with coefficient $16$.  Thus, $\alpha_{\calD^2_1}=16$.
    This glyph is ``well-behaved'' and we can prove
    \[
        \|\wt{\GraphMat}_{D_1^2} \| \leq \frac{1}{2!} \amag^4\cdot\aspec
    \]
    via the following factorization combined with \pref{lem:bound-from-factorization}:
    \[ 
        \wt{\GraphMat}_{D_1^2} = \frac{1}{2!}\cdot\frac{1}{2!}\cdot\calK\cdot\calL_1\cdot \calL_2 \cdot \calL_3\cdot\calL_4\cdot\calK^{\dagger}
    \]
    where for $i,j,k$ distinct
    \begin{align*}
        \calL_1[(i,j),(i,j)]&\coloneqq \langle M[i],M[j]\rangle^2 &\text{(two residue matrices )} \\
        \calL_2[(i,j),(j,k)]&\coloneqq  \langle M[i],M[k]\rangle &\text{(a swap matrix) } \\
        \calL_3[(j,k),(j,k)]&\coloneqq \langle M[j],M[j]\rangle &\text{(an identity matrix)}\\
        \calL_4[(j,k),(j,k)]&\coloneqq \langle M[j],M[k]\rangle^2 &\text{(two residue matrices)}
    \end{align*}
 \begin{figure}[h]
     \centering
      \begin{subfigure}[h]{0.2\textwidth}
\begin{tikzpicture}[
every edge/.style = {draw=black,very thick},
 vrtx/.style args = {#1/#2}{%
       draw, thick, fill=white,
      minimum size=2mm, label=#1:#2}]

\node[circle] (I) [vrtx=left/$i$] at (0, 0) {};
\node[circle](J) [vrtx=left/$j$] at (0, -1) {};
\node[circle] (K) [vrtx=right/$k$] at (3, -1) {};
\node[circle](J') [vrtx=right/$j$] at (3, -2) {};

\node[rectangle] (a) [vrtx=above/] at (1.5,1.5){};
\node (b) [vrtx=above/] at (1.5, 0.5) {};
\node (c) [vrtx=above/] at (1.5,-.5) {};
\node (d) [vrtx=below/] at (1.5,-1.5) {};
\node (e)  [vrtx=below/] at (1.5,-2.5) {};
\node (f)  [vrtx=below/] at (1.5,-3.5) {};

\path  
    (I) edge (b)
    (b) edge (J)
    (I) edge (c)
    (c) edge (J)
    (K) edge (d)
    (d) edge (J')
    (K) edge (e)
    (e) edge (J')
    (I) edge (a)
    (a) edge (K)
    (J) edge (f)
    (f) edge (J')
    
    ;
\end{tikzpicture}
\subcaption{Glyph $\calD^2_1$}
    \end{subfigure} \hfil
 \begin{subfigure}[h]{0.2\textwidth}
\begin{tikzpicture}[
every edge/.style = {draw=black,very thick},
 vrtx/.style args = {#1/#2}{%
       draw, thick, fill=white,
      minimum size=2mm, label=#1:#2}]

\node[circle] (I) [vrtx=left/$i$] at (0, 0) {};
\node[circle](J) [vrtx=left/$j$] at (0, -1) {};
\node[circle] (K) [vrtx=right/$j$] at (3, -1) {};
\node[circle](J') [vrtx=right/$k$] at (3, -2) {};

\node[rectangle] (a) [vrtx=above/] at (1.5,1.5){};
\node (b) [vrtx=above/] at (1.5, 0.5) {};
\node (c) [vrtx=above/] at (1.5,-.5) {};
\node (d) [vrtx=below/] at (1.5,-1.5) {};
\node (e)  [vrtx=below/] at (1.5,-2.5) {};
\node (f)  [vrtx=below/] at (1.5,-3.5) {};

\path  
    (I) edge (b)
    (b) edge (J)
    (I) edge (c)
    (c) edge (J)
    (K) edge (d)
    (d) edge (J')
    (K) edge (e)
    (e) edge (J')
    (I) edge (a)
    (a) edge (K)
    (J) edge (f)
    (f) edge (J')
    
    ;
\end{tikzpicture}
\subcaption{Glyph $\calD^2_2$}
    \end{subfigure}
 \end{figure}  
  
    \item Glyph $\calD^2_2 $
    \[
        \wt{\GraphMat}_{\calD^2_2} \left[\{i,j\}, \{j,k\}\right] = \frac{2}{3!\cdot 3!}\langle M[i], M[j]\rangle^3\langle M[j], M[k]\rangle^3
    \]
    $\calD^2_2$ appears in $\calH_{D_{(i,j)}\cdot D_{(j,k)}}$ and $\calH_{E_{(i,j)}\cdot E_{(j,k)}}$ with coefficient $16$ each.  Thus, $\alpha_{\calD^2_2}=32$.  This glyph is ``well-behaved'' and we can prove
    \[
        \|\wt{\GraphMat}_{\calD^2_2}\| \leq \altwo^2\cdot \amag^4
    \]
    via the following factorization combined with \pref{lem:bound-from-factorization}:
    \[
        \GraphMat_{D_2^2} = \frac{1}{3!\cdot 3!} \cdot \calK \cdot \calL_1\cdot \calL_2\cdot \calL_3\cdot\calL_4 \cdot \calK^{\dagger}
    \]
    where for $i,j,k$ distinct
    \begin{align*}
        \calL_1[(i,j),(i,j)]&\coloneqq \langle M[i],M[j]\rangle^2 &\text{(two residue matrices)} \\
        \calL_2[(i,j),(j)]&\coloneqq \langle M[i],M[j]\rangle &\text{(a shrinkage matrix)} \\
        \calL_3[(j),(j,k)]&\coloneqq \langle M[j],M[k]\rangle &\text{(a growth matrix) } \\
        \calL_4[(j,k),(j,k)]&\coloneqq \langle M[j],M[k]\rangle^2 &\text{(two residue matrices)}
    \end{align*}

    \item Glyph $\calB\calC_1$
    \[
        \wt{\GraphMat}_{\calB\calC_1}\left[\{i,j\}, \{j,k\}\right] = \langle M[i],M[j]\rangle\langle M[j], M[j]\rangle\langle M[j],M[k]\rangle
    \]
    $\calB\calC_1$ appears in $\calH_{B_{(i,j)}\cdot C_{(j,k)}}$ and $\calH_{C_{(j,i)}\cdot B_{(j,k)}}$ with coefficient $-2$ each.  Thus $\alpha_{\calB\calC_1}=-4$.  This glyph is ``ill-behaved''.
    \begin{figure}[h]
        \centering
        \begin{tikzpicture}[
every edge/.style = {draw=black,very thick},
 vrtx/.style args = {#1/#2}{%
       draw, thick, fill=white,
      minimum size=5mm, label=#1:#2}]

\node[circle] (I) [vrtx=left/$i$] at (0, 0) {};
\node[circle](J) [vrtx=left/$j$] at (0, -1) {};
\node[circle](J') [vrtx=right/$j$] at (3, 0) {};
\node[circle] (K) [vrtx=right/$k$] at (3, -1) {};
\node[rectangle] (a) [vrtx=above/] at (1.5,0.5){};
\node (b) [vrtx=above/] at (1.5, -0.5) {};
\node (d) [vrtx=below/] at (1.5,-1.5) {};

\path  
        (I) edge (a)
    (a) edge (J')
    (J') edge (b)
    (b) edge (J)
    (J) edge (d)
    (d) edge (K)
        ;
\end{tikzpicture}
        \caption{Glyph $\calB\calC_1$}
    \end{figure}

    \item Glyphs $\calB\calC_2$ and $\calB\calC_3$
    \begin{align*}
        \wt{\GraphMat}_{\calB\calC_2}\left[\{i,j\}, \{j,k\}\right] &= \langle M[i], M[k]\rangle \langle M[j], M[k]\rangle^2\\
        \wt{\GraphMat}_{\calB\calC_3}\left[\{i,j\}, \{j,k\}\right] &= \langle M[i], M[k]\rangle \langle M[j], M[k]\rangle^2.
    \end{align*}
    $\calB\calC_2$ appears in $\calH_{B_{(i,j)}\cdot C_{(k,j)}}$ with coefficient $-2$ and $\calB\calC_3$ appears in $\calH_{C_{(i,j)}B_{(j,k)}}$ with coefficient $-2$.  Thus, $\alpha_{\calB\calC_2}=\alpha_{\calB\calC_3}=-2$.  These glyphs are ``well-behaved'' and we can prove
    \[
        \|\wt{\GraphMat}_{\calB\calC_2}\|\le 2\aspec\cdot\amag \qquad \|\wt{\GraphMat}_{\calB\calC_3}\|\le 2\aspec\cdot\amag.
    \]
    We do so by illustrating a factorization of $\wt{\GraphMat}_{\calB\calC_2}$; $\wt{\GraphMat}_{\calB\calC_3}$ can be factorized in an identical way.
    \[ 
        \GraphMat_{\calB\calC_2} = \calK\cdot\calL_1\cdot \calL_2\cdot\calK^{\dagger}
    \]
    where for $i,j,k$ distinct
    \begin{align*}
        \calL_1[(i,j),(j,k)]&\coloneqq  \langle M[i],M[k]\rangle &\text{(a swap matrix)}  \\
        \calL_2[(j,k),(j,k)]&\coloneqq \langle M[j],M[k]\rangle^2 &\text{(two residue matrices)}
    \end{align*}
    
    \begin{figure}[h]
        \centering
        \begin{tikzpicture}[
every edge/.style = {draw=black,very thick},
 vrtx/.style args = {#1/#2}{%
      draw, thick, fill=white,
      minimum size=5mm, label=#1:#2}]
     
\node[circle] (I) [vrtx=left/$i$ (or $k$)] at (0, 0) {};
\node[circle](J) [vrtx=left/$j$] at (0, -1) {};
\node[circle](J') [vrtx=right/$j$] at (3, 0) {};
\node[circle] (K) [vrtx=right/$k$ (or $i$)] at (3, -1) {};
\node[rectangle] (a) [vrtx=above/] at (1.5,0.5){};
\node (c) [vrtx=above/] at (1.5,-.5) {};
\node (d) [vrtx=below/] at (1.5,-1.5) {};

\path  
        (I) edge (a)
    (a) edge (K)
   
    (J') edge (c)
    (c) edge (K)
    (J) edge (d)
    (d) edge (K)
        ;
\end{tikzpicture}
        \caption{Glyph $\calB\calC_2$ (or $\calB\calC_3$)}
    \end{figure}
    \item Glyph $\calB\calD_1$
    \[
    \wt{\GraphMat}_{\calB\calD_1}[\{i,j\}, \{j,k\}] = \frac{1}{3!} \langle M[i], M[j]\rangle^3 \langle M[j], M[k]\rangle
    \]
    $\calB\calD_1$ appears in $\calH_{B_{i,j}\cdot D_{j,k}}$ with coefficient $4$.  Thus, $\alpha_{\calB\calD_1}=4$.  This glyph is ``well-behaved'' and we can prove
    \[
        \|\wt{\GraphMat}_{\calB\calD_1}\|\le \frac{2}{3!}\amag^2\cdot\arow^2.
    \]
    via the following factorization combined with \pref{lem:bound-from-factorization}: 
    \[
        \wt{\GraphMat}_{\calB\calD_1} = \calK \cdot \calL_1 \cdot \calL_2 \cdot \calL_3 \cdot \calK^{\dagger}
    \]
    where for $i,j,k$ distinct
\begin{align*}
    \calL_1[(i,j),(i,j)]&\coloneqq \langle M[i], M[j]\rangle^2 & \text{(two residue matrices)}\\
    \calL_2[(i,j),(j)] &\coloneqq \langle M[i], M[j] \rangle & \text{(a shrinkage matrix)}\\
    \calL_3[(j),(j,k)] &\coloneqq \langle M[j], M[k] \rangle &\text{(a growth matrix)}
\end{align*}
    
    \begin{figure}[h]
        \centering
        \begin{tikzpicture}[
every edge/.style = {draw=black,very thick},
 vrtx/.style args = {#1/#2}{%
       draw, thick, fill=white,
      minimum size=5mm, label=#1:#2}]
     
\node[circle] (I) [vrtx=left/$i$] at (0, 0) {};
\node[circle](J) [vrtx=left/$j$] at (0, -1) {};
\node[circle](J') [vrtx=right/$j$] at (3, 0) {};
\node[circle] (K) [vrtx=right/$k$] at (3, -1) {};
\node[rectangle] (a) [vrtx=above/] at (1.5,1.5){};
\node (b) [vrtx=above/] at (1.5, 0.5) {};
\node (c) [vrtx=above/] at (1.5,-.5) {};
\node (d) [vrtx=below/] at (1.5,-1.5) {};

\path  
        (I) edge (a)
    (a) edge (J')
    (I) edge (b)
    (b) edge (J)
    (c) edge (I)
    (c) edge (J)
    (d) edge (J)
    (d) edge (K)
        ;
\end{tikzpicture}
        \caption{Glyph $\calB\calD_1$}
    \end{figure}
    
    \item Glyph $\calB\calD_2$
    \[
        \wt{\GraphMat}_{\calB\calD_2}[\{i,j\},\{j,k\}] = \frac{1}{2!} \langle M[i], M[j]\rangle^2\langle M[i], M[k]\rangle \langle M[j], M[j]\rangle
    \]
    $\calB\calD_2$ appears in $\calH_{B_{i,j} \cdot D_{j,k}}$ with coefficient $4$.  Thus, $\alpha_{\calB\calD_2} = 4$.  This glyph is ``well-behaved'' and we can prove
    \[
        \|\wt{\GraphMat}_{\calB\calD_2}\| \le \amag^2\cdot\aspec.
    \]
    via the following factorization combined with \pref{lem:bound-from-factorization}:
    \[ 
        \wt{\GraphMat}_{\calB\calD_2} = \calK \cdot \calL_1 \cdot \calL_2 \cdot \calL_3 \cdot \calK^{\dagger}
    \]
    where for $i,j,k$ distinct 
    \begin{align*}
        \calL_1[(i,j),(i,j)] &\coloneqq \langle M[i], M[j]\rangle^2 &\text{(two residue matrices) }\\
        \calL_2[(i,j),(j,k)] &\coloneqq \langle M[i], M[k]\rangle &\text{(a swap matrix) } \\
        \calL_3[(j,k),(j,k)] &\coloneqq \langle M[j], M[j] \rangle &\text{(an identity matrix)}
    \end{align*}
    
    \begin{figure}[h]
        \centering
        \begin{tikzpicture}[
every edge/.style = {draw=black,very thick},
 vrtx/.style args = {#1/#2}{%
       draw, thick, fill=white,
      minimum size=5mm, label=#1:#2}]
     
\node[circle] (I) [vrtx=left/$i$] at (0, 0) {};
\node[circle](J) [vrtx=left/$j$] at (0, -1) {};
\node[circle](J') [vrtx=right/$j$] at (3, 0) {};
\node[circle] (K) [vrtx=right/$k$] at (3, -1) {};
\node[rectangle] (a) [vrtx=above/] at (1.5,1.5){};
\node (b) [vrtx=above/] at (1.5, 0.5) {};
\node (c) [vrtx=above/] at (1.5,-.5) {};
\node (d) [vrtx=below/] at (1.5,-1.5) {};

\path  
        (I) edge (a)
    (a) edge (K)
    (I) edge (b)
    (b) edge (J)
    (c) edge (I)
    (c) edge (J)
    (d) edge (J)
    (d) edge (J')
;        
\end{tikzpicture}
        \caption{Glyph $\calB\calD_2$}
    \end{figure}
    
    \item Glyph $\calC\calD_1$
    \[
        \wt{\GraphMat}_{\calC\calD_1}\left[\{i,j\}, \{j,k\}\right] = \langle M[i], M[j]\rangle^2\langle M[j], M[k]\rangle^2\langle M[i], M[k]\rangle.
    \]
    $\calC\calD_1$ appears in $\calH_{C_{(i,j)}\cdot D_{(j,k)}}$ and $\calH_{D_{(i,j)}\cdot C_{(j,k)}}$ with coefficient $-8$ each.  Thus, $\alpha_{\calC\calD_1}=-16$.  This glyph is ``well-behaved'' and we can prove
    \[
        \|\wt{\GraphMat}_{\calC\calD_1}\| \le 2\aspec\cdot\amag^4.
    \]
    via the following factorization combined with \pref{lem:bound-from-factorization}:
    \[ 
        \wt{\GraphMat}_{\calC\calD_1} = \calK\cdot\calL_1\cdot\calL_2\cdot\calL_3\cdot\calK^{\dagger}
    \]
    where for $i,j,k$ distinct
    \begin{align*}
        \calL_1[(i,j),(i,j)]&\coloneqq\langle M[i],M[j]\rangle^2 &\text{(two residue matrices) }\\
        \calL_2[(i,j),(j,k)]&\coloneqq \langle M[i],M[k]\rangle &\text{(a swap matrix)} \\
        \calL_3[(j,k),(j,k)]&\coloneqq \langle M[j],M[k]\rangle^2 &\text{ (two residue matrices)}
    \end{align*}

    \begin{figure}[h]
        \centering
        \begin{tikzpicture}[
every edge/.style = {draw=black,very thick},
 vrtx/.style args = {#1/#2}{%
      draw, thick, fill=white,
      minimum size=2mm, label=#1:#2}]

\node[circle] (I) [vrtx=left/$i$] at (0, 0) {};
\node[circle](J) [vrtx=left/$j$] at (0, -1) {};
\node[circle] (K) [vrtx=right/$k$] at (3, -1) {};
\node[circle](J') [vrtx=right/$j$] at (3, -2) {};

\node[rectangle] (a) [vrtx=above/] at (1.5,1.5){};
\node (b) [vrtx=above/] at (1.5, 0.5) {};
\node (c) [vrtx=above/] at (1.5,-.5) {};
\node (d) [vrtx=below/] at (1.5,-1.5) {};
\node (e)  [vrtx=below/] at (1.5,-2.5) {};

\path  
    (I) edge (b)
    (b) edge (J)
    (I) edge (c)
    (c) edge (J)
    (K) edge (d)
    (d) edge (J')
    (K) edge (e)
    (e) edge (J')
    (I) edge (a)
    (a) edge (K)
    
    ;
\end{tikzpicture}
        \caption{Glyph $\calC\calD_1$}
    \end{figure}
    
    \item  Glyphs $\calC\calD_2$ and $\calC\calD_3$
    \begin{align*}
        \wt{\GraphMat}_{\calC\calD_2}[\{i,j\},\{j,k\}] &= \frac{1}{3!}\langle M[i],M[j]\rangle\langle M[j],M[j]\rangle \langle M[j],M[k]\rangle^3\\
        \wt{\GraphMat}_{\calC\calD_3}[\{i,j\}, \{j,k\}] &= \frac{1}{3!}\langle M[i],M[j]\rangle\langle M[j],M[j]\rangle \langle M[i],M[j]\rangle^3
    \end{align*}
    
    $\calC\calD_2$ appears in $\calH_{C_{(j,i)}\cdot D_{(j,k)}}$ with coefficient $-8$ and $\calC\calD_3$ appears in $\calH_{D_{(i,j)}\cdot C_{(k,j)}}$ with coefficient $-8$.  Thus, $\alpha_{\calC\calD_2} = \alpha_{\calC\calD_3} = -8$.  These glyphs are ``well-behaved'' and we can prove
    \[
        \|\wt{\GraphMat}_{\calC\calD_2}\| \leq \frac{2}{3!}\altwo^2\cdot \amag^2 \qquad \|\wt{\GraphMat}_{\calC\calD_3} \| \leq \frac{2}{3!}\altwo^2\cdot \amag^2.
    \]
    We do so by giving a factorization of $\wt{\GraphMat}_{\calC\calD_2}$ and applying \pref{lem:bound-from-factorization}; an identical factorization applies to $\wt{\GraphMat}_{\calC\calD_3}$.
    \[
        \wt{\GraphMat}_{\calC\calD_2}=\frac{1}{3!}\cdot\calK\cdot\calL_1\cdot \calL_2\cdot\calL_3\cdot\calK^{\dagger}
    \]
    where for $i,j,k$ distinct
    \begin{align*}
        \calL_1[(i,j),(i,j)]&\coloneqq \langle M[j],M[j]\rangle &\text{(an identity matrix) } \\
        \calL_2[(i,j),(j)]&\coloneqq \langle M[i],M[j]\rangle &\text{(a shrinkage matrix) }  \\
        \calL_3[(j),(j,k)]&\coloneqq 
        \langle M[j],M[k]\rangle &\text{(a growth matrix) } \\
        \calL_4[(j,k),(j,k)]&\coloneqq \langle M[j],M[k]\rangle^2 &\text{ (two residue matrices)}
    \end{align*}
\begin{figure}[h]
    \centering
     \begin{tikzpicture}[
every edge/.style = {draw=black,very thick},
 vrtx/.style args = {#1/#2}{%
       draw, thick, fill=white,
      minimum size=5mm, label=#1:#2}]
  
\node[circle] (I) [vrtx=left/$i$ (or $k$)] at (0, .5) {};
\node[circle](J) [vrtx=left/$j$] at (0, -1) {};
\node[circle](J') [vrtx=right/$j$] at (3, 0) {};
\node[circle] (K) [vrtx=right/$k$ (or $i$) ] at (3, -1) {};
\node[rectangle] (a) [vrtx=above/] at (1.5,1.5){};
\node (b) [vrtx=above/] at (1.5, 0.5) {};
\node (c) [vrtx=above/] at (1.5,-.5) {};
\node (d) [vrtx=below/] at (1.5,-1.5) {};
\node (e) [vrtx=below/] at (1.5,-2.5) {};

\path  
        (I) edge (a)
    (a) edge (J)
    (J') edge (b)
    (b) edge (J)
    (J') edge (c)
    (c) edge (K)
    (J') edge (d)
    (d) edge (K)
    (J) edge (e)
    (e) edge (K)
        ;
\end{tikzpicture}
    \caption{Glyph ${\calC\calD_2}$ (or ${\calC\calD_3}$ )}
\end{figure}
\end{enumerate}
\noindent Let $\wt{\calH}$ be the restriction of $\calM^{(1)}$ to the entries $(\{i,j\},\{j,k\})$ for $i,j,k$ distinct;  then
\[
    \wt{\calH}[\{i,j\},\{j,k\}] = \lim_{\HadDim\to\infty} \E_{\bz\sim\{\pm1\}^{\HadDim}}[q_{\{i,j\}}(\bz)]=\langle M[i], M[k]\rangle + \frac{4}{3!}\langle M[i], M[k]\rangle^3.
\]
Note that $\calE^{(3)} = \wt{\calH} - \calH$ and so from \pref{eq:graph-mat-decomp} we can write
\begin{align*}
    \calE^{(3)} \coloneqq -\sum_{\substack{\Glyph\in\Glyphs(\calH):\\ \Glyph~\text{well-behaved}}} \alpha_{\Glyph}\cdot\wt{\GraphMat}_{\Glyph} + \left(\wt{\calH} - \sum_{\substack{\Glyph\in\Glyphs(\calH):\\
    \Glyph~\text{ill-behaved}}} \alpha_{\Glyph}\cdot\wt{\GraphMat}_{\Glyph}\right)
\end{align*}
Then
\begin{align}   \label{eq:norm-bound-E3}
    \|\calE^{(3)}\| \le \sum_{\substack{\Glyph\in\Glyphs(\calH):\\ \Glyph~\text{well-behaved}}} |\alpha_{\Glyph}|\cdot\|\wt{\GraphMat}_{\Glyph}\| + \left\|\wt{\calH} - \sum_{\substack{\Glyph\in\Glyphs(\calH):\\
    \Glyph~\text{ill-behaved}}} \alpha_{\Glyph}\cdot\wt{\GraphMat}_{\Glyph}\right\|
\end{align}
Since $\amag\le 1$, the first term is at most $O(\amag\cdot(1+\aspec+\arow^2))$.  We call the second term as
\[
    \calE^{(3)}_{\sparse}\coloneqq\wt{\calH} - \sum_{\substack{\Glyph\in\Glyphs(\calH):\\
    \Glyph~\text{ill-behaved}}} \alpha_{\Glyph}\cdot\wt{\GraphMat}_{\Glyph}
\]
and
\[
    \sum_{\substack{\Glyph\in\Glyphs(\calH):\\
    \Glyph~\text{ill-behaved}}} \alpha_{\Glyph}\cdot\wt{\GraphMat}_{\Glyph} = 2\wt{\GraphMat}_{\calA^2} + \wt{\GraphMat}_{\calB^2} + 4\wt{\GraphMat}_{\calC_2^2} - 4\wt{\GraphMat}_{\calB\calC_1}.
\]
Now we're ready to bound $\|\calE^{(3)}_{\sparse}\|$.
\begin{align*}
    \calE^{(3)}_{\sparse}[\{i,j\},\{j,k\}] &= \langle M[i],M[k]\rangle + \frac{4}{3!} \langle M[i],M[k]\rangle^3 - 2\langle M[i], M[j]\rangle\langle M[j], M[k]\rangle -\\
    &\langle M[i], M[k]\rangle\langle M[j], M[j]\rangle - \frac{4}{2!}\langle M[i], M[j]\rangle\langle M[j], M[k]\rangle\langle M[j], M[j]\rangle^2\\
    &+4\langle M[i], M[j]\rangle\langle M[j], M[k]\rangle\langle M[j], M[j]\rangle.
\end{align*}
Since $\langle M[j], M[j]\rangle = 1$,
\[
    \calE^{(3)}_{\sparse}[\{i,j\},\{j,k\}] = \frac{4}{3!}\langle M[i], M[k]\rangle^3.
\]
We can factorize $\calE^{(3)}_{\sparse}$ as
\[
    \calE^{(3)}_{\sparse} = \frac{4}{3!}\calK\cdot\calL_1\cdot\calL_2\cdot\calL_3\cdot\calK^{\dagger}
\]
where for distinct $i,j,k$
\begin{align*}
    \calL_1[(i,j),(i,j,k)] &= \langle M_i, M_k\rangle &\text{(a growth matrix)}\\
    \calL_2[(i,j,k),(i,j,k)] &= \langle M_i, M_k\rangle &\text{(a residue matrix)}\\
    \calL_3[(i,j,k),(j,k)] &= \langle M_i, M_k\rangle &\text{(a shrinkage matrix)}
\end{align*}
and hence
\[
    \|\calE^{(3)}_{\sparse}\| \le O(\arow^2\cdot\amag).
\]
Plugging the above bound back in to \pref{eq:norm-bound-E3} proves
\[
    \|\calE^{(3)}\| \le O(\amag\cdot(1+\aspec+\arow^2)).
\]
\end{proof}

\subsection{Spectral Norm Bounds on $\calE^{(4)}$}\label{sec:E-4}
Throughout this section, $i,j,k,\ell$ are distinct elements of $[n]$ such that $i < j$ and $k < \ell$.
\begin{lemma}[$\calE^{(4)}$ has a small spectral norm]  \label{lem:E-4}
\[
    \|\calE^{(4)} \| \leq O(\amag\cdot(1+\altwo^4)\cdot(1+\aspec^2))
\]
\end{lemma}

\begin{proof}
When $S$ and $T$ are disjoint sets of size 2 each, we claim that $\calE^{(4)}[{S,T}]$ is equal to
\[
    \lim_{\HadDim\to\infty}\E_{\bz\sim\{\pm 1\}^{\HadDim}}[(q_{S,\HadDim}(\bz)-q_{S,\HadDim}(\bz)^{\le 2})(q_{T,\HadDim}(\bz)-q_{T,\HadDim}(\bz)^{\le 2})]
\]
By definition,
\begin{align*}
    \calM^{(1)}_{S,T} - \calM^{(2)}_{S,T} &= \lim_{\HadDim\to\infty}\E_{\bz\sim\{\pm 1\}^\HadDim}[q_{S,\HadDim}(\bz)q_{T,\HadDim}(\bz) - p_{S,\HadDim}(\bz)p_{T,\HadDim}(\bz)]
\end{align*}
and note that 
\begin{align*}
    &\lim_{\HadDim\to\infty}\E_{\bz\sim\{\pm 1\}^\HadDim}[(q_{S,\HadDim}(\bz)-q_{S,\HadDim}(\bz)^{\le 2})(q_{T,\HadDim}(\bz)-q_{T,\HadDim}(\bz)^{\le 2})] \\
    =~&\lim_{\HadDim\to\infty}\E_{\bz\sim\{\pm 1\}^\HadDim}[q_{S,\HadDim}(\bz)q_{T,\HadDim}(\bz)] -\E_{\bz\sim\{\pm 1\}^\HadDim}[q_{S,\HadDim}(\bz)q_{T,\HadDim}(\bz)^{\le 2}]\\&- \E_{\bz\sim\{\pm 1\}^\HadDim}[q_{S,\HadDim}(\bz)^{\le 2}q_{T,\HadDim}(\bz)] + \E_{\bz\sim\{\pm 1\}^\HadDim}[q_{S,\HadDim}(\bz)^{\le 2}q_{T,\HadDim}(\bz)^{\le 2}]\\
    =~&\lim_{\HadDim\to\infty}\E_{\bz\sim\{\pm 1\}^\HadDim}[q_{S,\HadDim}(\bz)q_{T,\HadDim}(\bz)] - \E_{\bz\sim\{\pm 1\}^\HadDim}[q_{S,\HadDim}(\bz)^{\le 2}q_{T,\HadDim}(\bz)^{\le 2}]\\
    =~&\lim_{\HadDim\to\infty}\E_{\bz\sim\{\pm 1\}^\HadDim}\E_{\bz\sim\{\pm 1\}^\HadDim}[q_{S,\HadDim}(\bz)q_{T,\HadDim}(\bz)] - \E_{\bz\sim\{\pm 1\}^\HadDim}[p_{S,\HadDim}(\bz)p_{T,\HadDim}(\bz)]\\
    =~&\calM^{(1)}[{S,T}] - \calM^{(2)}[{S,T}]
\end{align*}
Thus, if we set $\Delta_{S,\HadDim} \coloneqq q_{S,\HadDim}(z)-q_{S,\HadDim}(z)^{\leq 2}$, then $\calE^{(4)}_{S,T} = \lim_{\HadDim\to\infty}\E_{\bz\sim\{\pm 1\}^\HadDim}[\Delta_S(\bz)\Delta_T(\bz)]$.
We can write $\Delta_S$ as a linear combination of graphical polynomials of the following half-glyphs from \pref{fig:half-rib-E4};  in particular
\[
    \Delta_{\{i,j\},\HadDim} = -2\GraphPoly{T_{(i,j)},\HadDim}+ -2\GraphPoly{T_{(j,i)},\HadDim}-2\GraphPoly{W_{\{i,j\}},\HadDim}+4\GraphPoly{D_{\{i,j\}},\HadDim}.
\]
Thus, $\calE^{(4)}[S,T]$ can be written as a linear combination of terms of the form
\[
    \lim_{\HadDim\to\infty}\E_{\bz\sim\{\pm 1\}^{\HadDim}}[\GraphPoly{\calS_1,\HadDim}(\bz)\GraphPoly{\calS_2,\HadDim}(\bz)],
\]
which (just like in \pref{sec:E-3}) we denote as $\calS_1\cdot\calS_2$.  When $\calS_1$ and $\calS_2$ do not have the same number of odd-degree right vertices, $\calS_1\cdot\calS_2$ is $0$, so after discarding away such pairs:
\begin{align*}
    \calE^{(4)}[\{i,j\},\{k,\ell\}] =~&4\left(T_{(i,j)}\cdot T_{(k,\ell)} + T_{(i,j)}\cdot T_{(\ell,k)} + T_{(j,i)}\cdot T_{(k,\ell)} + T_{(j,i)}\cdot T_{(\ell,k)}\right)+
    \\& 4\left(T_{(i,j)}\cdot W_{\{k,\ell\}} + T_{(j,i)}\cdot W_{\{k,\ell\}} + W_{\{i,j\}}\cdot T_{(k,\ell)} + W_{\{i,j\}}\cdot T_{(\ell,k)}\right) +
    \\& 4W_{\{i,j\}}\cdot W_{\{k,\ell\}} + 16D_{\{i,j\}}\cdot D_{\{k,\ell\}}\numberthis \label{eq:E4-expansion}.
\end{align*}

\begin{figure}[h]
\centering
      \begin{subfigure}[h]{0.3\textwidth}\centering
\begin{tikzpicture}[
every edge/.style = {draw=black,very thick},
 vrtx/.style args = {#1/#2}{%
      draw, thick, fill=white,
      minimum size=5mm, label=#1:#2}]
       
\node[circle](B) [vrtx=left/$i$] at (0, 0) {};
\node[rectangle] (C) [vrtx=right/] at (1,1) {};
\node[circle](E) [vrtx=left/$j$] at (0, -2) {};
\node(D) [vrtx=right/] at (1,0) {};
\node(F) [vrtx=right/] at (1,-1) {};
\node(H) [vrtx=right/] at (1,-2) {};

\path  
        (B) edge (C)
        (B) edge (D)
        (B) edge (F)
        (E) edge (H)
        ;
\end{tikzpicture}
\caption{$T_{(i,j)}$}
    \end{subfigure}
\hfil    
 \begin{subfigure}[h]{0.3\textwidth}\centering
\begin{tikzpicture}[
every edge/.style = {draw=black,very thick},
 vrtx/.style args = {#1/#2}{%
     draw, thick, fill=white,
      minimum size=5mm, label=#1:#2}]
       
\node[circle](B) [vrtx=left/$i$] at (0, 0) {};
\node[rectangle](C) [vrtx=right/] at (1,1) {};
\node[circle](E) [vrtx=left/$j$] at (0, -2) {};
\node(D) [vrtx=right/] at (1,0) {};
\node(F) [vrtx=right/] at (1,-1) {};
\node(H) [vrtx=right/] at (1,-2) {};
\node(G) [vrtx=right/] at (1,-3) {};

\path  
        (B) edge (C)
        (B) edge (D)
        (B) edge (F)
        (E) edge (F)
        (E) edge (H)
        (E) edge (G)
        ;
\end{tikzpicture}
\caption{$W_{\{i,j\}}$}
    \end{subfigure}
\hfil    

 \begin{subfigure}[h]{0.3\textwidth}\centering
\begin{tikzpicture}[
every edge/.style = {draw=black,very thick},
 vrtx/.style args = {#1/#2}{%
      , draw, thick, fill=white,
      minimum size=5mm, label=#1:#2}]
       
\node[circle](B) [vrtx=left/$i$] at (0, 0) {};
\node[rectangle](C) [vrtx=right/] at (1,1.2) {};
\node[circle](E) [vrtx=left/$j$] at (0, -1) {};
\node(D) [vrtx=right/] at (1.6,0) {};
\node(F) [vrtx=right/] at (1.3,0.6) {};
\node(G) [vrtx=right/] at (1,-1) {};
\node(H) [vrtx=right/] at (1.3,-1.6) {};

\node(I) [vrtx=right/] at (1.6,-2.3) {};

\path  
        (B) edge (C)
        (B) edge (D)
        (B) edge (F)
        (E) edge (G)
        (E) edge (H)
        (E) edge (I)
        ;
\end{tikzpicture}
\caption{$D_{\{i,j\}}$}
    \end{subfigure}
\caption{Half-Glyphs for $\Delta_{\{i,j\},\HadDim}$}
\label{fig:half-rib-E4}
\end{figure}


\noindent $\calE^{(4)}$ can then be written as
\[
    4\calH_1 + 4\calH_2 + 4\calH_3 + 16\calH_4
\]
where
\begin{align}
    \calH_1[\{i,j\},\{k,\ell\}] &\coloneqq T_{(i,j)}\cdot T_{(k,\ell)} + T_{(i,j)}\cdot T_{(\ell,k)} + T_{(j,i)}\cdot T_{(k,\ell)} + T_{(j,i)}\cdot T_{(\ell,k)}  \label{eq:H1} \\
    \calH_2[\{i,j\},\{k,\ell\}] &\coloneqq T_{(i,j)}\cdot W_{\{k,\ell\}} + T_{(j,i)}\cdot W_{\{k,\ell\}} + W_{\{i,j\}}\cdot T_{(k,\ell)} + W_{\{i,j\}}\cdot T_{(\ell,k)}   \label{eq:H2}\\
    \calH_3[\{i,j\},\{k,\ell\}] &\coloneqq W_{\{i,j\}}\cdot W_{\{k,\ell\}}   \label{eq:H3}\\
    \calH_4[\{i,j\},\{k,\ell\}] &\coloneqq D_{\{i,j\}}\cdot D_{\{k,\ell\}}.\label{eq:H4}
\end{align}


To attain an upper bound on $\|\calE^{(4)}\|$, we will upper bound $\|\calH_1\|,\|\calH_2\|,\|\calH_3\|,\|\calH_4\|$ and appeal to a triangle inequality.  Henceforth, we index the rows and columns of $\calE^{(4)}$ and its components by \emph{ordered pairs} $(i,j)$ where $i < j$ instead of a size-$2$ set $\{i,j\}$.

\paragraph{Spectral norm bound for $\calH_1$.}   
$\calH_1$ can be further broken into a sum of four matrices $\calH_{1,1}+\calH_{1,2}+\calH_{1,3}+\calH_{1,4}$ where
\begin{align*}
    \calH_{1,1}[(i,j),(k,\ell)] &= T_{(i,j)}\cdot T_{(k,\ell)}\\
    \calH_{1,2}[(i,j),(k,\ell)] &= T_{(i,j)}\cdot T_{(\ell,k)}\\
    \calH_{1,3}[(i,j),(k,\ell)] &= T_{(j,i)}\cdot T_{(k,\ell)}\\
    \calH_{1,4}[(i,j),(k,\ell)] &= T_{(j,i)}\cdot T_{(\ell,k)}.
\end{align*}
We illustrate how to bound the spectral norm of $\calH_{1,1}$;  the spectral norm bounds for $\calH_{1,2},\calH_{1,3}$ and $\calH_{1,4}$, and their proofs, are exactly identical.  An application of triangle inequality lets us conclude a final bound on $\|\calH_1\|$.

It can be verified that
\[
    \calH_{1,1}[(i,j),(k,\ell)] = \frac{1}{3!}\langle M[i], M[k]\rangle^3\langle M[j], M[\ell]\rangle + \frac{1}{2!}\langle M[i], M[k]\rangle^2\langle M[i], M[\ell]\rangle\langle M[j], M[k]\rangle
\]
which lets us write $\calH_{1,1}$ as a sum of two graphical matrices.  In particular,
\[
    \calH_{1,1} = \GraphMat_{\calT\calT_1}+\GraphMat_{\calT\calT_2}
\]
where
\begin{align*}
    \GraphMat_{\calT\calT_1}[(i,j),(k,\ell)] &= \frac{1}{3!}\langle M[i], M[k]\rangle^3\langle M[j], M[\ell]\rangle\\
    \GraphMat_{\calT\calT_2}[(i,j),(k,\ell)] &= \frac{1}{2!}\langle M[i], M[k]\rangle^2\langle M[i], M[\ell]\rangle\langle M[j], M[k]\rangle.
\end{align*}

\begin{figure}[h]
\centering
      \begin{subfigure}[h]{0.3\textwidth}\centering
\begin{tikzpicture}[
every edge/.style = {draw=black,very thick},
 vrtx/.style args = {#1/#2}{%
     draw, thick, fill=white,
      minimum size=5mm, label=#1:#2}]
       
\node[circle](I) [vrtx=left/$i$] at (0, 0) {};
\node[circle](J) [vrtx=left/$j$] at (0,-1.5) {};
\node[circle] (K) [vrtx=right/$k$] at (2,0) {};
\node[circle](L) [vrtx=right/$\ell$] at (2,-1.5) {};
\node[rectangle] (a) [vrtx=above/] at (1,1){};
\node (b) [vrtx=above/] at (1, 0) {};
\node (c) [vrtx=above/] at (1,-1) {};
\node (d) [vrtx=below/] at (1,-2) {};
\path  
     (I) edge (a)
     (a) edge (K)
     (I) edge (b)
     (b) edge (K)
     (I) edge (c)
     (c) edge (K)
     (J) edge (d)
     (d) edge (L)
     ;
\end{tikzpicture}
\caption{$\calT\calT_1$}
    \end{subfigure}
\hfil    
 \begin{subfigure}[h]{0.3\textwidth}\centering
 
\begin{tikzpicture}[
every edge/.style = {draw=black,very thick},
 vrtx/.style args = {#1/#2}{%
      draw, thick, fill=white,
      minimum size=5mm, label=#1:#2}]

\node[circle] (I) [vrtx=left/$i$] at (0, 0) {};
\node[circle](J) [vrtx=left/$j$] at (0,-1.5) {};
\node[circle] (K) [vrtx=right/$k$] at (2,0) {};
\node[circle](L) [vrtx=right/$\ell$] at (2,-1.5) {};
\node[rectangle] (a) [vrtx=above/] at (1,1){};
\node (b) [vrtx=above/] at (1, 0) {};
\node (c) [vrtx=above/] at (1,-1) {};
\node (d) [vrtx=below/] at (1,-2.5) {};
\path  
     (I) edge (a)
     (a) edge (K)
     (I) edge (b)
     (b) edge (K)
     (I) edge (c)
     (c) edge (L)
     (J) edge (d)
     (d) edge (K)
     ;
\end{tikzpicture}
\caption{$\calT\calT_2$}
    \end{subfigure}
\caption{Graphical matrices arising out of $\calH_{1,1}$.}
\label{fig:H11}
\end{figure}

We use factorizations of the graphical matrices combined with \pref{lem:bound-from-factorization} to bound their spectral norms.  Concretely,
\begin{align*}
    \GraphMat_{\calT\calT_1} = \frac{1}{3!}\calL_1^{(1)}\cdot\calL_2^{(1)}\cdot \calL_3^{(1)}\cdot\calL_4^{(1)} 
\end{align*}
where
\begin{align*}
    \calL_1^{(1)}[(i,j),(i,j,k)] &\coloneqq \langle M[i], M[k]\rangle &\text{(a growth matrix) }\\
    \calL_2^{(1)}[(i,j,k),(i,j,k)] &\coloneqq \langle M[i], M[k]\rangle &\text{(a residue matrix)} \\
    \calL_3^{(1)}[(i,j,k),(j,k)] &\coloneqq \langle M[i], M[k]\rangle&\text{(a shrinkage matrix)} \\
    \calL_4^{(1)}[(j,k),(k,\ell)] &\coloneqq \langle M[j], M[\ell]\rangle  &\text{(a swap matrix)}
\end{align*}
which implies (via \pref{lem:bound-from-factorization}) that
\[
    \|\GraphMat_{\calT\calT_1}\| \leq \frac{1}{3!}\altwo^2\cdot\aspec\cdot\amag
\]
Similarly, 
\begin{align*}
     \GraphMat_{\calT\calT_2}=\frac{1}{2!} \calL_1^{(2)}\cdot\calL_2^{(2)}\cdot \calL_3^{(2)}
\end{align*}
where
\begin{align*}
    \calL_1^{(2)}[(i,j),(i,k)] &\coloneqq \langle M[j], M[k]\rangle &\text{(a swap matrix) }\\
    \calL_2^{(2)}[(i,k),(i,k)] &\coloneqq \langle M[i], M[k]\rangle^2 &\text{(two residue matrices)} \\
    \calL_3^{(2)}[(i,k),(k,\ell)] &\coloneqq \langle M[i], M[\ell]\rangle &\text{(a swap matrix)}
\end{align*}
which implies (via \pref{lem:bound-from-factorization}) that
\[
    \|\GraphMat_{\calT\calT_2}\| \leq \frac{1}{2!}\amag^2\cdot\aspec^2
\]
Therefore, by triangle inequality and the fact that $\amag\le 1$,
\[
    \|\calH_{1,1}\| \leq \|\GraphMat_{\calT\calT_1}\| + \|\GraphMat_{\calT\calT_2}\| \le O(\amag(1+\arow^2)(1+\aspec^2))
\]
and by applying triangle inequality on $\|\calH_{1,1}+\dots+\calH_{1,4}\|$, it follows that
\[
    \|\calH_1\| \le O(\amag(1+\arow^2)(1+\aspec^2)).
\]


\paragraph{Spectral norm bound for $\calH_2$.}
Following \pref{eq:H2} we can write $\calH_2$ as $\calH_{2,1}+\calH_{2,2}+\calH_{2,3}+\calH_{2,4}$ where
\begin{align*}
    \calH_{2,1}[(i,j),(k,\ell)] &= T_{(i,j)}\cdot W_{\{k,\ell\}}\\
    \calH_{2,2}[(i,j),(k,\ell)] &= T_{(j,i)}\cdot W_{\{k,\ell\}}\\
    \calH_{2,3}[(i,j),(k,\ell)] &= W_{\{i,j\}}\cdot T_{(k,\ell)}\\
    \calH_{2,4}[(i,j),(k,\ell)] &= W_{\{i,j\}}\cdot T_{(\ell,k)}.
\end{align*}
It can be verified that each $\calH_{2,t}$ can be written as a sum of two graphical matrices of glyphs isomorphic to $\calT\calW$, where
\[
    \GraphMat_{\calT\calW}[(i,j),(k,\ell)] = \frac{1}{2!} \langle M[i], M[k]\rangle^2 \langle M[i], M[\ell]\rangle \langle M[j], M[\ell]\rangle\langle M[k], M[\ell]\rangle.
\]
which means $\calH_2$ is the sum of $8$ graphical matrices of glyphs isomorphic to $\calT\calW$.  For each such glyph $\Glyph$, $\|\GraphMat_{\Glyph}\|$ can be obtain an identical bound to that we obtain on $\|\GraphMat_{\calT\calW}\|$ using an identical proof.  Thus, from a triangle inequality, we can bound $\|\calH_2\|$ by $8C$ where obtain a bound of $C$ on $\|\GraphMat_{\calT\calW}\|$.


\begin{figure}[h] \centering
\begin{tikzpicture}[
every edge/.style = {draw=black,very thick},
 vrtx/.style args = {#1/#2}{%
      draw, thick, fill=white,
      minimum size=5mm, label=#1:#2}]
       
\node[circle](B) [vrtx=left/$i$] at (0, 0) {};
\node[rectangle](C) [vrtx=right/] at (1,1) {};
\node[circle](E) [vrtx=left/$j$] at (0, -3) {};
\node(D) [vrtx=right/] at (1,0) {};
\node(F) [vrtx=right/] at (1.4,-1) {};
\node(H) [vrtx=right/] at (1,-2) {};
\node(G) [vrtx=right/] at (1,-3) {};
\node[circle](K) [vrtx=right/$k$] at (2, 0) {};
\node[circle](L) [vrtx=right/$\ell$] at (2, -3) {};

\path  
        (B) edge (C)
        (C) edge (K)
        (B) edge (D)
        (D) edge (K)
        (B) edge (H)
        (F) edge (L)
        (K) edge (F)
        (H) edge (L)
        (E) edge (G)
        (G) edge (L)
        ;
\end{tikzpicture}
\caption{Glyph $\calT\calW$}
\label{fig:TW}
\end{figure}

Towards obtaining the bound, we factorize
\[ 
    \GraphMat_{\calT\calW} = \calL_1 \cdot \calL_2 \cdot \calL_3 \cdot  \calL_4 \cdot \calL_5
\]
where
\begin{align*}
    \calL_1[(i,j),(i,j,k)]&\coloneqq \langle M[i],M[k]\rangle &\text{(a growth matrix)} \\
    \calL_2[(i,j,k),(i,j,k)]&\coloneqq \langle M[i],M[k]\rangle &\text{(a residue matrix)}\\
    \calL_3[(i,j,k),(j,k,\ell)]&\coloneqq \langle M[i],M[\ell]\rangle &\text{(a swap matrix)} \\
    \calL_4[(j,k,\ell),(k,\ell)]&\coloneqq \langle M[j],M[\ell]\rangle &\text{(a shrinkage matrix)} \\
     \calL_5[(k,\ell),(k,\ell)]&\coloneqq \langle M[k],M[\ell]\rangle &\text{(a residue matrix)} \\
\end{align*}
which gives (via \pref{lem:bound-from-factorization})
\[
    \| \GraphMat_{\calT\calW}\| = O(\altwo^2\cdot \aspec \cdot \amag^2).
\]

\paragraph{Spectral norm bound for $\calH_3$.}
It can be verified that $\calH_3$ is the sum of three graphical matrices $\GraphMat_{\calW\calW_1}+\GraphMat_{\calW\calW_2}+\GraphMat_{\calW\calW_3}$ where
\begin{align*}
    \GraphMat_{\calW\calW_1}[(i,j),(k,\ell)] &= \frac{1}{2!\cdot 2!}\langle M[i], M[k]\rangle^2\langle M[j], M[\ell]\rangle^2\langle M[i], M[j]\rangle \langle M[k], M[\ell]\rangle \\
    \GraphMat_{\calW\calW_2}[(i,j),(k,\ell)] &= \frac{1}{2!\cdot 2!}\langle M[i], M[\ell]\rangle^2\langle M[j], M[k]\rangle^2\langle M[i], M[j]\rangle \langle M[k], M[\ell]\rangle \\
    \GraphMat_{\calW\calW_3}[(i,j),(k,\ell)] &= \langle M[i],M[j]\rangle\langle M[i], M[k]\rangle\langle M[i], M[\ell]\rangle\langle M[j], M[k]\rangle\langle M[j], M[\ell]\rangle\langle M[k], M[\ell]\rangle.
\end{align*}
This implies that $\|\calH_3\|\le\|\GraphMat_{\calW\calW_1}\|+\|\GraphMat_{\calW\calW_2}\|+\|\GraphMat_{\calW\calW_3}\|$.
$\calW\calW_1$ and $\mathcal{WW}_2$ are isomorphic, and an identical proof yields an identical bound on $\|\GraphMat_{\mathcal{WW}_2}\|$ as $\|\GraphMat_{\calW\calW_1}\|$, and hence we only show how to bound $\|\GraphMat_{\calW\calW_1}\|$ and $\|\GraphMat_{\calW\calW_3}\|$.

\begin{figure}[h] \centering
  \begin{subfigure}[h]{0.3\textwidth}\centering
\begin{tikzpicture}[
every edge/.style = {draw=black,very thick},
 vrtx/.style args = {#1/#2}{%
       draw, thick, fill=white,
      minimum size=5mm, label=#1:#2}]
       
\node[circle](B) [vrtx=left/$i$] at (0, 0) {};
\node[rectangle](C) [vrtx=right/] at (1,1) {};
\node[circle](E) [vrtx=left/$j$] at (0, -2) {};
\node(D) [vrtx=right/] at (1,0) {};
\node(F) [vrtx=right/] at (.4,-1) {};
\node(H) [vrtx=right/] at (1,-2) {};
\node(G) [vrtx=right/] at (1,-3) {};
\node[circle](K) [vrtx=right/$k$] at (2, 0) {};
\node[circle](L) [vrtx=right/$\ell$] at (2, -2) {};
\node(P) [vrtx=right/] at (1.6,-1) {};

\path  
        (B) edge (C)
        (B) edge (D)
        (B) edge (F)
        (E) edge (F)
        (E) edge (H)
        (E) edge (G)
        (K) edge (C)
        (K) edge (D)
        (K) edge (P)
        (P) edge (L)
        (L) edge (H)
        (L) edge (G)
        ;
\end{tikzpicture}
\caption{$\calW\calW_1$}
    \end{subfigure}
    \hfil
 \begin{subfigure}[h]{0.3\textwidth}\centering
\begin{tikzpicture}[
every edge/.style = {draw=black,very thick},
 vrtx/.style args = {#1/#2}{%
       draw, thick, fill=white,
      minimum size=5mm, label=#1:#2}]
       
\node[circle](B) [vrtx=left/$i$] at (0, 0) {};
\node[rectangle] (C) [vrtx=right/] at (1,1) {};
\node[circle](E) [vrtx=left/$j$] at (0, -2) {};
\node(D) [vrtx=right/] at (1,-.4) {};
\node(F) [vrtx=right/] at (.3,-1) {};
\node(H) [vrtx=right/] at (1,-1.6) {};
\node(G) [vrtx=right/] at (1,-3) {};
\node[circle](K) [vrtx=right/$k$] at (2, 0) {};
\node[circle](L) [vrtx=right/$\ell$] at (2, -2) {};
\node(P) [vrtx=right/] at (1.9,-1) {};

\path  
        (B) edge (C)
        (B) edge (D)
        (B) edge (F)
        (E) edge (F)
        (E) edge (H)
        (E) edge (G)
        (K) edge (C)
        (L) edge (D)
        (K) edge (P)
        (P) edge (L)
        (K) edge (H)
        (L) edge (G)
        ;
\end{tikzpicture}
\caption{$\calW\calW_3$}
    \end{subfigure}    
    
\caption{Graphical matrices arising out of $\calH_3$. $\calW\calW_1$ and $\calW\calW_2$ are isomorphic.}
\end{figure}



\noindent We can write
\begin{align*}
    \GraphMat_{\calW\calW_1} &= \frac{1}{2!\cdot 2!}\calL_{1}^{(1)}\cdot \calL_{2}^{(1)} \cdot \calL_{3}^{(1)}\cdot \calL_{4}^{(1)} \cdot \calL_{5}^{(1)} \cdot \calL_{6}^{(1)}
\end{align*} 
where
\begin{align*}
       \calL_1^{(1)}[(i,j),(i,j)] &\coloneqq \langle M[i], M[j]\rangle &\text{(a residue matrix)} \\
      \calL_2^{(1)}[(i,j),(i,j,k)] &\coloneqq \langle M[i], M[k]\rangle &\text{(a growth matrix) } \\
    \calL_3^{(1)}[(i,j,k),(j,k)] &\coloneqq  \langle M[i], M[k]\rangle &\text{(a shrinkage matrix)} \\
    \calL_4^{(1)}[(j,k),(j,k,\ell)] &\coloneqq \langle M[j], M[\ell]\rangle &\text{(a growth matrix)} \\
    \calL_5^{(1)}[(j,k,\ell),(k,\ell)] &\coloneqq \langle M[j], M[\ell]\rangle &\text{(a shrinkage matrix)} \\
    \calL_6^{(1)}[(k,\ell),(k,\ell)] &\coloneqq \langle M[k], M[\ell]\rangle &\text{(a residue matrix)}
\end{align*}
which via \pref{lem:bound-from-factorization} implies
\[ 
    \|\GraphMat_{\calW\calW_1}\| \leq \frac{1}{2!\cdot 2!}\altwo^4\cdot\amag^2.
\]
We also have:
\begin{align*}
    \GraphMat_{\calW\calW_3} &= \calL_{1}^{(2)}\cdot \calL_{2}^{(2)} \cdot \calL_{3}^{(2)} \cdot \calL_{4}^{(2)} \cdot \calL_{5}^{(2)} \cdot \calL_{6}^{(2)}
\end{align*}
where
\begin{align*}
    \calL_1^{(2)}[(i,j),(i,j)] &\coloneqq \langle M[i], M[j]\rangle &\text{(a residue matrix) } \\
    \calL_2^{(2)}[(i,j),(i,j,k)] &\coloneqq \langle M[i], M[k]\rangle &\text{(a growth matrix)} \\
    \calL_3^{(2)}[(i,j,k),(i,j,k)] &\coloneqq  \langle M[j], M[k]\rangle  &\text{(a residue matrix) } \\
    \calL_4^{(2)}[(i,j,k),(j,k,\ell)] &\coloneqq \langle M[i], M[\ell]\rangle &\text{(a swap matrix)} \\
    \calL_5^{(2)}[(j,k,\ell),(k,\ell)] &\coloneqq \langle M[j], M[\ell]\rangle &\text{(a shrinkage matrix)} \\
    \calL_6^{(2)}[(k,\ell),(k,\ell)] &\coloneqq \langle M[k], M[\ell]\rangle &\text{(a residue matrix)} 
\end{align*}
which via \pref{lem:bound-from-factorization} implies
\[
    \|\GraphMat_{\calW\calW_3}\| \leq    \altwo^2\cdot\amag^3\cdot\aspec
\]
Putting the above bounds together with $\amag\le 1$,
\[
    \|\calH_3\| \leq O(\amag^2\cdot(1+\arow^2)\cdot(1+\aspec)).
\]

\paragraph{Spectral norm bound for $\calH_4$.}
It can be verified that $\calH_4$ is the sum of four graphical matrices $\GraphMat_{\calD\calD_1}+\GraphMat_{\calD\calD_2}+\GraphMat_{\calD\calD_3}+\GraphMat_{\calD\calD_4}$ where
\begin{align*}
    \GraphMat_{\calD\calD_1}[(i,j),(k,\ell)] &= \frac{1}{3!\cdot 3!}\CholProd{i}{k}^3\CholProd{j}{\ell}^3\\
    \GraphMat_{\calD\calD_2}[(i,j),(k,\ell)] &= \frac{1}{3!\cdot 3!}\CholProd{i}{\ell}^3\CholProd{j}{k}^3\\
    \GraphMat_{\calD\calD_3}[(i,j),(k,\ell)] &= \frac{1}{2!\cdot 2!}\CholProd{i}{k}^2\CholProd{j}{\ell}^2\CholProd{i}{\ell}\CholProd{j}{k}\\
    \GraphMat_{\calD\calD_4}[(i,j),(k,\ell)] &= \frac{1}{2!\cdot 2!}\CholProd{i}{\ell}^2\CholProd{j}{k}^2\CholProd{i}{k}\CholProd{j}{\ell}.
\end{align*}
The glyphs $\calD\calD_1$ and $\calD\calD_2$ are isomorphic and the glyphs $\calD\calD_3$ and $\calD\calD_4$ are isomorphic.  We bound $\|\GraphMat_{\calD\calD_1}\|$ and $\|\GraphMat_{\calD\calD_3}\|$;  we can achieve the same bounds on $\|\GraphMat_{\calD\calD_2}\|$ (and $\GraphMat_{\calD\calD_4}$ resp.) as we do on $\GraphMat_{\calD\calD_1}$ (and $\GraphMat_{\calD\calD_3}$ resp.) via identical proofs.

\begin{figure}[h] \centering
  \begin{subfigure}[h]{0.3\textwidth}\centering
\begin{tikzpicture}[
every edge/.style = {draw=black,very thick},
 vrtx/.style args = {#1/#2}{%
       draw, thick, fill=white,
      minimum size=5mm, label=#1:#2}]
       
\node[circle](B)  [vrtx=left/$i$] at (0, 0) {};
\node[rectangle](C) [vrtx=right/] at (1,1) {};
\node[circle](E) [vrtx=left/$j$] at (0, -3) {};
\node(D) [vrtx=right/] at (1,0) {};
\node(F) [vrtx=right/] at (1,-1) {};
\node(H) [vrtx=right/] at (1,-2) {};
\node(G) [vrtx=right/] at (1,-3) {};
\node[circle](K) [vrtx=right/$k$] at (2, 0) {};
\node[circle](L) [vrtx=right/$\ell$] at (2, -3) {};
\node(P) [vrtx=right/] at (1,-4) {};

\path  
        (B) edge (C)
        (B) edge (D)
        (B) edge (F)
        (K) edge (F)
        (E) edge (H)
        (E) edge (G)
        (K) edge (C)
        (K) edge (D)
        (E) edge (P)
        (P) edge (L)
        (L) edge (H)
        (L) edge (G)
        ;
\end{tikzpicture}
\caption{$\calD\calD_1$}
    \end{subfigure}
    \hfil
 \begin{subfigure}[h]{0.3\textwidth}\centering
\begin{tikzpicture}[
every edge/.style = {draw=black,very thick},
 vrtx/.style args = {#1/#2}{%
      draw, thick, fill=white,
      minimum size=5mm, label=#1:#2}]
              
\node[circle](B) [vrtx=left/$i$] at (0, 0) {};
\node[rectangle](C) [vrtx=right/] at (1,1) {};
\node[circle](E) [vrtx=left/$j$] at (0, -3) {};
\node(D) [vrtx=right/] at (1,0) {};
\node(F) [vrtx=right/] at (1,-1) {};
\node(H) [vrtx=right/] at (1,-2) {};
\node(G) [vrtx=right/] at (1,-3) {};
\node[circle](K) [vrtx=right/$k$] at (2, 0) {};
\node[circle](L) [vrtx=right/$\ell$] at (2, -3) {};
\node(P) [vrtx=right/] at (1,-4) {};

\path  
        (B) edge (C)
        (B) edge (D)
        (B) edge (F)
        (L) edge [bend right] (F)
        (E) edge (H)
        (E) edge (G)
        (K) edge (C)
        (K) edge (D)
        (E) edge (P)
        (P) edge (L)
        (K) edge [bend left] (H)
        (L) edge (G)
        ;
\end{tikzpicture}
\caption{$\calD\calD_3$}
    \end{subfigure}    
    
\caption{Graphical matrices arising out of $\calH_4$.}
\end{figure}

We can factorize
\begin{align*}
    \GraphMat_{\calD\calD_1} &= \frac{1}{3!\cdot 3!}\calL_{1}^{(1)}\cdot \calL_{2}^{(1)}\cdot \calL_{3}^{(1)}\cdot  \calL_{4}^{(1)}\cdot \calL_{5}^{(1)}\cdot  \calL_{6}^{(1)}
\end{align*}
where
\begin{align*}
    \calL_1^{(1)}[(i,j),(i,j,k)] &\coloneqq \langle M[i], M[k]\rangle &\text{(a growth matrix)} \\
    \calL_2^{(1)}[(i,j,k),(i,j,k)] &\coloneqq \langle M[i], M[k]\rangle &\text{(a residue matrix)} \\
    \calL_3^{(1)}[(i,j,k),(j,k)] &\coloneqq \langle M[i], M[k]\rangle &\text{(a shrinkage matrix)} \\ 
    \calL_4^{(1)}[(j,k),(j,k,\ell)] &\coloneqq \langle M[j], M[\ell]\rangle &\text{(a growth matrix)} \\
    \calL_5^{(1)}[(j,k,\ell),(j,k,\ell)] &\coloneqq \langle M[j], M[\ell]\rangle &\text{(a residue matrix)} \\
    \calL_6^{(1)}[(j,k,\ell),(k,\ell)] &\coloneqq \langle M[j], M[\ell]\rangle &\text{(a shrinkage matrix)} \\  
\end{align*}
Hence,
\[
    \|\GraphMat_{\calD\calD_1} \|\leq O(\altwo^4\cdot \amag^2).
\]
And
\begin{align*}
    \GraphMat_{\calD\calD_3} &= \frac{1}{2!\cdot 2!} \calL_{1}^{(2)}\cdot \calL_{2}^{(2)} \cdot \calL_{3}^{(2)}\cdot \calL_{4}^{(2)} \cdot \calL_{5}^{(2)}\cdot \calL_{6}^{(2)}   
\end{align*}
where
\begin{align*}
    \calL_1^{(2)}[(i,j),(i,j,k)] &\coloneqq \langle M[i], M[k]\rangle  &\text{(a growth matrix)} \\
    \calL_2^{(2)}[(i,j,k),(i,j,k)] &\coloneqq \langle M[i], M[k]\rangle  &\text{(a residue matrix)} \\
    \calL_3^{(2)}[(i,j,k),(i,j,k)] &\coloneqq \langle M[j], M[k]\rangle &\text{(a residue matrix)} \\
    \calL_4^{(2)}[(i,j,k),(j,k,\ell)] &\coloneqq \langle M[i], M[\ell]\rangle    &\text{(a swap matrix) }\\
    \calL_5^{(2)}[(j,k,\ell),(j,k,\ell)] &\coloneqq \langle M[j], M[\ell]\rangle    &\text{(a residue matrix) }\\
    \calL_6^{(2)}[(j,k,\ell),(k,\ell)] &\coloneqq \langle M[j], M[\ell] \rangle &\text{(a shrinkage matrix)} \\
\end{align*}
which gives
\[
    \|\GraphMat_{\calD\calD_3}\| \leq O(\altwo^2\cdot\aspec\cdot\amag^3).
\]
Putting the above bounds together with $\amag\le 1$:
\[
    \GraphMat_{\calH_4} \leq O(\amag^2\cdot(1+\arow^4)\cdot(1+\aspec)).
\]


The lemma statement follows immediately from the spectral norm bounds on $\calH_1$, $\calH_2$, $\calH_3$ and $\calH_4$, and a triangle inequality.
\end{proof}

\section{Degree-$4$ SoS Lower Bound  for the Sherrington--Kirkpatrick Hamiltonian}
\subsection{Gaussian concentration}
In this section, we give a brief review of standard concentration results related to Gaussian random variables, vectors, and matrices.

As in previous sections, let $\bM$ be a $n\times d$ matrix where each entry is independently sampled from $\calN\left(0,\frac{1}{d}\right)$ and assume $d < n$.

\begin{lemma}[Concentration of singular values of Gaussian matrices, {\cite[Corollary 5.35]{Ver10}}] \label{lem: singular-val-Gauss}
    Except with probability $2\exp\left(-\frac{t^2}{2}\right)$,
    \[
        \frac{\sqrt{n}-\sqrt{d}-t}{\sqrt{d}}\leq s_{\min}(\bM)\leq s_{\max} (\bM)\leq \frac{\sqrt{n}+\sqrt{d}+t}{\sqrt{d}} \ .
    \]
\end{lemma}

\begin{corollary} \label{cor: gaussian-covariance-matrix}
Except with probability $2\exp\left(\frac{-t^2}{2}\right)$,
\[
    \|\bM\bM^{\dagger}\| \leq \|\bM\|^2 \leq \frac{n + d + 2\sqrt{dn} + t^2 + 2(\sqrt{d}+\sqrt{n})t}{d}
\]
\end{corollary}
\begin{fact}[Concentration of norm of Gaussian vector] \label{fact: norm-of-gauss}
Let $\bx$ be a vector of i.i.d. Gaussian entries.  There exist absolute constants $\alpha,\beta>0$ such that,
\[
    \Pr\left[\|\bx\|\notin \left[\sqrt{d}- t, \sqrt{d}+t\right] \right] \le \alpha\exp(-\beta t^2).
\]
\end{fact}

An implication of the above fact is the following:
\begin{corollary} \label{cor: norm-concentration}
Except with probability $n^{-100}$, for all $i$,
\[
    \langle \bM_i,\bM_i\rangle \in \left[1-100\sqrt{\frac{\log n}{d}}, 1+100\sqrt{\frac{\log n}{d}}\right]
\]
 \end{corollary}

\begin{lemma} \label{cor: inner-product-ortho}
Except with probability at least $n^{-100}$, for all pairs of distinct $i,j$,
\[
    \langle \bM_i, \bM_j\rangle \in \left[-100\sqrt{\frac{\log n}{d}}, 100\sqrt{\frac{\log n}{d}} \right].
\]
\end{lemma}

\begin{lemma}[$\frac{d}{n} \bM\bM^\dagger $ approximates a projection matrix]\label{lem:proj-approx} 
    With probability at least $1-2e^{-d/2}$, for all $x\in\R^n$,
    \[
        x^\dagger \bM \left(\bM^\dagger \bM\right)^{-1}\bM^\dagger x = \left(1\pm O\left(\sqrt{\frac{d}{n}}\right)\right)\frac{d}{n} x^\dagger \bM\bM^\dagger x.
    \]
    Note: $\bM\left(\bM\bM^{\dagger}\right)^{-1}\bM^{\dagger}x$ is the projection matrix onto the column space of $\bM$.
\end{lemma}
\begin{proof}
    By \pref{lem: singular-val-Gauss}, except with probability $2e^{-d/2}$ all singular values of $\bM$ lie in the interval
    \[
        \left[\frac{\sqrt{n} - 2\sqrt{d}}{\sqrt{d}}, \frac{\sqrt{n}+2\sqrt{d}}{\sqrt{d}}\right],
    \]
    and hence
    \[ 
        \left\|\frac{d}{n} \bM^\dagger \bM -  I \right\|= O\left(\sqrt{\frac{d}{n}}\right).
    \]
    Rearranging the formula gives the desired claim.
\end{proof}

\begin{lemma} \label{lem: frob-of-M}
    With probability at least $1-2e^{-t^2/2}$,
    \[
        \left\|\bM\bM^\dagger \right\|_F^2 \ge \left(1-4\frac{\sqrt{d}+t}{\sqrt{n}}\right)\frac{n^2}{d}
    \]
\end{lemma}
\begin{proof}
    Recall that by \pref{lem: singular-val-Gauss}, except with probability $2e^{-t^2/2}$ all singular values of $\bM$ lie in the interval
    \[
        \left[\frac{\sqrt{n}-\sqrt{d}-t}{\sqrt{d}}, \frac{\sqrt{n}+\sqrt{d}+t}{\sqrt{d}}\right]
    \]
    and hence
    \[
        \left\|\bM\bM^\dagger \right\|_F^2 = \sum_{1\leq i\leq d}    \lambda_i^2\left(\bM\bM^T\right)\geq d\cdot \left(\frac{\sqrt{n}-\sqrt{d}-t}{\sqrt{d}}\right)^4\geq \left(1-4\frac{\sqrt{d}+t}{\sqrt{n}}\right)\frac{n^2}{d}.
    \]
\end{proof}

\subsection{Degree-$2$ Pseudoexpectation for \SBV}
We call the following problem $\SBV$.  Given a $n\times d$ matrix $\bM$ where each entry is independently sampled from $\calN\left(0,\frac{1}{d}\right)$, certify an upper bound on $\max_{x\in\{\pm1\}^n}x^{\dagger}\bM\bM^{\dagger}x$. 
Let $\bM$ be a $n\times d$ matrix where each entry is independently sampled from $\calN\left(0,\frac{1}{d}\right)$.  The degree-$2$ Sum-of-Squares relaxation is as follows:
\[
    \max_{\pE~\text{degree-$2$}}\pE[x^{\dagger}\bM\bM^{\dagger}x] \qquad\text{s.t.}~\pE[x_i^2]=1.
\]

\begin{lemma}   \label{lem:deg-2-SBV}
    Except with probability $n^{-90}$, there is a degree-$2$ pseudoexpectation $\pE$ with pseudomoment matrix $\calM$ such that its maximum magnitude off-diagonal entry is at most $100\sqrt{\frac{\log n}{d}}$, the $\ell_2$ norms of its rows are bounded by $\sqrt{\frac{n\log n}{d}}$, its spectral norm is bounded by $1.2\frac{n}{d}$, and
    \[
        \frac{d}{n}\pE[x^{\dagger}\bM\bM^{\dagger}x] \ge \left(1-O\left(\sqrt{\frac{\log n}{d}}\right)-O\left(\sqrt{\frac{d}{n}}\right)\right)n.
    \]
\end{lemma}
\begin{proof}
A degree-$2$ pseudoexpectation $\pE$ (that is due to \cite{MS16}) can be constructed in the following way.  Let $\gamma \coloneqq 100\sqrt{\frac{\log n}{d}}$.
\[
    \pE[x^S] =
    \begin{cases}
        1 &\text{when $|S|=0$}\\
        0 &\text{when $|S|=1$}\\
        \left(1-\gamma\right)(\bM\bM^{\dagger})[i,j] &\text{when $S=\{i,j\}$}
    \end{cases}
\]
The pseudomoment matrix $\calM$ of $\pE$ can thus be written as
\[
    \begin{bmatrix}
        1 & 0 \\
        0 & (1-\gamma)\bM\bM^{\dagger}+\bD
    \end{bmatrix}
\]
where $\bD$ is some diagonal matrix.

It remains to prove that $\pE$ is a valid Boolean pseudoexpectation.  It is clear that $\pE$ satisfies the Booleanness and symmetry constraints.  It remains to prove that $\calM$ is PSD.  And to do so, it suffices to show that $(1-\gamma)\bM\bM^{\dagger}+\bD$ is PSD.  $\bD[i,i] = 1-(1-\gamma)\bM\bM[i,i]$.  From \pref{cor: norm-concentration} along with a union bound over all diagonal entries of $\bD$ we can conclude that for all $i\in[n]$, $1\ge \bD[i,i] \ge 0$ with probability at least $1-n^{-99}$ which means $\bD$ is PSD.  $(1-\gamma)\bM\bM^{\dagger}$ is clearly PSD, which means $\calM$ is PSD.

Next, we determine the objective value attained by $\pE[\cdot]$.
\begin{align*}
    \frac{d}{n}\pE[x^{\dagger}\bM\bM^{\dagger}x] &= \frac{d}{n}\langle \bM\bM^{\dagger}, (1-\gamma)\bM\bM^{\dagger}+\bD\rangle\\
    &= \frac{d}{n}\left((1-\gamma)\langle\bM\bM^{\dagger},\bM\bM^{\dagger}\rangle+\langle\bM\bM^{\dagger},\bD\rangle\right)\\
    &\ge \frac{d}{n}(1-\gamma)\|\bM\bM^{\dagger}\|_F^2.
\end{align*}
From \pref{lem: frob-of-M}, the above is at least $(1-\gamma)\left(1-O\left(\sqrt{\frac{d}{n}}\right)\right)n$ except with probability at most $n^{-100}$.

Finally, we establish bounds on the maximum absolute off-diagonal entry, the row norm, and the spectral norm of $\calM$.

From \pref{cor: inner-product-ortho} except with probability $n^{-100}$ all off-diagonal entries of $\calM$ are bounded in magnitude by $100\sqrt{\frac{\log n}{d}}$; combined with the fact that the diagonal entries are equal to $1$, we see that the $\ell_2$ norm of each row is bounded by $\sqrt{\frac{n\log n}{d}}$.  The spectral norm of $\|\bM\bM^{\dagger}\|$ is bounded by $1.1\frac{n}{d}$ and each $\bD[i,i]$ is between $0$ and $1$ except with with probability at most $n^{-100}$.  Thus, the spectral norm of $\calM$ is bounded by $1.2\frac{n}{d}$ except with probability at most $n^{-100}$. 
\end{proof}

\subsection{Degree-$2$ Pseudoexpectation for the Sherrington--Kirkpatrick Hamiltonian}  \label{sec:boolean-vec-ref-to-SK}
Recall that $\bG\sim\GOE(n)$ and $\bM$ is a $n\times d$ matrix where each entry is independently sampled from $\calN\left(0,\frac{1}{d}\right)$.

\begin{theorem} \label{thm:SK-deg-2}
    With probability $1-o_n(1)$, there is a degree-$2$ Boolean pseudoexpectation $\pE$ such that
    \begin{align*}
        \frac{1}{n^{3/2}}\pE[x^{\dagger}\bG{x}]\ge 2-o_n(1).
    \end{align*}
    The pseudomoment matrix $\calM$ satisfies the following:
    \begin{enumerate}
        \item The off-diagonal entries of $\calM$ are bounded in magnitude by $100\sqrt{\frac{\log n}{n^{.99}}}$.
        \item The $\ell_2$ norms of rows of $\calM$ are bounded by $\sqrt{n^{.01}\log n}$.
        \item The spectral norm of $\calM$ is at most $1.2n^{.01}$.
    \end{enumerate}
\end{theorem}
Towards proving \pref{thm:SK-deg-2} we first recall the following facts from random matrix theory.
\begin{fact}[{\cite[Sec. 1.14]{Erd_s_2011}}]
    The empirical distribution of eigenvalues of any $\bG\sim\GOE(n)$ follows a universal pattern, namely the \emph{Wigner Semicircle Law}.  For any real numbers $a\leq b$,
    \[
        \frac{1}{n}\# \{i: \lambda_i\in [a,b ] \}  = (1\pm o_n(1))\int_{a}^{b} \rho_{sc}(x) dx
    \]
    with probability $1-o_n(1)$, where $\rho_{sc}(x) \coloneqq \frac{1}{2\pi}\sqrt{\max(4-x^2, 0)}$.
\end{fact}

\begin{corollary}   \label{cor:eigen-big}
    For every $\eps > 0$, there is $\delta>0$ such that $\lambda_{\delta n}(\bG)\ge (2-\eps)\sqrt{n}$ with probability $1-o_n(1)$.  In particular $\lambda_{n^{.99}}\ge (2-o_n(1))\sqrt{n}$.
\end{corollary}

\begin{lemma} \label{lem: input-space-is-random}
    The distribution of the column space of $\bM$ is that of a $d$-dimensional uniformly random subspace in $\R^n$.
\end{lemma}

\begin{lemma}[\cite{O'Rourke:2016:ERM:2988556.2988772}] \label{lem: GOE-is-also-random} 
    Let $\bG\sim\GOE(n)$.  Its sequence of normalized eigenvectors $\bv_1,\bv_2,...,\bv_n$ has the same distribution as choosing a uniformly random orthonormal basis of $\R^n$, i.e., the distribution of first choosing unit $\bv_1$ uniformly at random on $\bbS^{n-1}$, then choosing unit $\bv_2$ uniformly at random orthogonal to $\bv_1$, then choosing unit $\bv_3$ uniformly at random orthogonal to $\mathrm{span}\{v_1,v_2\}$ and so on.
\end{lemma}

\begin{lemma}   \label{lem:random-subspace-pseudoexp}
    Let $\bV$ be a uniformly random subspace of $\R^n$ of dimension $d$, and let $\Pi_{\bV}$ be the projection matrix onto $\bV$.  With probability $1-o_n(1)$ there is a degree-$2$ pseudoexpectation operator $\pE_{\bV}[\cdot]$ over polynomials in $x$ on the hypercube $\{\pm 1\}^n $  such that
    \[
       \pE_{\bV}\left[x^{\dagger}\Pi_{\bV}x\right] \ge (1-o_n(1))n.
    \]
    Additionally, the pseudomoment matrix of $\pE$ satisfies identical bounds on its off-diagonal entries, its row norms and its spectral norm as $\calM$ from the statement of \pref{lem:deg-2-SBV}.
\end{lemma}

\begin{proof}[Proof of \pref{lem:random-subspace-pseudoexp}] 
Let $\bM$ be a random $n\times d$ matrix where each entry is sampled from $\calN\left(0,1/d\right)$.  Consider the degree-$2$ pseudoexpectation $\pE_{\bM}$ for \SBV~on input $\bM$ given by \pref{lem:deg-2-SBV}.  By \pref{lem:proj-approx}, with probability $1-o_n(1)$
\begin{align*}
    \frac{d}{n}\pE_{\bM} [x^\dagger\bM(\bM ^\dagger \bM)^{-1}\bM^\dagger x] &\geq (1-o_n(1))\frac{d}{n}\pE^{(2)}_{\bM} [x^\dagger \bM \bM^\dagger x] \\
    &\geq (1-o_n(1))n 
\end{align*}
By \pref{lem: input-space-is-random}, $\bM\left(\bM^{\dagger}\bM\right)^{-1}\bM^{\dagger}$ and $\Pi_{\bV}$ are identically distributed and hence we are done.
\end{proof}
We are now ready to prove 
\pref{thm:SK-deg-2}.
\begin{proof}[Proof of \pref{thm:SK-deg-2}]
Let $\{\blambda_1,...,\blambda_{n^{.99}}\}$ be the top $\delta n$ eigenvalues of $\bG$, let $\bV$ be the subspace spanned by the top $n^{.99}$ eigenvectors of $\bG$, and let $\Pi_{\bV}$ be the projection matrix onto $\bV$.  By \pref{lem: GOE-is-also-random}, $\bV$ is a uniformly random $n^{.99}$-dimensional subspace of $\R^n$.  Let $\pE_{\bV}$ be the promised pseudoexpectation from \pref{lem:random-subspace-pseudoexp}.
\begin{align*}
    \frac{1}{n^{3/2}}\pE_{\bV}[x^\dagger \bG x]&\geq \pE_{\bV}\left[\frac{\blambda_{n^{.99}}}{n^{3/2}}\langle \Pi_{\bV},xx^{\dagger} \rangle \right] + \pE_{\bV}\left[\frac{\lambda_{\min}(\bG)}{n^{3/2}}\langle\Pi_{\bV^{\perp}},xx^{\dagger}\rangle\right]&\text{(by spectral theorem)} \\
    &\ge (1-o_n(1)) \frac{\blambda_{n^{.99}}}{n^{3/2}}  \pE_{\bV}\left[x^\dagger \Pi_{\bV} x\right] - o_n(1)  \\
    &\geq (1-o_n(1)) \frac{\blambda_{n^{.99}}}{\sqrt{n}} - o_n(1) &\text{(by \pref{lem:random-subspace-pseudoexp})}\\
    &\geq 2-o_n(1).  &\text{(by \pref{cor:eigen-big})} 
\end{align*}
The bounds on off-diagonal entries, row norms and spectral norm of the pseudomoment matrix of $\pE_{\bV}$ follow by plugging in $d=n^{.99}$ into the bounds from \pref{lem:random-subspace-pseudoexp}.
\end{proof}

\subsection{Wrap-up}
The degree-$4$ Sum-of-Squares lower bound is then an immediate consequence of \pref{thm:SK-deg-2} and our lifting theorem \pref{thm:main-lifting}/\pref{thm:main-lift-obj-val}
\begin{theorem}[Restatement of \pref{thm:sk-main}] \label{thm:sk-main-restate}
    Let $\bG\sim\GOE(n)$. With probability $1-o_n(1)$, there exists a degree-$4$ SoS SDP solution with value at least $(2 - o_n(1)) \cdot n^{3/2}$.
\end{theorem}

\section{Degree-$4$ SoS Lower Bound for $\maxcut$ in random $d$-regular graphs}
In this section, we first give a degree-$2$ pseudoexpectation for $\maxcut$ in random $d$-regular graphs, which is used as a ``seed'' to derive a degree-$4$ pseudoexpectation from \pref{thm:main-lifting} and \pref{thm:main-lift-obj-val}.

This degree-$2$ pseudoexpectation is only a slight variant of the known construction of \cite{CGHV15,MS16}.
\begin{theorem} \label{thm:deg-2-seed-maxcut}
    Let $\bG$ be a random $d$-regular graph.  For every constant $\eps > 0$ with probability $1-o_n(1)$ there is a degree-$2$ Boolean pseudoexpectation $\pE$ such that:
    \[
        \pE[x^{\dagger}(-A_{\bG})x] \ge (1-2\eps-o_n(1))2\sqrt{d-1}n.
    \]
    Additionally, the pseudomoment matrix $\calM$ of $\pE$ satisfies the following:
    \begin{enumerate}
        \item Its row norms are bounded by a constant $\gamma(\eps)$ which only depends on $\eps$.
        \item Its spectral norm is bounded by constant $\gamma'(\eps)$ which only depends on $\eps$.
        \item Its off-diagonal entries are bounded in magnitude by $\frac{\gamma''(\eps)}{\sqrt{d}}$ where $\gamma''(\eps)$ is some constant that only depends on $\eps$.
    \end{enumerate}
\end{theorem}
We first develop some tools and then prove \pref{thm:deg-2-seed-maxcut} in \pref{sec:maxcut-wrapup}.

\subsection{The \cite{CGHV15,MS16}
 construction}  \label{sec:gaussian-wave}
We first revisit the degree-$2$ pseudoexpectation for Max Cut due to \cite{CGHV15,MS16}.  Given a random $d$-regular graph $\bG$ on $n$ vertices, we state the moment matrix of a degree-$2$ pseudoexpectation.  We call a vertex \emph{$C$-good} if its radius-$(2C+1)$ neighborhood is a tree, and \emph{$C$-bad} otherwise.

First, we define vector $x_v$ corresponding to vertex $v$.  Let $\rho,C,\alpha$ be constants that we'll set later.  If $v$ is $C$-bad, then we let 
\begin{align*}
    \bx_v[u]\coloneqq
    \begin{cases}
        1 & \text{if $u=v$}\\
        0 & \text{if $u\ne v$},
    \end{cases}
\end{align*}
otherwise, we let
\begin{align*}
    \bx_v[u] \coloneqq
    \begin{cases}
        \displaystyle\alpha\cdot\rho^{d_{\bG}(u,v)} &\text{if $d_{\bG}(u,v)\le C$}\\
        0 &\text{otherwise.}
    \end{cases}
\end{align*}
Finally, we also define a vector $\bx_{\emptyset}$ which is orthogonal to all $\{\bx_v\}_{v\in\bG}$.

Once $\rho,C$ are chosen, we pick $\alpha$ so that the vectors $\bx_v$ for $C$-good $v$ have unit norm.  The degree-$2$ pseudomoment matrix $\calM$ is indexed by pairs of sets $S, T$ such that $|S|,|T|\le 1$ and is defined as follows:
\begin{align*}
    \calM[S,T] \coloneqq \langle \bx_S, \bx_T \rangle.
\end{align*}

A nice feature of this solution is that one can derive a closed form for $\langle \bx_v, \bx_w\rangle$ when $\{v,w\}$ is an edge between two $C$-good vertices.
\begin{lemma}   \label{lem:edge-corr}
    Let $\{v,w\}$ be an edge in $\bG$.  Then
    \[
        \langle \bx_v, \bx_w \rangle =
        \begin{cases}
            2\cdot\left(\frac{d-1}{d}\right)\cdot\rho\cdot\left(1-\alpha^2\rho^{2C}d(d-1)^{C-1}\right) & \text{if $v,w$ are both $C$-good}\\
            0 & \text{otherwise}
        \end{cases}
    \]
\end{lemma}
\begin{proof}
    If either $v$ or $w$ is $C$-bad, then it is clear that $\langle \bx_v, \bx_w \rangle = 0$.  Thus, we assume they are both $C$-good.
    
    \begin{align*}
        \langle \bx_v, \bx_w \rangle &= \sum_{u\in V(\bG)} \bx_v[u]\cdot \bx_w[u] \\
        &= \sum_{\substack{u \in V(\bG) \\ d_{\bG}(u,v) < d_{\bG}(u,w) \le C}} \bx_v[u]\cdot \bx_w[u] + \sum_{\substack{u \in V(\bG) \\ d_{\bG}(u,w) < d_{\bG}(u,v) \le C}} \bx_v[u]\cdot \bx_w[u] \\
        &= \alpha^2 \rho\cdot\left(\sum_{\substack{u \in V(\bG) \\ d_{\bG}(u,v) < d_{\bG}(u,w) \le C}} \rho^{2d_{\bG}(u,v)} + \sum_{\substack{u \in V(\bG) \\ d_{\bG}(u,w) < d_{\bG}(u,v) \le C}} \rho^{2d_{\bG}(u,w)} \right) \\
        &= \alpha^2\rho\cdot\left(2\sum_{\ell=0}^{C-1} \rho^{2\ell}(d-1)^{\ell}\right)\\
        &= 2\alpha^2\rho\cdot\left(\frac{d-1}{d}\right)\cdot\left(\frac{1}{\alpha^2}-\rho^{2C}d(d-1)^{C-1}\right) &\text{(since $\bx_v$ has unit norm)}\\
        &= 2\cdot\left(\frac{d-1}{d}\right)\cdot\rho\cdot\left(1-\alpha^2\rho^{2C}d(d-1)^{C-1}\right)
    \end{align*}
\end{proof}

\begin{remark}  \label{rem:good-choice-gwave}
    For any $0 < \eps \le 1$, if we choose $\rho = -\frac{1-\eps}{\sqrt{d-1}}$, then for an edge between $C$-good vertices $\{v,w\}$ we would have
    \[
        \langle \bx_v, \bx_w \rangle = -\frac{2\sqrt{d-1}(1-\eps)}{d}\cdot\left(1-\alpha^2\cdot\left(\frac{d}{d-1}\right)\cdot(1-\eps)^C)\right).
    \]
    One can make $(1-\eps)^C$ arbitrarily small by increasing $C$, and additionally, increasing $C$ only makes $\alpha$ smaller.  Further, since $\frac{d}{d-1}\le\frac{3}{2}$ for $d\ge 3$, there exists a choice for $C$ depending only on $\eps$ such that
    \[
        \langle \bx_v, \bx_w \rangle \le -(1-2\eps)\frac{2\sqrt{d-1}}{d}.
    \]
\end{remark}

For the purposes of our proof, we will also need bounds on $|\langle \bx_v, \bx_w \rangle|$ when $v$ and $w$ are within distance $C$ of each other.  A similar calculation to that in the proof of \pref{lem:edge-corr} lets us show:
\begin{lemma}   \label{lem:inner-product-bound}
    Let $v$ and $w$ be any two vertices.  We have
    \[
        |\langle \bx_v, \bx_w \rangle| \le
        \begin{cases}
            |\rho|^{d_{\bG}(v,w)}(d_{\bG}(v,w)+1) &d_{\bG}(v,w)\le C\\
            0 &\text{otherwise}
        \end{cases}
    \]
\end{lemma}
\begin{proof}
    If $v$ or $w$ are $C$-bad, then $\langle \bx_v, \bx_w\rangle = 0$, in which case the bound holds.  Thus, for the rest of the proof we will assume $v$ and $w$ are both $C$-good.  Let $a$ be a $C$-good vertex and $b$ be a vertex with distance at most $C$ from $a$.  We use $P_{ab}$ denote the unique path of length at most $C$ between vertices $a$ and $b$.
    \begin{align*}
        \langle \bx_v, \bx_w \rangle &= \sum_{u\in V(\bG)} \bx_v[u]\cdot\bx_w[u] \\
        &= \alpha^2\sum_{s\in P_{vw}}\sum_{\substack{u\in V(\bG)\\ d_{\bG}(u,v),~d_{\bG}(u,w)\le C\\ s\in P_{vu},~s\in P_{wu}}} \rho^{d_{\bG}(v,w)} \rho^{2d_{\bG}(s,u)}\\
        &\le \sum_{s\in P_{vw}} |\rho|^{d_{\bG}(v,w)} \sum_{\ell=0}^C d(d-1)^{\ell-1} \rho^{2\ell}\\
        &= \sum_{s\in P_{vw}} |\rho|^{d_{\bG}(v,w)} &\text{(since $\bx_v$ has unit norm)}\\
        &= |\rho|^{d_{\bG}(v,w)}\cdot(d_{\bG}(v,w)+1)
    \end{align*}
\end{proof}

\subsection{Nonbacktracking Polynomials}
We define a sequence of polynomials $g_0,g_1,\dots$ which we call \emph{nonbacktracking polynomials} below (see, for example, \cite{ABLS07}):
\begin{definition}
    Let the \emph{nonbacktracking polynomials} be the following sequence of polynomials defined recursively below.
    \begin{align*}
        g_0(x) &= 1\\
        g_1(x) &= x\\
        g_2(x) &= x^2 - d\\
        g_t(x) &= xg_{t-1}(x) - (d-1)g_{t-2}(x) &\text{for $t \ge 3$.}
    \end{align*}
\end{definition}
An elementary fact about nonbacktracking polynomials, which earns them their name is:
\begin{fact}
    For any $d$-regular graph $G$,
    $
        g_i\left(A_G\right)_{uv} = \#\text{ of nonbacktracking walks from $u$ to $v$.}
    $
\end{fact}
We will be interested in $g_i(\lambda)$ for eigenvalues $\lambda$ of $A_{\bG}$.  The following can be extracted from \cite[Proof of Lemma 2.3]{ABLS07}:
\begin{lemma}   \label{lem:nb-poly-bound}
    When $x\in[-2\sqrt{d-1},2\sqrt{d-1}]$, $|g_i(x)|\le 2(i+1)\sqrt{(d-1)^i}$.
\end{lemma}
By a simple continuity argument, this implies:
\begin{corollary}   \label{cor:nb-poly-bound}
    For any $\eps > 0$, there exists $\delta > 0$ such that $|g_i(x)|\le 2(i+1)\sqrt{(d-1)^i}+\eps$ when $x\in[-2\sqrt{d-1}-\delta,2\sqrt{d-1}+\delta]$.
\end{corollary}

\subsection{Random graphs}
We need the following two facts about random regular graphs.
\begin{lemma}[Easy consequence of {\cite[Theorem 2.5]{Wor99}}]   \label{lem:few-bad-vertices}
    Let $d\ge 3$ be a fixed constant, let $\bG$ be a random $d$-regular graph on $n$ vertices, and let $C$ be any constant.  Then w.h.p. the number of $C$-bad vertices in $\bG$ is $O(\log n)$.
\end{lemma}

\begin{theorem}[Friedman's theorem {\cite{Fri08,Bor19}}] \label{thm:friedman}
    Let $d\ge 3$ be a fixed constant, and let $\bG$ be a random $d$-regular graph on $n$-vertices.  Then with probability $1-o_n(1)$:
    \[
        \max\{\lambda_2(\bG),|\lambda_n(\bG)|\} \le 2\sqrt{d-1}+o_n(1).
    \]
\end{theorem}

\subsection{Construction}   \label{sec:construction-max-cut}
\paragraph{Stage 1.}   First choose constant $\eps > 0$, and let $\rho,C,\alpha$ be chosen according to \pref{rem:good-choice-gwave} so that each $\bx_v$ is a unit vector, and $\langle \bx_v,\bx_w\rangle\le -(1-2\eps)\frac{2\sqrt{d-1}}{d}$ for every edge $\{v,w\}$ between two $C$-good vertices $v$ and $w$.  Next, define polynomial $g$ as follows:
\[
    g(x) \coloneqq \alpha \sum_{i=0}^C \rho^i g_i(x).
\]

\paragraph{Stage 2.}   Let $\displaystyle\calW\coloneqq g\left(A_{\bG}\right)^2 - g(d)^2\cdot\left(\frac{\vec{1}\vec{1}^{\dagger}}{n}\right)$.
\begin{claim}
    $\calW\psdge 0$.
\end{claim}
\begin{proof}
    Let $d=\lambda_1(\bG)\ge\dots\ge\lambda_n(\bG)$ denote the eigenvalues of $A_{\bG}$ in decreasing order.  Decomposing $A_{\bG}$ in its eigenbasis lets us write
    \begin{align*}
        A_{\bG} = \lambda_1 v_1 v_1^{\dagger} + \dots + \lambda_n v_n v_n^{\dagger}
    \end{align*}
    where $v_1 = \frac{\vec{1}}{\sqrt{n}}$ and $v_1,\dots,v_n$ are an orthonormal basis.
    Consequently,
    \[
        g(A_{\bG})^2 = g(d)^2  \frac{\vec{1}\vec{1}^{\dagger}}{n} + \dots + g(\lambda_n)^2 v_n v_n^{\dagger},
    \]
    which means
    \[
        \calW = \sum_{i=2}^n g(\lambda_i)^2 v_i v_i^{\dagger},
    \]
    which is positive semidefinite since each $g(\lambda_i)^2$ is nonnegative.
\end{proof}

\paragraph{Stage 3.}   Let $S_{\bG}$ be the collection of $C$-bad vertices in $\bG$.  Let $\calW'$ be the matrix obtained by zeroing out all rows and columns in $S_{\bG}$ and then setting $\calW'[v,v]$ to $1$ for all $v\in V({\bG})$.  Symbolically,
\[
    \calW'[v,w] \coloneqq
    \begin{cases}
        1 &\text{if $v = w$}\\
        \calW[v,w] &\text{if $v\ne w$ and $v,w\notin S_{\bG}$}\\
        0 &\text{otherwise}
    \end{cases}
\]
\begin{remark}
    $\calW'$ is a PSD matrix since it is a $2\times 2$ block diagonal matrix where each block is PSD.  In particular one block, $\calW'\left[S_{\bG},S_{\bG}\right]$, is an identity matrix and is thus PSD.  The other block can be seen to satisfy:
    \[
        \calW' \left[ V(\bG)\setminus S_{\bG}, V(\bG)\setminus S_{\bG} \right] \psdge \calW\left[V(\bG)\setminus S_{\bG}, V(\bG)\setminus S_{\bG} \right].
    \]
    Thus, the other block is also PSD since it PSD-dominates a principal submatrix of the PSD matrix $\calW$.
\end{remark}

\begin{remark}
    Note that while the vectors $\{\bx_u\}_{u\in V(\bG)}$ didn't play an explicit role in the construction, they have a role in the analysis.
\end{remark}

\subsection{Various norm bounds}
In this section, we give bounds on the $\ell_2$ norm of a subset of indices of rows/columns of $\calW'$ and the spectral norm of $\calW'$.
\begin{observation} \label{obs:dot-product-vs-calW}
    For any pair of vertices $v,w$, $\left|\langle \bx_v, \bx_w\rangle-\calW'[v,w]\right| \le \frac{\kappa(\eps,d)}{n}$ where the $\kappa(\eps,d)$ is a constant depending on $\eps$ and $d$.
\end{observation}

\begin{lemma}   \label{lem:max-cut-row-bound}
    Let $\calW'[u]$ be the $u$-th row of $\calW'$.  Then when $n$, the number of vertices in the graph is large enough,
    \[
        \|\calW'[u,V(\bG)\setminus\{u\}]\|_2\le\gamma(\eps)
    \]
    where $\gamma(\eps) > 0$ is a constant that depends only on $\eps$ chosen in Stage 1 of the construction in \pref{sec:construction-max-cut}.
\end{lemma}
\begin{proof}
    If $u$ is $C$-bad, then $\|\calW'[u]\|_2 = 1$.  When $u$ is $C$-good,
    \begin{align*}
        \|\calW'[u]\|_2^2 &= \sum_{v\in V(\bG)} \calW'[u,v]^2\\
        &\le \sum_{v\in V(\bG)} \left(\langle \bx_u,\bx_v\rangle + \frac{\kappa(\eps,d)}{n}\right)^2 &\text{(via \pref{obs:dot-product-vs-calW})}\\
        &\le 1 + \sum_{\ell=1}^C d(d-1)^{\ell-1} \left(\rho^{\ell}(\ell+1) + \frac{\kappa(\eps,d)}{n}\right)^2 &\text{(via \pref{lem:inner-product-bound})}\\
        &= 1 + \frac{d}{d-1}\sum_{\ell=1}^C\Bigg[ (1-\eps)^{2\ell}(\ell+1)^2\\&+2(-1)^{\ell}(1-\eps)^{\ell}\sqrt{(d-1)^{\ell-1}}\cdot\frac{\kappa(\eps,d)}{n}+\frac{\kappa(\eps,d)^2}{n^2}\Bigg]
        &\text{(plugging in $\rho$)}
    \end{align*}
    We bound the $3$ terms above separately.  First, note that
    \begin{align*}
        1+\sum_{\ell=1}^C(1-\eps)^{2\ell}(\ell+1)^2
    \end{align*}
    can be upper bounded by a constant $\gamma_1(\eps)$ that depends only on $\eps$ (since as we noted in \pref{rem:good-choice-gwave} $C$ depends only on $\eps$).  Next,
    \begin{align*}
        \sum_{\ell=1}^C 2(-1)^{\ell}(1-\eps)^\ell\sqrt{(d-1)^{\ell}}\cdot\frac{\kappa(\eps,d)}{n}
    \end{align*}
    is bounded by $\frac{\kappa_1(\eps,d)}{n}$ where $\kappa_1(\eps,d)$ is a constant depending on $\eps$ and $d$.  And finally,
    \begin{align*}
        \sum_{\ell=1}^C \frac{\kappa(\eps,d)^2}{n^2} \le \frac{\kappa_2(\eps,d)}{n^2}
    \end{align*}
    for constant $\kappa_2(\eps,d)$ depending only on $\eps$ and $d$.  Thus,
    \[
        \|\calW'[u,V(\bG)\setminus\{u\}]\| \le \gamma_1(\eps)+\frac{\kappa_1(\eps,d)}{n}+\frac{\kappa_2(\eps,d)}{n^2}
    \]
    and for $n$ large enough, we can bound the above by a constant $\gamma(\eps)$ depending on $\eps$ and not on $d$.
\end{proof}
Next, we upper bound the spectral norm of $\calW'$.
\begin{lemma}   \label{lem:max-cut-spec-bound}
    When $n$, the number of vertices in $V(\bG)$ is large enough, $\|\calW'\| \le \gamma'(\eps)$ where $\gamma'(\eps)$ is a constant that depends only on $\eps$ chosen in Stage 1 of the construction in \pref{sec:construction-max-cut}.
\end{lemma}
\begin{proof}
    First, recall the notation $S_{\bG}$ to denote the set of $C$-bad vertices in $\bG$ and that up to permutation of rows and columns, $\calW'$ has the following block diagonal structure:
    \begin{align*}
        \calW' =
        \begin{bmatrix}
            A & 0 \\
            0 & B
        \end{bmatrix}
    \end{align*}
    where $A = \calW[V(\bG)\setminus S_{\bG}, V(\bG)\setminus S_{\bG}] + \frac{g(d)^2}{n}\cdot\Id$
    and $B$ is an identity matrix.  Thus, $\|\calW'\|\le\max\{\|A\|,\|B\|\}$.  We already know that $\|B\|\le 1$, and thus it remains to obtain a bound on $\|A\|$.
    \begin{align*}
        \|A\| &= \left\|\calW[V(\bG)\setminus S_{\bG}, V(\bG)\setminus S_{\bG}]\right\| + \frac{g(d)^2}{n} \\
        &\le \|\calW\| + o_n(1) \\
        &= \left\|\sum_{i=2}^n g(\lambda_i)^2 v_iv_i^{\dagger}\right\| + o_n(1) \\
        &\le \max_{i\in\{2,\dots,n\}} g(\lambda_i(\bG))^2 + o_n(1).
    \end{align*}
    Now, recall Friedman's theorem \pref{thm:friedman}, according to which whp $\lambda_2(\bG),\dots,\lambda_n(\bG)$ are all in $[-2\sqrt{d-1}-o_n(1),2\sqrt{d-1}+o_n(1)]$.  Thus it suffices to bound $|g(x)|$ on the specified interval.  
    For the below calculation, assume $x\in[-2\sqrt{d-1}-o_n(1),2\sqrt{d-1}+o_n(1)]$.
    \begin{align*}
        |g(x)| &\le \alpha\sum_{i=0}^C \left(\frac{1-\eps}{\sqrt{d-1}}\right)^i|g_i(x)|\\
        &\le \alpha\sum_{i=0}^C 2(i+1)\left(\frac{1-\eps}{\sqrt{d-1}}\right)^i\sqrt{(d-1)^i} + o_n(1) &\text{(by \pref{cor:nb-poly-bound})}\\
        &\le 2\alpha\sum_{i=0}^C (i+1)(1-\eps)^i + o_n(1)
    \end{align*}
    which bounds $\|A\|$ by a constant $\gamma'(\eps)$ only depending on $\eps$ (as $C$ also depends only on $\eps$) when $n$ is large enough.
\end{proof}

\subsection{$\maxcut$ Wrap-Up}  \label{sec:maxcut-wrapup}
We are now finally ready to prove \pref{thm:deg-2-seed-maxcut} and \pref{thm:max-cut-main}.
\begin{proof}[Proof of \pref{thm:deg-2-seed-maxcut}]
    Define $\pE$ in the following way:
    \begin{align*}
        \pE[x^S] =
        \begin{cases}
            1 &\text{if $|S|=0$}\\
            0 &\text{if $|S|=1$}\\
            \calW'[u,v] &\text{if $S=\{u,v\}$.}
        \end{cases}
    \end{align*}
    Its pseudomoment matrix is then
    \[
        \calM = \begin{bmatrix}
            1 & 0\\
            0 & \calW'
        \end{bmatrix}
    \]
    and hence is PSD.  The bounds on the row norms and spectral norm on $\calM$ follow from \pref{lem:max-cut-row-bound} and \pref{lem:max-cut-spec-bound} respectively and the bound on the magnitude of off-diagonal entries follows from \pref{lem:inner-product-bound} and \pref{obs:dot-product-vs-calW}.  Finally, we show that the objective value is indeed at least $(1-2\eps-o_n(1))2\sqrt{d-1}n$.  Our choice of parameters combined with \pref{obs:dot-product-vs-calW} tells us that $\pE[x_ux_v]\le -(1-2\eps-o_n(1))\frac{2\sqrt{d-1}}{d}$ for edges $\{u,v\}$ between $C$-good vertices.  Since we additionally know that the number of $C$-bad vertices is $O(\log n)$, the fraction of edges that are between $C$-good vertices is $1-o_n(1)$.  Consequently, it follows that
    \[
        \pE[x^{\dagger}(-A_{\bG})x] \ge (1-2\eps-o_n(1))2\sqrt{d-1}n.
    \]
\end{proof}


\begin{theorem}[Restatement of \pref{thm:max-cut-main}]  \label{thm:max-cut-main-restatement}
    Let $\bG$ be a random $d$-regular graph.  For every constant $\eps > 0$ with probability $1-o_n(1)$, there is a degree-$4$ SoS SDP solution with $\maxcut$ value at least \[
        \frac{1}{2} + \frac{\sqrt{d-1}}{d} \left(1-2\eps - \frac{\gamma(\eps)}{d^{1/2}}\right)
    \]
    for some constant $\gamma$ that depends only on $\eps$.
\end{theorem}
\begin{proof}
    By applying our lifting theorem \pref{thm:main-lifting}/\pref{thm:main-lift-obj-val} to the degree-$2$ pseudoexpectation $\pE_2$ from \pref{thm:deg-2-seed-maxcut}, we obtain a degree-$4$ pseudoexpectation $\pE_4$ such that
    \[
        \pE_4[x^{\dagger}(-A_{\bG})x] \ge (1-2\eps-\frac{\gamma(\eps)}{d^{1/2}})2\sqrt{d-1}n \numberthis \label{eq:lower-bound-max-cut}
    \]
    where $\gamma(\eps)$ is a constant that depends only on $\eps$.  As a result:
    \begin{align*}
        \frac{1}{4|E(\bG)|}\pE_4[x^{\dagger}(D_{\bG}-A_{\bG})x] &= \frac{dn}{4|E(\bG)|}+\pE_4[x^{\dagger}(-A_{\bG})x]\\
        &\ge \frac{1}{2} + \frac{\sqrt{d-1}}{d} \left(1-\eps - \frac{\gamma(\eps)}{d^{1/2}}\right) &\text{(by \pref{eq:lower-bound-max-cut}).}
    \end{align*}
\end{proof}

\section*{Acknowledgments}
We would like to thank Sam Hopkins for several very insightful conversations on pseudocalibration and the pioneering paper \cite{BHKKMP19} in which it was introduced, from which many ideas in this paper draw inspiration. We would also like to thank Tselil Schramm and Nikhil Srivastava for valuable discussions related to random matrix theory.  S.M.\ is grateful to Ryan O'Donnell for numerous technical conversations about random graphs and semidefinite programming, and to Vijay Bhattiprolu for such conversations on Sum-of-Squares lower bounds. J.X is thankful to  Siu On Chan for insightful discussions on SoS lower bounds.

\bibliographystyle{alpha}
\bibliography{main}

\appendix
\section{Limits of graphical matrices}  \label{app:lim-graph-mat}

In this section, we prove \pref{claim:graph-mat-limit} and \pref{lem:graph-mat-explicit}.  We begin by first proving a couple of technical lemmas.

\begin{lemma}[The graphical polynomial of a well-glued glyph mimics inner products.]
\label{lem: mimic-inner-product}
    For any well-glued glyph $\Glyph$ whose right-hand-side vertices all have degree $2$, let $\mathcal{P}$ be the collection of length-2 walks in $\Glyph$.  For any valid glyph labeling $S\circ T$ of $\Glyph$,
    \[
        \beta_{\Glyph,\HadDim,S\circ T} = \prod_{i\le j\in L(\Glyph)\cup R(\Glyph)} \frac{\langle M[S\circ T(i)],M[S\circ T(j)]\rangle^{|\calP_{i,j}|}}{|\calP_{i,j}|!} \pm \frac{f(n,\Glyph)}{\kappa}
    \]
    where $\calP_{i,j}\subseteq\calP$ is the collection of length-$2$ paths that have endpoints $i$ and $j$ and $f(n,\Glyph)$ is a value that depends only on $n$ and $\Glyph$ but is independent of $\HadDim$.
\end{lemma}
\begin{proof}
    We use $\calA$ to represent the set of all functions from $M(\Glyph)$ to $[\HadDim]$ and $\calA_{\distinct}$ to represent the set of all injective functions in $\calA$.  Let $\tau(\Glyph)$ denote the number of automorphisms of $\cG$ that keep $L(\Glyph)\cup R(\Glyph)$ fixed.  Observe that
    \[
        \beta_{\Glyph,\HadDim,S\circ T} = \frac{1}{\tau(\Glyph)}\sum_{\pi\in\calA_{\distinct}}\prod_{a\in L\cup R(\Glyph)}\prod_{b\in M(\Glyph)} M_{\kappa}[S\circ T(a),\pi(b)]^{\Glyph_{ab}}.
    \]
    We now define a related quantity 
    \[
        \wt{\beta}_{\Glyph,\HadDim,S\circ T} \coloneqq \frac{1}{\tau(\Glyph)}\sum_{\pi\in\calA}\prod_{a\in L\cup R(\Glyph)}\prod_{b\in M(\Glyph)} M_{\kappa}[S\circ T(a),\pi(b)]^{\Glyph_{ab}}.
    \]
    We now show that $\beta_{\Glyph,\HadDim,S\circ T}$ and $\wt{\beta}_{\Glyph,\HadDim,S\circ T}$ are equal up to additive error terms of $O\left(\frac{1}{\sqrt{\kappa}}\right)$.
    \begin{align}   \label{eq:diff-dist-not-dist}
        \wt{\beta}_{\Glyph,\HadDim,S\circ T} - \beta_{\Glyph,\HadDim,S\circ T} = \frac{1}{\tau(\Glyph)}\sum_{\pi\in\calA\setminus\calA_{\distinct}}\prod_{a\in L\cup R(\Glyph)}\prod_{b\in M(\Glyph)} M_{\kappa}[S\circ T(a),\pi(b)]^{\Glyph_{ab}}.
    \end{align}
    Each term in the RHS of \pref{eq:diff-dist-not-dist} is the product of $m(\Glyph)$ entries of $M_{\HadDim}$ where $m(\Glyph)$ is the number of edges in $\Glyph$.  The magnitude of every entry of $M$ from \pref{sec:construction} is bounded by $1$ since the rows are unit vectors, and the magnitude of every entry of $H_{\HadDim}^{\le n}$ is equal to $\frac{1}{\sqrt{\HadDim}}$, and thus the magnitude of each entry of $M_{\kappa}$ is bounded by $\frac{n}{\sqrt{\kappa}}$.  As a result, every term in the \pref{eq:diff-dist-not-dist} is bounded by
    \[
        \left(\frac{n}{\sqrt{\kappa}}\right)^{m(\Glyph)} = \frac{n^{m(\Glyph)}}{\kappa^{m(\Glyph)/2}} \numberthis \label{eq:bound-per-term}
    \]
    $\calA\setminus\calA_{\distinct}$ is the set of non-injective functions from $M(\Glyph)$ to $[\kappa]$ and hence has cardinality bounded by $|M(\Glyph)|^{|M(\Glyph)|}\cdot\kappa^{|M(\Glyph)|-1}$.  
    Since $\Glyph$ is well-glued, every middle vertex has degree at least $2$, and hence $|M(\Glyph)|\le\frac{m(\Glyph)}{2}$.  Thus, the number of terms in the sum in the RHS of \pref{eq:diff-dist-not-dist} is at most
    \[
        \frac{1}{\kappa}\cdot\left(\frac{\kappa m(\Glyph)}{2}\right)^{\frac{m(\Glyph)}{2}} \numberthis \label{eq:number-of-terms}
    \]
    The sum is then bounded by the product of the RHS of \pref{eq:bound-per-term} and \pref{eq:number-of-terms}, which is
    \[
        \frac{1}{\kappa}\left(\frac{n^2m(\Glyph)}{2}\right)^{\frac{m(\Glyph)}{2}}    \numberthis \label{eq:bound-on-diff}
    \]
    It can be verified that
    \[
        \wt{\beta}_{\Glyph,\HadDim,S\circ T} = \frac{1}{\tau(\Glyph)}\prod_{i\le j\in L(\Glyph)\cup R(\Glyph)} \langle M_{\kappa}[S\circ T(i)], M_{\kappa}[S\circ T(j)]\rangle^{|\calP_{i,j}|}
    \]
    and
    \[
        \tau(\Glyph) = \prod_{i\le j\in L(\Glyph)\cup R(\Glyph)} |\calP_{i,j}|!.
    \]
    The desired statement follows from our bound on $\wt{\beta}_{\Glyph,\HadDim,S\circ T} - \beta_{\Glyph,\HadDim,S\circ T}$ in \pref{eq:bound-on-diff}.
\end{proof}

\begin{lemma}   \label{lem:RHS-deg-4-neg}
    Let $\Glyph$ be a well-glued glyph that has some degree-$\ge 4$ right vertex.  Then each entry of $\GraphMat_{\Glyph,\HadDim}$ is entrywise bounded by 
    the quantity $\frac{f(n,\Glyph)}{\HadDim}$ in the statement of \pref{lem: mimic-inner-product}.
\end{lemma}
\begin{proof}
    Each entry can be bounded in magnitude by $O\left(\frac{f(n,\Glyph)}{\HadDim} \right)$ via an identical argument to how $\wt{\beta}_{\srib,\HadDim}-\beta_{\srib,\HadDim}$ is bounded in the proof of \pref{lem: mimic-inner-product}.
\end{proof}

The statements of \pref{claim:graph-mat-limit} and \pref{lem:graph-mat-explicit} follow from taking the $\kappa\to\infty$ limit of the statements of \pref{lem: mimic-inner-product} and \pref{lem:RHS-deg-4-neg}.

We now prove that the limits on the RHS of \pref{eq:lim-def-M1} and \pref{eq:lim-def-M2} exist.  Towards proving this, we define matrices $\calM^{(1)}_{\HadDim}$ and $\calM^{(2)}_{\HadDim}$ as follows:
\begin{align*}
    \calM^{(1)}_{\HadDim}[S,T] &= \E_{\bz\sim\{\pm1\}^{\HadDim}}[q_{S\Delta T}(\bz)]\\
    \calM^{(2)}_{\HadDim}[S,T] &= \E_{\bz\sim\{\pm1\}^{\HadDim}}[p_S(\bz)p_T(\bz)].
\end{align*}
Our definitions for $\calM^{(1)}$ and $\calM^{(2)}$ from \pref{eq:lim-def-M1} and \pref{eq:lim-def-M2} are equivalent to:
\begin{align*}
    \calM^{(1)} &\coloneqq \lim_{\HadDim\to\infty}\calM^{(1)}_{\HadDim}\\
    \calM^{(2)} &\coloneqq
    \lim_{\HadDim\to\infty}\calM^{(2)}_{\HadDim}.
\end{align*}
Hence, it suffices to prove that the limits on the RHS of the above exist.

The following can be verified.
\begin{claim}   \label{claim:moment-matrices-from-graph-mats}
    $\calM^{(1)}_{\HadDim}$ and $\calM^{(2)}_{\HadDim}$ can be expressed as a linear combination of $\HadDim$-graphical matrices of well-glued glyphs.
\end{claim}

Then as a consequence of \pref{claim:moment-matrices-from-graph-mats} and \pref{claim:graph-mat-limit} we get:
\begin{corollary}   \label{cor:moment-matrices-well-defined}
    $\lim_{\HadDim\to\infty} \calM^{(1)}_{\HadDim}$ and $\lim_{\HadDim\to\infty} \calM^{(2)}_{\HadDim}$ exist and can be expressed as a linear combination of graphical matrices of well-glued glyphs.  In particular, this establishes that $\calM^{(1)}$ and $\calM^{(2)}$ are well-defined.
\end{corollary}

\section{Pseudocalibration} \label{sec:pseudocalibration}
We recall some basic facts about Hermite polynomials first.  For our convention, let $h_i(x)$ be the $i$-th Hermite polynomial normalized so that it is monic. For $\alpha\in\bbN^S$ and $v\in\bbR^S$, we use $H_{\alpha}(v)$ to denote $\prod_{i\in S}h_{\alpha_i}(v_i)$.

\begin{fact}    \label{fact:hermite-generatingfunc}
    \[
        \exp\left(tz-\frac{t^2}{2}\right) = \sum_{i=0}^\infty \frac{1}{i!}h_i(z)t^i
    \]
\end{fact}

Suppose $\calD_0$ and $\calD_1$ are two distributions over $n\times d$ matrices. In $\calD_0$, every entry is an independent Gaussian. Let $\bV$ be a matrix where the first column is a random $\{\pm 1\}^n$ vector $\bx$, and the rest of the matrix has Gaussian entries, let $H$ be a Hadamard matrix normalized so it is unitary, and let $\bD$ be a random $\pm 1$ diagonal matrix --- $\calD_1$ is the distribution of $\bM = \bV H\bD$.

For each $\alpha\in\bbN^{n\times d}$ and $L\subseteq [n]$, we care about computing $\E_{(\bx,\bM)\sim\calD_1}\left[H_{\alpha}(\bM)\bx^L\right]$ since a relevant Hermite coefficient is $\frac{\E_{(\bx,\bM)\sim\calD_1}\left[H_{\alpha}(\bM)\bx^L\right]}{\|H_\alpha\|^2}$.

We carry out this computation below. For each $i$, let $\{\bg_{ij}\}_{1\le j\le d}$ be an independent Gaussian process where each Gaussian has variance $\frac{d-1}{d}$ and all pairwise covariances are $\frac{-1}{d}$. The expression is equal to
\begin{align*}
    &\E_{\bD,\bx,\bg} \bx^L \prod_{j=1}^d \bD_{jj}^{\sum_{1\le i\le n} \alpha_{ij}} \prod_{i=1}^n h_{\alpha_{ij}}\left(\frac{x_i}{\sqrt{d}} + \bg_{ij}\right)\\
    =& \E_{\bx} \bx^L \E_{\bg} \prod_{j=1}^d \E_{\bD}\left[\bD_{jj}^{\sum_{1\le i\le n} \alpha_{ij}} \right] \prod_{i=1}^n h_{\alpha_{ij}}\left(\frac{x_i}{\sqrt{d}} + \bg_{ij}\right)
\end{align*}
If for any $j$, $\sum_{1\le i\le n} \alpha_{ij}$ is odd, then the expression is equal to zero. Thus for the rest of the computation assume this is not the case and the goal is to now compute
\begin{align}   \label{eq:herm-coeff}
    \E_\bx \bx^L \prod_{i=1}^n \E_{\bg} \prod_{j=1}^d h_{\alpha_{ij}}\left(\frac{x_i}{\sqrt{d}} + \bg_{ij}\right)
\end{align}

We zoom into the computation of $\E_{\bg}\prod_{j=1}^d h_{\alpha_{ij}}\left(\frac{x_i}{\sqrt{d}} + \bg_{ij}\right)$. We can find vectors $v_1,\dots,v_d$, each of norm $\sqrt{\frac{d-1}{d}}$ and pairwise dot products $-\frac{1}{d}$, such that $\bg_{ij} = \langle v_j, \widetilde{\bg}_i\rangle$ where $\widetilde{\bg}_i$ is a vector of i.i.d. standard Gaussians. We use $\bz_{ij}$ to denote $\frac{x_i}{\sqrt{d}}+\bg_{ij}$. On one hand, using \pref{fact:hermite-generatingfunc}, we have
\begin{align}   \label{eq:genfunc_direct}
    \E_{\bz} \prod_{j=1}^d \exp\left(t_j\bz_{ij}-\frac{1}{2}t_j^2\right) = \sum_{k_1,\dots,k_d\in\bbN} \left(\frac{1}{k_1!\cdots k_d!}\cdot\E_{\bz} \prod_{j=1}^d h_{k_j}(\bz_{ij})\right) t_1^{k_1}\cdots t_d^{k_d}
\end{align}

On the other hand, the LHS of the above expression simplifies to
\begin{align*}
    &\E_{\widetilde{\bg}} \prod_{j=1}^d \exp\left(t_j\left(\frac{x_i}{\sqrt{d}} + \langle v_{j}, \widetilde{\bg_i}\rangle\right)-\frac{1}{2}t_j^2\right) \\
    =& \frac{\exp\left(\frac{x_i}{\sqrt{d}}\sum_{j=1}^d t_j\right)}{\exp\left(\frac{1}{2}\sum_{j=1}^d t_j^2\right)} \E_{\widetilde{\bg}} \exp\left(\widetilde{\bg}_{i1}\left(\sum_{j=1}^d t_j v_{j1}\right) + \cdots + \widetilde{\bg}_{id}\left(\sum_{j=1}^d t_j v_{jd}\right)\right)
\end{align*}
The expectation term can be simplified further as
\begin{align*}
    \E_{\widetilde{\bg}}\exp\left(\widetilde{\bg}_{i\ell}\left(\sum_{j=1}^d t_j v_{j\ell}\right)\right) &= \prod_{\ell=1}^d \exp\left(\frac{1}{2}\left(\sum_{j=1}^d t_j v_{j\ell}\right)^2\right)\\
    &= \exp\left(\frac{1}{2}\sum_{\ell=1}^d \left(\sum_{j=1}^d t_j v_{j\ell}\right)^2 \right)\\
    &= \exp\left(\frac{1}{2}\left(\sum_{j=1}^d t_j^2\|v_j\|^2 + \sum_{k \ne j} t_k t_j \langle v_k, v_j \rangle\right)\right)
\end{align*}
As a result, we know that the LHS of \pref{eq:genfunc_direct} is equal to
\begin{align*}
    &\frac{\exp\left(\frac{x_i}{\sqrt{d}}\sum_{j=1}^d t_j\right)}{\exp\left(\frac{1}{2}\sum_{j=1}^d t_j^2\right)} \exp\left(\frac{1}{2}\left(\sum_{j=1}^d t_j^2\left(\frac{d-1}{d}\right) - \frac{1}{d}\sum_{k \ne j} t_k t_j \right)\right)\\
    =& \exp\left(x_i\frac{\sum_{j=1}^d t_j}{\sqrt{d}} - \frac{1}{2}\left(\frac{\sum_{j=1}^d t_j}{\sqrt{d}}\right)^2\right)
\end{align*}
Applying \pref{fact:hermite-generatingfunc} to the above and using \pref{eq:genfunc_direct}, we get the identity
\begin{align}   \label{eq:genfunc_equality}
    \sum_{k=0}^\infty \frac{1}{k!} h_k(x_i)\left(\frac{\sum_{j=1}^d t_j}{\sqrt{d}}\right)^k = \sum_{k_1,\dots,k_d\in\bbN} \left(\frac{1}{k_1!\cdots k_d!}\cdot\E_{\bz} \prod_{j=1}^d h_{k_j}(\bz_{ij})\right) t_1^{k_1}\cdots t_d^{k_d}
\end{align}
Equating the coefficient of $t_1^{k_1}\cdots t_d^{k_d}$ on both sides of \pref{eq:genfunc_equality}, we get
\[
    \frac{1}{(k_1+\cdots+k_d)!}{\binom{k_1+\cdots+k_d}{k_1, \dots, k_d}} \frac{h_{k_1+\cdots+k_d}(x_i)}{\left(\sqrt{d}\right)^{k_1+\cdots+k_d}} = 
    \frac{1}{k_1!\cdots k_d!}\cdot\E_{\bz} \prod_{j=1}^d h_{k_j}(\bz_{ij})
\]
which can be simplified to the equality
\[
    \E_{\bz} \prod_{j=1}^d h_{k_j}(\bz_{ij}) = \frac{h_{k_1+\cdots+k_d}(x_i)}{\left(\sqrt{d}\right)^{k_1+\cdots+k_d}}
\]
which in turn simplifies \pref{eq:herm-coeff} into
\[
    \E_{\bx} \bx^L \prod_{i=1}^n \frac{h_{|\alpha_i|}(x_i)}{\left(\sqrt{d}\right)^{|\alpha_i|}}
\]
where $|\alpha_i| = \sum_{j=1}^d \alpha_{ij}$. Suppose there is $i$ such that either (i) $i\not\in L$ and $|\alpha_i|$ is odd, or (ii) $i\in L$ and $|\alpha_i|$ is even, then the above expression is equal to 0. Otherwise, it is
\[
    \prod_{i=1}^n \frac{h_{|\alpha_i|}(1)}{\left(\sqrt{d}\right)^{|\alpha_i|}}
\]
The relevant Hermite coefficient is then
\[
    \prod_{(a,b)\in[n]\times[d]}\frac{1}{\alpha_{ab}!}\prod_{i=1}^n\frac{h_{|\alpha_i|}(1)}{d^{|\alpha_i|/2}}
\]

To summarize, we have
\[
    \widehat{\pE_M[x^L](\alpha)} =
    \begin{cases}
        \prod\limits_{(a,b)\in[n]\times[d]}\frac{1}{\alpha_{a,b}!}\prod\limits_{i=1}^n\frac{h_{|\alpha_i|}(1)}{d^{|\alpha_i|/2}} &\text{if $|\alpha_i| = 1_{i\in L}~\mathrm{mod}~2$ and $\sum_{i=1}^n \alpha_{ij}$ is even for all $j$} \\
        0 &\text{otherwise.}
    \end{cases}
\]

\end{document}